\documentclass{article}
\pdfoutput=1
\usepackage[a4paper, margin=3cm, bottom=4cm]{geometry}
\usepackage[utf8]{inputenc}
\usepackage[T1]{fontenc}
\usepackage{lmodern}

\usepackage[
    backend=biber,
    style=alphabetic, maxalphanames=5,
    maxbibnames=100, maxsortnames=100,
    sorting=nyt,
    natbib=true,
    url=false,doi=false,isbn=false,eprint=false
]{biblatex}
\AtEveryBibitem{\clearfield{pages}}

\usepackage{graphicx}
\usepackage{subcaption}
\graphicspath{ {images/} }
\usepackage{enumitem}
\usepackage[normalem]{ulem}
\usepackage[toc,page]{appendix}

\usepackage{tikz}
\usepackage{adjustbox}
\usetikzlibrary{decorations.pathreplacing}
\usetikzlibrary{patterns}
\usepackage[bb=stix]{mathalpha}

\usepackage[bookmarks, bookmarksnumbered,
			colorlinks=true,
            linkcolor = blue,
            urlcolor  = blue,
            citecolor = teal]{hyperref}

\newcommand\fnsurl[1]{{\footnotesize\url{#1}}}

\usepackage{myquicksetup}
\numberwithin{definition}{section}
\numberwithin{theorem}{section}
\numberwithin{corollary}{section}
\numberwithin{proposition}{section}
\numberwithin{lemma}{section}
\numberwithin{claim}{section}
\numberwithin{fact}{section}
\numberwithin{remark}{section}
\numberwithin{example}{section}
\numberwithin{equation}{section}
\mathtoolsset{showonlyrefs}

\usepackage{mylivemacros}
\usepackage{higherorder_otto_GCMC}
\usepackage{upgreek}
\usepackage{centernot}
\usepackage{tensor}
\usepackage{nicefrac}
\usepackage{microtype}

\newif\ifextended
\extendedtrue

\title{A higher-order Otto calculus approach to the Gaussian completely monotone conjecture}
\date{\today}

\author{Guillaume Wang\footnote{EPFL~ \texttt{guillaume.wang@epfl.ch}}}

\hypersetup{
    pdftitle={A higher-order Otto calculus approach to the Gaussian completely monotone conjecture},
    pdfauthor={Guillaume Wang}
}

\begin{document}
\maketitle

\begin{abstract}
    The Gaussian completely monotone (GCM) conjecture states that the $m$-th time-derivative of the entropy along the heat flow on $\mathbb{R}^d$ is positive for $m$ even and negative for $m$ odd.
    We prove the GCM conjecture for orders up to $m=5$, assuming that the initial measure is log-concave, in any dimension.
    Our proof differs significantly from previous approaches to the GCM conjecture: it is based on Otto calculus and on the interpretation of the heat flow as the Wasserstein gradient flow of the entropy.
    Crucial to our methodology is the observation that the convective derivative behaves as a flat connection over probability measures on $\mathbb{R}^d$.
    In particular we prove a form of the univariate Faa di Bruno's formula on the Wasserstein space (despite it being curved), and we compute the higher-order Wasserstein differentials of internal energy functionals (including the entropy), both of which are of independent interest.
\end{abstract}


\section{Introduction} \label{sec:intro}

Let, for some probability measure $\mu_0$ over $\RR^d$, $\mu_t = \mathrm{Law}(X_0 + \sqrt{2 t} G)$ with $G$ a standard Gaussian vector and $X_0 \sim \mu_0$.
The Gaussian completely monotone (GCM) conjecture, proposed in \cite{cheng_higher_2015}, states that the following holds for any $m \geq 1$:
\begin{equation} \label{eq:GCM_d_m} \tag{$\GCM_{d, m}$}
    \forall t > 0,~ (-1)^m \frac{d^m}{dt^m} H(\mu_t) \geq 0,
\end{equation}
where $H(\mu) = \int_{\RR^d} d\mu \log \frac{d\mu}{dx}$ denotes the (negative) differential entropy.

\paragraph{State of the art.}
It is known that \eqref{eq:GCM_d_m} holds true
\begin{itemize}[noitemsep, topsep=4pt]
	\item for any $d$ and $m \leq 2$ \cite{costa_new_1985,villani_short_2000};
	\item for any $d$ and $m \leq 3$ provided that $\mu_0$ is log-concave \cite{toscani_concavity_2015};
    \item for $d=1$ and $m \leq 4$ \cite{cheng_higher_2015};
    \item for $d=1$ and $m \leq 5$ provided that $\mu_0$ is log-concave \cite{zhang_gaussian_2018};
    \item for $d \leq 4$ and $m \leq 3$ \cite{guo_lower_2022};
    \item for $d \leq 2$ and $m \leq 4$ provided that $\mu_0$ is log-concave \cite{guo_lower_2022}.
\end{itemize}
(See \autoref{fig:intro:sota_gcmc} for a graphical representation.)
In fact finer lower bounds are known in some cases, see the brief survey \cite{ledoux_differentials_2020}.

The proof of \cite{toscani_concavity_2015} is based on information-theoretic considerations; the proof of \cite{villani_short_2000} is based on Bakry and Emery's $\Gamma_2$ calculus \cite{bakry_diffusions_1985}.
In all other cited works, the proofs are based on a sum-of-squares methodology which seems difficult to generalize to arbitrary dimensions.
In this work, using a significantly different approach, we prove the following.
\begin{theorem} \label{thm:intro:GCMC_45}
    \eqref{eq:GCM_d_m} holds true for any $d$ and $m \leq 5$ provided that $\mu_0$ is log-concave.
\end{theorem}

We review the history of the GCM conjecture in \autoref{subsec:intro:GCM}. Then in \autoref{subsec:intro:approach} we outline our proof strategy and highlight its key novel features, some of which are of independent interest for the study of probability flows.

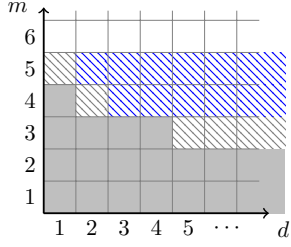
\begin{figure}
  \vspace{-1.2em}
  \begin{minipage}[c]{0.4\textwidth}
	\centering
	\adjustbox{scale=0.85}{%
		\begin{tikzpicture}[x=0.5cm,y=0.5cm]
            \fill[lightgray] (0,0) rectangle (4, 3);
            \fill[lightgray] (0,0) rectangle (1, 4);
            \fill[pattern=north west lines, pattern color=gray] (1, 3) rectangle (2, 4);
            \fill[lightgray] (0,0) rectangle (7.5, 2);
            \fill[pattern=north west lines, pattern color=gray] (4, 2) rectangle (7.5, 3);
            \fill[pattern=north west lines, pattern color=gray] (0, 4) rectangle (1, 5);
            \fill[pattern=north west lines, pattern color=blue] (1, 4) rectangle (7.5, 5);
            \fill[pattern=north west lines, pattern color=blue] (2, 3) rectangle (7.5, 4);
            \draw[step=1,gray,very thin] (0,0) grid (6.7, 6.3);
            \draw[thick,->] (0,0) -- (7,0) node[anchor=north west] {$d$};
            \draw[thick,->] (0,0) -- (0,6.4) node[anchor=east] {$m$~~};
            \foreach \x in {1,2,3,4,5}
            \draw (\x, 0) ++ (-0.5, 0) node[anchor=north] {$\x$};
            \draw (6, 0) ++ (-0.25, -.05) node[anchor=north] {$\cdots$};
            \foreach \y in {1,2,3,4,5,6}
            \draw (0, \y) ++ (0, -0.5) node[anchor=east] {$\y$};
		\end{tikzpicture}
	}
  \end{minipage}
  ~
  \begin{minipage}[c]{0.55\textwidth}
    \caption{State of the GCM conjecture. Filled: values of $(d,m)$ for which \eqref{eq:GCM_d_m} is known to hold true. Crosshatched: \eqref{eq:GCM_d_m} holds true provided that $\mu_0$ is log-concave. Blue crosshatched: our contribution.}
    \label{fig:intro:sota_gcmc}
  \end{minipage}
  \hfill
\end{figure}

\subsection{Gaussian completely monotone conjecture} \label{subsec:intro:GCM}
It was first noticed in \cite{cheng_higher_2015} that \eqref{eq:GCM_d_m} holds for $d=1$ and any $m \leq 4$, inspired by a previous result by Costa \cite{costa_new_1985}. This led the authors to formulate the GCM conjecture in the same paper.
We refer the reader to \cite{cheng_higher_2015,toscani_concavity_2015,ledoux2021log} for an in-depth discussion of the information-theoretic significance of the conjecture, notably in connection with several convexity-type results in information theory.
Section~IX of \cite{cheng_higher_2015} also notes the following two connections to other fields of mathematics.
From an analysis perspective, the GCM conjecture is the assertion that $t \mapsto -\frac{d}{dt} H(\mu_t)$ is a completely monotone function \cite{widder1946laplace}---hence the name, which we borrowed from \cite{cheng2022reformulation}.
From the perspective of mathematical physics, the GCM conjecture is a variant of a conjecture by McKean \cite[Section~13, problem~a)]{mckean1966speed} concerning the signs of the time-derivatives of the entropy along the Boltzmann equation, which was ultimately disproved \cite{lieb1982comment,olaussen1982extension}.
The reasons for its failure do not easily translate to the GCM conjecture, as the underlying dynamics (Boltzmann equation vs.\ heat flow) are quite different.

\paragraph{Related conjectures.}
The proposal of the GCM conjecture prompted interest in some related questions, which we now briefly review following \cite{ledoux_differentials_2020}. Still denoting $\mu_t = \mathrm{Law}(X_0 + \sqrt{2t} G)$,
\begin{enumerate}[label=$(\roman*)$]
	\item The entropy power conjecture asserts that for any $m \geq 1$,
	$(-1)^{m-1} \frac{d^m}{dt^m} e^{-\frac{2}{d} H(\mu_t)} \geq 0$.
	The original result by Costa in \cite{costa_new_1985} is that this inequality holds for any $d$ and $m \leq 2$.
	\item The Gaussian optimality conjecture (or McKean conjecture \cite[Section~12]{mckean1966speed}) asserts that if the covariance of $\mu_0$ is bounded by $\sigma^2 I_d$, then for any $m \geq 1$,
	$(-1)^m \frac{d^m}{dt^m} H(\mu_t) \geq \frac{d 2^{m-1} (m-1)!}{(\sigma^2 + 2t)^m}$.
	One can check that equality holds when $\mu_0$ is the centered Gaussian distribution with covariance $\sigma^2 I_d$.
	In other words, the conjecture is that, under a covariance constraint, the quantity $(-1)^m \frac{d^m}{dt^m} H(\mu_t)$ is minimized at all times for $\mu_0$ being a standard Gaussian.
	\item The GCM conjecture is that for any $ m \geq 1, (-1)^m \frac{d^m}{dt^m} H(\mu_t) \geq 0$.
	\item The log-convexity of Fisher information conjecture asserts that $t \mapsto \log \big( -\frac{d}{dt} H(\mu_t) \big)$ is convex.
	It is proved in \cite{ledoux2021log} for $d=1$, and follows from \cite{toscani_concavity_2015} when $\mu_0$ is log-concave for any~$d$.
\end{enumerate}
Clearly $(ii) \implies (iii)$.
It also holds $(iii) \implies (iv)$, as shown in \cite[Section~IX]{cheng_higher_2015} using a deep result on completely monotone functions.
The implication $(i) \implies (ii)$ is proved in \cite{wang2024implies}.

Other works investigated the time-convexity of other functions of $\mu_t$: R{\'e}nyi entropies in \cite{savare2014concavity}, mutual information in \cite{wibisono2018convexity}, the square-root of Fisher information in \cite{liu2023square}.

\paragraph{Previous approaches.}
Since $\mu_t = \mathrm{Law}(X_0 + \sqrt{2t} G)$ is also the law of $Y_t$ where $dY_t = \sqrt{2} \,dB_t$ and $Y_0 \sim \mu_0$, then by Ito's formula, $\partial_t \mu_t = \Delta \mu_t$.
This PDE, called the heat equation, allows to explicitly express the time-derivatives of $H(\mu_t)$ in terms of spatial~derivatives of~$\mu_t$.
For example, by integration by parts, $\frac{d}{dt} H(\mu_t) = \int_{\RR^d} (\Delta \mu_t) \log \frac{d\mu_t}{dx} = -\int_{\RR^d} \big( \nabla \log \frac{d\mu_t}{dx} \big)^\top \nabla \mu_t = -\int_{\RR^d} \norm{\nabla \log \frac{d\mu_t}{dx}}^2 d\mu_t$ is the negative of the Fisher information,
thus recovering De Bruijn's identity.

Using explicit computations and integrations by parts, \cite{costa_new_1985} and later \cite{cheng_higher_2015} obtain expressions for $\frac{d^m}{dt^m} H(\mu_t)$ for the first few values of $m$, that involve sums of squares of the spatial~derivatives of $\mu_t$ or of $\log \frac{d\mu_t}{dx}$.
\cite{ledoux2021log} use a similar method for the log-convexity of Fisher information conjecture.
This method is systematized in \cite{zhang_gaussian_2018,guo_lower_2022,liu2023square} using semi-definite programming techniques to identify the sum-of-squares decompositions.

The study of the heat equation lies at the intersection of a number of mathematical fields, so it comes with a variety of interpretations. In the context of the GCM conjecture, one point of view that at first appears promising is that of $\Gamma_2$ calculus, an abstract framework for analyzing diffusion Markov processes \cite{bakry_diffusions_1985,bakry2014analysis}.
Indeed, the contribution of \cite{villani_short_2000} is that this framework allows for a much shorter proof of Costa's result.
However, it is unclear whether this approach can be extended for higher orders $m>2$, as discussed by \cite[end of Section~1]{ledoux_differentials_2020}.

Our approach is based on a different interpretation of the heat equation, namely, as the Wasserstein gradient flow of the entropy.
This point of view has not been considered before for the GCM conjecture, although ideas from optimal transport and Otto calculus have been successfully applied to related questions, in \cite{rioul2017yet} for Shannon's entropy power inequality and in \cite[Section~2.1]{tamanini2022generalization} for a generalization of Costa's result
(both correspond to order $m=2$).

\subsection{Higher-order Otto calculus approach} \label{subsec:intro:approach}
Let us explain our approach.
Similar to prior works, we start from the fact that the evolution of $\mu_t$ is described by the heat flow $\partial_t \mu_t = \Delta \mu_t$.
The seminal work of Jordan, Kinderlehrer and Otto \cite{jordan_variational_1998} showed that this PDE can be interpreted as the gradient flow of $H$ in the Wasserstein geometry.
This interpretation was the starting point of a series of works by Otto which developed the idea that the structure of the Wasserstein space is formally%
\footnote{
    Throughout this paper, ``formal'' means ``described by systematic calculations but without a rigorous justification''.
}
analogous to that of a Riemannian manifold, albeit infinite-dimensional \cite{otto_generalization_2000,otto_geometry_2001,otto_eulerian_2005}.
This non-rigorous analogy is by now well-known in the field of PDEs, in optimal transport and in some sub-fields of machine learning, and is called the Otto calculus.
We provide extensive background on it in \autoref{sec:apx_bg_otto}, assuming some basic familiarity with Riemannian geometry.

From the point of view of Otto calculus, the GCM conjecture can be understood as a statement on the time-derivatives of $H$ along its own Wasserstein gradient flow. A possible strategy---which is almost what we will do---is thus to proceed by analogy: 1) work out the formulas for the time-derivatives of a function $f: \MMM \to \RR$ over a Riemannian manifold $\MMM$ along its own Riemannian gradient flow; 2) show that the time-derivatives of $H$ along its own Wasserstein gradient flow obey the same formulas; 3) try to show that the obtained expressions have a sign.
However, the first of these steps turns out to be very tedious, due to the appearance of Riemann curvature terms at orders~$4$ and higher~(Eq.~\eqref{eq:fdb_wasserstein:findim:curved:d4f_ct}).
Instead, we show that it suffices to work out the formulas for the case where the manifold $\MMM$ is flat, thanks to a univariate Faa di Bruno's formula on the Wasserstein space over $\RR^d$.

The expressions of $\frac{d^m}{dt^m} H(\mu_t)$ thus obtained consist of some terms which clearly have a sign, and of others which require further analysis.
For the latter ones, we express them as integrals over $\RR^d$, and we perform integrations by parts, until an integrand is obtained whose sign is constant.
These derivations can be intricate, although no more so than the sum-of-squares approach used in \cite{cheng_higher_2015}.
Perhaps interestingly, we find that up to order $m=5$, each individual term is non-negative when $\mu_0$ is log-concave, i.e., there is no need to consider compensations across terms (\autoref{rk:main_pf:allnonneg}).

\paragraph{Convective calculus.}
The Otto calculus as such appears to be difficult to generalize to orders higher than~$2$, due to the curvature intrinsically induced by the Wasserstein distance \cite[Section~7.3]{ambrosio_gradient_2008} \cite[Section~6]{gigli_second_2012}.
Prior to our work, we are not aware of any attempts to do so.
The structure of $\PPP(\RR^d)$ leveraged by our computations, leading to the univariate Faa di Bruno formula, 
is actually that of Eulerian calculus, in the sense of \cite[Chapter~14]{villani_optimal_2009}, rather than the full Otto calculus.
The distinction is explained in detail in \autoref{subsec:fdb_wasserstein:discussion} and \autoref{sec:apx_connection}.
Briefly, in technical terms, we make use of the convective derivative---in the sense of fluid mechanics---as an affine connection on $\PPP(\RR^d)$, instead of (the Otto calculus analog of) the Levi-Civita connection \cite[Definition~5.1]{gigli_second_2012}.
The former connection turns out to be flat in a certain sense, contrary to the latter one, and this allows for simpler computations.

More generally, we show that for $(\XXX, \nabla)$ a smooth manifold equipped with a connection,
the convective derivative formally endows $\PPP(\XXX)$ with a (non-Riemannian) differential manifold structure which inherits some properties of $\XXX$ itself, such as being torsion-free or having zero curvature (\autoref{sec:apx_connection}).
In this sense, convective calculus lifts the differential structure of $\XXX$ onto one for $\PPP(\XXX)$, in the same way that for a metric space $\MMM$, the Wasserstein distance lifts metric and topological (and ``Riemannian'' according to Otto) properties of $\MMM$ onto $\PPP(\MMM)$ \cite[Chapters~7-9]{ambrosio_lectures_2021}.

Let us point out that the role played by the convective derivative of vector fields in Otto calculus is discussed explicitly in \cite[Sections~3,4]{gigli_second_2012}
and in \cite[Section~3.3]{gentil_dynamical_2020}.%
\footnote{In~\cite{gigli_second_2012,gentil_dynamical_2020}, only velocity fields which are gradient field are considered (in the sense of \autoref{def:fdb_wasserstein:convgeod:transportcouple}), which can cause some conceptual unclarity. Separately, the fact that formal computations can be carried out for general velocity fields is remarked, e.g., in~\cite{schachter_new_2018}.}
Furthermore, the geodesic curves of $\PPP(\RR^d)$ equipped with the convective derivative as a connection almost coincide with the ``acceleration-free curves'' introduced by \cite{parker_convexity_2024} (\autoref{prop:fdb_wasserstein:convgeod:charact_convgeod}).

\paragraph{Relation to the Lie algebra of iterated gradients for the Laplacian \cite{ledoux_algebre_1995}.}
Our approach yields one particular way of expressing the higher-order time-derivatives of the entropy along the heat flow, for any order and in any dimension (\autoref{prop:gf:wass:expr}).
Another way is given by \cite[Section~3]{ledoux_algebre_1995},
which can be seen as an extension of Bakry and Emery's $\Gamma_2$ calculus to higher orders, and was used in some works indirectly related to the GCM conjecture \cite{ledoux_heat_2016,mansanarez_derivatives_2024}.
These two ways appear to be fundamentally different, as discussed in \autoref{subsec:ent:ptwise}.
At a high level, our computations are based on using convective geodesics (straight-line mass transports) as ``reference'' probability flows, whereas $\Gamma_2$ calculus and its higher-order extension by Ledoux are tailored to analyzing Markov diffusions.

\paragraph{Omitted regularity issues.}
For the results concerning the heat flow, including \autoref{thm:intro:GCMC_45}, 
one can show (with some effort) by adapting the proofs of \cite{costa_new_1985} and \cite[Proposition~2]{cheng_higher_2015} that the integrations by parts and the time-derivative/integration interchanges of \autoref{subsec:gf:wass} and \autoref{sec:main_pf} are justified for all $t>0$.
All other results are formal in that we did not investigate the regularity conditions under which the formulas holds rigorously.

\subsection{Contributions}
\begin{itemize}
    \item We show that the GCM conjecture holds for orders up to $m=5$, assuming that the initial measure is log-concave, in any dimension (\autoref{thm:intro:GCMC_45}).

    Our approach differs significantly from other works on the GCM conjecture. It would be interesting to see how far it could be pushed, be it to tackle higher orders, or the case of non-log-concave initial measures, or the related conjectures reviewed in~\autoref{subsec:intro:GCM}.
    \item We show that a form of the univariate Faa di Bruno's formula holds on the Wasserstein space over $\RR^d$ (\autoref{thm:fdb_wasserstein:fdb_wasserstein}), by reasoning in terms of convective derivatives.
    
    This result illustrates the usefulness of the convective derivative for formal computations.
    In contrast, higher-order calculus using the full Otto calculus as presented in reference texts \cite{villani_optimal_2009,gigli_second_2012} would be quite cumbersome, because the Wasserstein space has non-zero curvature, and the univariate Faa di Bruno's formula does not hold on curved Riemannian manifolds (\autoref{subsec:fdb_wasserstein:findim}).
    
    The formal analogy used in this work is thus subtly different from the usual Otto calculus, which it recovers at orders~$1$ and~$2$. 
    This ``convective calculus'' offers an alternative formal framework for the analysis of probability flows
	(Secs.~\ref{subsec:fdb_wasserstein:convgeod}~to~\ref{subsec:fdb_wasserstein:convderiv}, \autoref{sec:apx_connection}),
    so it may be of independent interest in the context of particle methods for optimization, or diffusion models.
    \item We compute the higher-order Wasserstein differentials of internal energy functionals, including the entropy (\autoref{thm:ent:difftls_EEE}).

    For the second-order differentials, we recover the classical expression for the Wasserstein Hessian, with a simpler proof \cite[Formula~15.7]{villani_optimal_2009}.
    For orders~$3$ and higher, our computation does not seem to generalize easily to probability measures over Riemannian manifolds instead of $\RR^d$ (\autoref{rk:ent:hess:riem_hard}).
\end{itemize}

\subsection{Notation}
We denote by $\PPP_2(\RR^d)$ the set of probability measures over $\RR^d$ with finite second moments.
For a functional $\FFF: \PPP_2(\RR^d) \to \RR$, we write $\FFF'[\mu]: \RR^d \to \RR$ for its first variation at $\mu$, defined (up to an arbitrary additive constant) by $\int_{\RR^d} \FFF'[\mu](x) d(\nu - \mu)(x) = \lim_{\eps \to 0} \frac{1}{\eps} \left[ \FFF(\mu + \eps (\nu - \mu)) - \FFF(\mu) \right]$ for any $\nu \in \PPP(\RR^d)$.
We also write $\FFF''[\mu]: \RR^d \times \RR^d \to \RR$ for its second variation and $\FFF^{(k)}[\mu]$ for its $k$-th variation.
``$\nabla \cdot$'' denotes divergence of vector-valued measures, i.e., $\int_{\RR^d} \varphi(x) \nabla \cdot \bm\mu(dx) = -\int_{\RR^d} \nabla \varphi(x) \cdot \bm\mu(dx)$ for $\varphi$ or $\bm\mu$ vanishing at infinity.
For any measure $\mu$ and mapping $T$, $T_\sharp \mu$ denotes the pushforward measure, characterized by $\int \varphi d(T_\sharp \mu) = \int (\varphi \circ T) d\mu$ for any function $\varphi$.
$\norm{\cdot}$ denotes the Euclidean norm.
$\NN^*$ is the set of positive integers and $\NN = \NN^* \cup \{0\}$.
Integrals $\int$ are taken over $\RR^d$ unless otherwise specified.

We denote by $\mathfrak{X}(\RR^d)$ the set of smooth vector fields on $\RR^d$.
For $\mu \in \PPP_2(\RR^d)$, $L^2_\mu$ denotes the Hilbert space of vector fields $\Phi: \RR^d \to \RR^d$ such that $\int_{\RR^d} \norm{\Phi}^2 d\mu < \infty$.
(We do not introduce notation for the analogous set of scalar fields.)
If $g_t$ is a time-dependent tensor field, we will write indifferently $\frac{d}{dt} g_t$ or $\partial_t g_t$ for its pointwise time-derivative.

Einstein summation notation will be used freely, that is, for any two expressions $\mathsf{expr1}, \mathsf{expr2}$, one should interpret $(\mathsf{expr1})_i\, (\mathsf{expr2})^i$ as $\sum_{i=1}^d (\mathsf{expr1})_i\, (\mathsf{expr2})^i$, and likewise when several indices $i, j, k, ...$ are repeated, each once in superscript and once in subscript.
In \autoref{sec:fdb_wasserstein} and \autoref{sec:ent}, we respect the convention that primal objects such as vector fields will have index in superscript, e.g.\ $\Phi^i$,
and dual objects such as differential forms will have index in subscript, e.g.\ $\nabla_i f$.
In \autoref{sec:gf} and \autoref{sec:main_pf}, because we consider gradient flows and because differential forms and gradients on $\RR^d$ coincide, the position of the indices (sub- or superscript) holds no meaning, and we may write $(\mathsf{expr1})_i\, (\mathsf{expr2})_i$ for $\sum_{i=1}^d (\mathsf{expr1})_i\, (\mathsf{expr2})_i$ and likewise for repeated indices in superscript.
In particular, $\nabla^2_{ii} h = \nabla_i \nabla_i h = \Delta h$ refers to the Laplacian of $h: \RR^d \to \RR$.

If $g$ is a $(p, q)$-tensor field on $\RR^d$, we denote by $\nabla g$ the $(p, q+1)$-tensor field with $\nabla_i g^{\bm{j}}_{~\bm{k}} = \frac{\partial}{\partial x_i} g^{\bm{j}}_{~\bm{k}}$ for all $i \in \{1,...,d\}, \bm{j} \in \{1,...,d\}^p, \bm{k} \in \{1,...,d\}^q$.
Moreover if $\Phi$ is a vector field, we will write $\Phi \cdot \nabla g$ for the $(p, q)$-tensor field $(\Phi \cdot \nabla g)^{\bm{j}}_{~\bm{k}} = \Phi^i \nabla_i g^{\bm{j}}_{~\bm{k}}$.
In particular, if $g$ is a vector field then $\Phi \cdot \nabla g$ can interpreted at each point as the product of a row-vector with a matrix $\nabla g$ which is the \emph{transpose} of the Jacobian of $g$.
For tensors $A, B$ of the same type, we write $A : B \in \RR$ for the tensor contraction over all indices.

\paragraph{Organization.}
In \autoref{sec:fdb_wasserstein} we introduce a convective calculus framework for formal higher-order computations over the Wasserstein space $\PPP_2(\RR^d)$.
In \autoref{sec:ent} we derive a formula for the higher-order Wasserstein differentials of internal energy functionals, including the entropy $H$.
These developments are applied to the GCM conjecture in the last two sections: 
in \autoref{sec:gf} we express the time-derivatives of $H$ along the heat flow in terms of its higher-order differentials, thanks to an analogy with the case of finite-dimensional gradient flows,
and in \autoref{sec:main_pf} we use these expressions to show \autoref{thm:intro:GCMC_45}.

\section{A Faa di Bruno's formula on the Wasserstein space over~\texorpdfstring{$\RR^d$}{Rd}} \label{sec:fdb_wasserstein}

In this section, we show that the higher-order time derivatives of a functional $\FFF: \PPP_2(\RR^d) \to \RR$ along an absolutely continuous curve of measures $(\mu_t)_t$, in the Wasserstein sense, can be computed using the same formula as for a function $f: \RR^d \to \RR$ along a curve $(x_t)_t$ in $\RR^d$.

\subsection{The higher-order univariate chain rule in finite dimension} \label{subsec:fdb_wasserstein:findim}

Since our goal is to compute higher-order time-derivatives on the Wasserstein space, and we know by Otto calculus that it behaves similarly to a Riemannian manifold,
in this subsection we first discuss how to compute higher-order time-derivatives on a Riemannian manifold $(\MMM, g)$.
In a nutshell, for flat manifolds such as $\RR^d$ or the torus, the answer is classical and given by the univariate Faa di Bruno's formula, while for curved manifolds the same formula cannot hold.

\subsubsection{Flat case: the Bell polynomials and the univariate Faa di Bruno's formula} \label{subsubsec:fdb_wasserstein:findim:flat}

In this subsection, we consider the case where $(\MMM, g) = (\RR^d, \norm{\cdot})$.
Let $f: \RR^d \to \RR$ and $(x_t)_{t \in \RR}$ any curve in $\RR^d$.
It is classical that, by the Faa di Bruno's formula on $(f \circ x)(t)$,
\begin{equation} \label{eq:fdb_wasserstein:findim:flat:FdB}
    \frac{d^n}{dt^n} f(x_t)
    = \sum_{k=1}^n (\nabla^k f) : B_{n,k}\left( \frac{dx}{dt}, \frac{d^2 x}{dt^3}, ..., \frac{d^{n-k+1} x}{dt^{n-k+1}} \right)
\end{equation}
where 
$\nabla^k f$ refers to the tensor $\nabla_{i_1} ... \nabla_{i_k} f(x_t)$ and $B_{n,k}(X_1, ..., X_{n-k+1})$ is the partial exponential Bell polynomial of order $(n,k)$,
with the understanding that multiplication refers to tensor product of vectors in $\RR^d$.
For example, dropping the dependency on the indeterminates $(X_1, X_2, ...)$ for brevity,
\vspace{-0.75em}
\begin{align*}
    B_{1,1} &= X_1, \\
    B_{2,1} &= X_2,~~~~ B_{2,2} = X_1^2, \\
    B_{3,1} &= X_3,~~~~ B_{3,2} = 3 X_1 X_2,~~~~ \qquad~~~ B_{3,3} = X_1^3, \\
    B_{4,1} &= X_4,~~~~ B_{4,2} = 3 X_2^2 + 4 X_1 X_3,~~~~ B_{4,3} = 6 X_1^2 X_2,~~~~ B_{4,4} = X_1^4,
\end{align*}
and correspondingly, writing $\dot{x} = \frac{dx}{dt}$, $\ddot{x} = \frac{d^2x}{dt}$, etc.,
\vspace{-0.25em}
\begin{align*}
	\text{\scalebox{0.7}{$\frac{d}{dt}$}} f(x_t) &= (\nabla f) : \dot{x} \\
	\text{\scalebox{0.7}{$\frac{d^2}{dt^2}$}} f(x_t) &= (\nabla f) : \ddot{x} + (\nabla^2 f) : \dot{x}^{\otimes 2} \\
	\text{\scalebox{0.7}{$\frac{d^3}{dt^3}$}} f(x_t) &= (\nabla f) : \dddot{x} + 3 (\nabla^2 f) : \left[ \dot{x} \otimes \ddot{x} \right] + (\nabla^3 f) : \dot{x}^{\otimes 3} \\
	\text{\scalebox{0.7}{$\frac{d^4}{dt^4}$}} f(x_t) &=
	(\nabla f) : \ddddot{x}
	+ (\nabla^2 f) : \left[ 3 \ddot{x}^{\otimes 2} + 4 \dot{x} \otimes \dddot{x} \right] 
	+ 6 (\nabla^3 f) : \left[ \dot{x}^{\otimes 2} \otimes \ddot{x} \right]
	+ (\nabla^4 f) : \dot{x}^{\otimes 4}.
\end{align*}

In \autoref{sec:apx_bell}, we prove some facts about the Bell polynomials, including an abstract version of the univariate Faa di Bruno's formula, which will be useful for some proofs for the next subsections.

\subsubsection{Curved case: asymmetry of the third-order differential tensor} \label{subsubsec:fdb_wasserstein:findim:curved}

We recall the following facts from \cite[Chapter~4]{lee_introduction_2018}, where the tensor fields are over a Riemannian manifold $\MMM$, whose Levi-Civita connection is denoted by $\nabla$.
\begin{proposition}[{\cite[Propositions~4.15, 4.17]{lee_introduction_2018}}] \label{prop:fdb_wasserstein:findim:curved:def_totalcovariantderiv}
    For a $(p, q)$-tensor field $F$, for differential forms $\omega^{(1)}, ..., \omega^{(p)}$ and vector fields $Y_{(1)}, ..., Y_{(q)}, X$, we have
    \begin{multline*}
        (\nabla_X F)(\omega^{(1)}, ..., \omega^{(p)}, Y_{(1)}, ..., Y_{(q)}) 
        = \nabla_X \left[ F(\omega^{(1)}, ..., \omega^{(p)}, Y_{(1)}, ..., Y_{(q)}) \right] \\
        - \sum_{i=1}^p F(\omega^{(1)}, ..., \nabla_X \omega^{(i)}, ..., \omega^{(p)}, Y_{(1)}, ..., Y_{(q)})
        - \sum_{j=1}^q F(\omega^{(1)}, ..., \omega^{(p)}, Y_{(1)}, ..., \nabla_X Y_{(j)}, ... Y_{(q)}).
    \end{multline*}
    This expresses a Leibniz product rule: indeed, formally, if we let $G = \omega^{(1)} \otimes ... \otimes \omega^{(p)} \otimes Y_{(1)} \otimes ... \otimes Y_{(q)}$
    and if we define $\nabla_X G$ by the outer product rule $\nabla_X \left( h \otimes K \right) = (\nabla_X h) \otimes K + h \otimes (\nabla_X K)$ for all tensor fields $h, K$, 
    then the above equality rewrites $(\nabla_X F) \cdot G = \nabla_X \left( F \cdot G \right) - F \cdot (\nabla_X G)$.
    
    The total covariant derivative of $F$ is defined as the $(p, q+1)$-tensor field $\nabla F$ given by
    \begin{equation*}
        (\nabla F)(\omega^{(1)}, ..., \omega^{(p)}, Y_{(1)}, ..., Y_{(q)}, X) = (\nabla_X F)(\omega^{(1)}, ..., \omega^{(p)}, Y_{(1)}, ..., Y_{(q)}).
    \end{equation*}
    For $m \geq 1$, we define the $(p, q+m)$-tensor field $\nabla^m F$ as the total covariant derivative of $\nabla^{m-1} F$.
\end{proposition}

The following differential geometry fact was noted, e.g., in \cite[Remark~3.2]{criscitiello2023accelerated}.
\begin{proposition} \label{prop:fdb_wasserstein:findim:curved:asym}
    Consider a function $f: \MMM \to \RR$, to be interpreted as a scalar field over $\MMM$.
    The tensor field $\nabla^3 f$ is in general asymmetric; more precisely, for any vector fields $U, V, W$,
    \begin{equation*}
        (\nabla^3 f)(U, V, W) = (\nabla^3 f)(V, U, W)
        = (\nabla^3 f)(U, W, V) + (\nabla f) \left( R(V, W) U \right)
    \end{equation*}
    where $R$ denotes the Riemann curvature tensor: $R(V, W) U = \nabla_V \nabla_W U - \nabla_W \nabla_V U - \nabla_{[V,W]} U$.
\end{proposition}

Thus the third-order differential tensors are asymmetric as soon as $\MMM$ has non-zero curvature (unlike the second-order differentials, i.e., the Hessians, which are always symmetric).
Because of this, the univariate Faa di Bruno's formula does not hold, and it seems difficult to formulate a generalization of it.
To illustrate this, consider a curve $c: \RR \to \MMM$ and a function $f: \MMM \to \RR$, and let us derive explicitly the expression of $\frac{d^n}{dt^n} f(c(t))$ for $n \leq 4$.
Denote $c'$ for $c'(t)$ the velocity vector at time~$t$, $\nabla^k f$ for the tensor $\restr{\nabla^k f}{c(t)}$, and let $c'' = \nabla_{c'} c'$, $c''' = \nabla_{c'} c''$, etc., the iterated accelerations.
Then
\begin{align}
    \frac{d}{dt} f(c(t)) &= (\nabla f)(c') \\
    \frac{d^2}{dt^2} f(c(t)) &= (\nabla f)(c'') + (\nabla^2 f)(c', c') \\
    \frac{d^3}{dt^3} f(c(t)) &= (\nabla f)(c''') + 3 (\nabla^2 f)(c'', c') + (\nabla^3 f)(c', c', c') 
    ~~~~\text{since $\nabla^2 f$ is symmetric} \\
    \frac{d^4}{dt^4} f(c(t)) &= (\nabla f)(c'''') + (\nabla^2 f)(c''', c')
    + 3 \left[ (\nabla^2 f)(c''', c') + (\nabla^2 f)(c'', c'') + (\nabla^3 f)(c'', c', c') \right]    \nonumber \\
    &~~~~ + 2 (\nabla^3 f)(c'', c', c') + (\nabla^3 f)(c', c', c'') + (\nabla^4 f)(c', c', c', c') 
    ~~~~\text{since $\nabla^3_{ijk} f = \nabla^3_{jik} f$}    \nonumber \\
    &= (\nabla f)(c'''') + 3 (\nabla^2 f)(c'', c'')
    + 4 (\nabla^2 f)(c''', c')    \nonumber \\
    &~~~~ + 5 (\nabla^3 f)(c'', c', c') + (\nabla^3 f)(c', c', c'') 
    + (\nabla^4 f)(c', c', c', c').
\label{eq:fdb_wasserstein:findim:curved:d4f_ct}
\end{align}
These expressions resemble the ones for the Euclidean case, except for the terms $5 (\nabla^3 f)(c'', c', c') + (\nabla^3 f)(c', c', c'')$
on the last line, 
to be compared with $6 (\nabla^3 f) : \left[ \dot{x}^{\otimes 2} \otimes \ddot{x} \right]$ in the Euclidean case.
The asymmetry of the third-order differential tensor will also propagate to higher orders, so the expressions of $\frac{d^n}{dt^n} f(c(t))$ for $n \geq 5$ would deviate even more from the univariate Faa di Bruno's formula.

\subsection{Convective geodesics} \label{subsec:fdb_wasserstein:convgeod}

Classically, the Wasserstein Hessian of a functional $\FFF: \PPP_2(\RR^d) \to \RR$ at $\mu$ is defined%
\footnote{\cite[Section~3.3]{gentil_dynamical_2020} confirms that one would obtain the same object by taking the (Otto calculus analog of) Levi-Civita covariant derivative of the Wasserstein gradient $\nabla \FFF'[\mu_t]$.}
as the symmetric bilinear form arising from the second-order time-derivative of $\FFF(\mu_t)$ when $(\mu_t)_t$ is a Wasserstein geodesic \cite[Chapter~15]{villani_optimal_2009} (see \autoref{subsec:apx_bg_otto:wass_geodesic} for background on the Wasserstein geodesic equation):
\begin{equation} \label{eq:fdb_wasserstein:convgeod:wassgeod}
    \frac{d^2}{dt^2} \FFF(\mu_t)
    \eqqcolon \Hess_{\mu_t} \FFF(\phi_t, \phi_t)
    ~~~~\text{where}~~~~
    \begin{cases}
        \partial_t \mu_t = -\nabla \cdot (\mu_t \nabla \phi_t) \\
        \partial_t \phi_t = -\frac{1}{2} \norm{\nabla \phi_t}^2.
    \end{cases}
\end{equation}
For reasons that will become clear in \autoref{subsec:fdb_wasserstein:discussion}, it will be advantageous in this paper to work with a slight extension of this definition.

The following definition is a minor modification of \cite[Definition~2.5]{gigli_second_2012} where we additionally ask for smoothness of the velocity field.
\begin{definition} \label{def:fdb_wasserstein:convgeod:transportcouple}
    We call \emph{transport couple} a map $I \ni t \mapsto (\mu_t, \Phi_t)$ with $I$ an interval of $\RR$, 
    $\mu_t \in \PPP_2(\RR^d)$, 
    and $\Phi_t \in L^2_{\mu_t} \cap \mathfrak{X}(\RR^d)$, 
    such that $\int_K \int \norm{\Phi_t}^2 d\mu_t dt < \infty$ for any compact $K \subset I$ 
    and the PDE
    \begin{equation*}
        \partial_t \mu_t = -\nabla \cdot (\mu_t \Phi_t)
        ~~\text{for}~~ t \in I,
    \end{equation*}
    known as the continuity equation, holds in the sense of distributions.
    In particular $(\mu_t)_t$ is an absolutely continuous curve in the Wasserstein sense \cite[Theorem~2.29]{ambrosio_users_2013}.
    
    If $(\mu_t, \Phi_t)_t$ is a transport couple, we will say that the time-dependent vector field $\Phi_t$ is an admissible velocity field for $(\mu_t)_t$.
    Conversely, we will refer to any time-dependent smooth vector field $\Phi_t$ as a \emph{velocity field}, without reference to any probablity measures.
\end{definition}

The following definition is a special case of the notion of acceleration-free curves from \cite{parker_convexity_2024}.
\begin{definition}
    A \emph{convective geodesic} is a transport couple $(\mu_t, \Phi_t)_t$ such that $\partial_t \Phi_t + \Phi_t \cdot \nabla \Phi_t = 0$.
    We will also refer to the curve $(\mu_t)_t$ itself as a convective geodesic.
\end{definition}

The following proposition, which is not essential for the remainder of the paper, gives some intuition on the notion of convective geodesic. 
Namely, the density of a cloud of points is a convective geodesic if and only if each of the points follows a geodesic, with potentially arbitrary directions and speeds (except for a technical injectivity condition). This shows in particular the (almost) equivalence with the notion of acceleration-free curves from \cite{parker_convexity_2024}.
So this is a very permissive notion of geodesic in probability space, yet it is all that is needed for formal computations of arbitrary-order derivatives, as the sequel will show.

\begin{proposition} \label{prop:fdb_wasserstein:convgeod:charact_convgeod}
    Let any smooth vector field $v \in \mathfrak{X}(\RR^d)$ and let $T_t(x) = x+tv(x)$.
    Assume $T_t$ is injective for each $t$. Then the curve $\mu_t = (T_t)_\sharp \mu_0$ is a convective geodesic.%
    \footnote{The assumption that the points do not overlap at any given time (injectivity of $T_t$) is necessary, since the velocity field would be ill-defined at overlap times.
    The proper relaxation of the notion of convective geodesics, in this respect, is that of acceleration-free curves \cite{parker_convexity_2024}.
    Intuitively, they can be characterized by the fact that $\frac{d^2}{dt^2} \int f d\mu_t$ does not depend explicitly on $\nabla f$ but only on $\nabla^2 f$, for any smooth scalar field $f$.}
    
    Conversely, let $\partial_t \mu_t = -\nabla \cdot (\mu_t \Phi_t)$ be a convective geodesic.
    For any $x_0 \in \support(\mu_0)$, let $(x_t)_t$ be the solution of the ODE $\frac{d}{dt} x_t = \Phi_t(x_t)$.
    Then $(x_t)_t$ is a geodesic in $\RR^d$.
\end{proposition}

\begin{proof}
    Let $f: \RR^d \to \RR$ a scalar field.
    Then
    \begin{equation*}
        \frac{d}{dt} \int f d\mu_t
        = \frac{d}{dt} \int f(T_t(x)) d\mu_0(x)
        = \int \nabla f(T_t(x)) \cdot \frac{d}{dt} T_t(x) d\mu_0(x)
        = \int \nabla f(x) \cdot v(T_t^{-1}(x)) d\mu_t(x).
    \end{equation*}
    So an admissible velocity field for $(\mu_t)_t$ is $\Phi_t(x) = v(T_t^{-1}(x))$. 
    Now fix $x$ and let $y_t = T_t^{-1}(x)$, i.e., $y_t + t v(y_t) = x$. 
    Then $\frac{dy_t}{dt} + v(y_t) + t \nabla v(y_t) \cdot \frac{dy_t}{dt} = 0$, i.e., $\frac{dy_t}{dt} = -\left[ I + t \nabla v(y_t) \right]^{-1} v(y_t)$, so
    \begin{align*}
        \frac{d}{dt} v(y_t) &= \nabla v(y_t) \cdot \frac{dy_t}{dt}
        = -\nabla v(y_t) \cdot \left[ I + t \nabla v(y_t) \right]^{-1} v(y_t) \\
        \frac{\partial}{\partial x} v(T_t^{-1}(x)) \cdot v(T_t^{-1}(x))
        &= \nabla v(y_t) \cdot \left[ \nabla T_t(x) \right]^{-1} \cdot v(y_t)
        = \nabla v(y_t) \cdot \left[ I + t \nabla v(T_t^{-1}(x)) \right]^{-1} \cdot v(y_t).
    \end{align*}
    So $\left[ \partial_t \Phi_t + \Phi_t \cdot \nabla \Phi_t \right](x) = \frac{d}{dt} v(T_t^{-1}(x)) + \frac{\partial}{\partial x} v(T_t^{-1}(x)) \cdot v(T_t^{-1}(x)) = 0$ as claimed.

    For the second part of the proposition, it suffices to show that 
    $\frac{d^2}{dt^2}x_t = 0$. Now by definition, $\frac{d^2}{dt^2} x_t = \frac{d}{dt} \left( \Phi_t(x_t) \right) = \partial_t \Phi_t(x_t) + \frac{dx_t}{dt} \cdot \nabla \Phi_t(x_t) = \left[ \partial_t \Phi_t + \Phi_t \cdot \nabla \Phi_t \right](x_t) = 0$.
\end{proof}

\subsection{Higher-order Wasserstein differentials} \label{subsec:fdb_wasserstein:wassdiffs}

Wasserstein differentials are defined in this paper as the symmetric forms arising from the time-derivatives of functionals along convective geodesics, as we record below.
\begin{definition} \label{def:fdb_wasserstein:wassdiffs:wassdiffs}
	The $n$-th order Wasserstein differential of a functional $\FFF: \PPP_2(\RR^d) \to \RR$ at $\mu$ is the symmetric $n$-linear form $D^n_\mu \FFF(\Phi_{(1)}, ..., \Phi_{(n)})$ over vector fields $\Phi_{(i)} \in L^2_\mu \cap \mathfrak{X}(\RR^d)$,
	such that
	\begin{equation*}
		\frac{d^n}{dt^n} \FFF(\mu_t) = D^n_\mu \FFF(\Phi_t, ..., \Phi_t)
		~~~~\text{for any $(\mu_t, \Phi_t)_t$ such that}~~~~
		\begin{cases}
			\partial_t \mu_t = -\nabla \cdot (\mu_t \Phi_t) \\
			\partial_t \Phi_t + \Phi_t \cdot \nabla \Phi_t = 0.
		\end{cases}
	\end{equation*}
	We will also write $D^n_\mu \FFF(\Phi)$ for $D^n_\mu \FFF(\Phi, ..., \Phi)$.
\end{definition}

\begin{remark}
	One can check by straightforward computations that, 
	for any time-dependent scalar field $(\phi_t)_t$ on $\RR^d$, denoting $\Phi_t = \nabla \phi_t$, we have~\cite[Remark~3.11]{gigli_second_2012}
	\begin{equation*}
		\partial_t \phi_t = -\frac{1}{2} \norm{\nabla \phi_t}^2
		~~\implies~~
		\partial_t \Phi_t = -\Phi_t \cdot \nabla \Phi_t.
	\end{equation*}
	In particular, if $(\mu_t, \phi_t)_t$ satisfies the Wasserstein geodesic equation, then $(\mu_t, \Phi_t)_t$ is a convective geodesic.
	As a consequence, the second-order Wasserstein differential of a functional in the sense of the definition above is a proper extension of the Wasserstein Hessian in the sense of \eqref{eq:fdb_wasserstein:convgeod:wassgeod}.
\end{remark}

The following proposition gives the general formula for the Wasserstein differentials of $\FFF$ in terms of its $k$-th variations.
The formula for the second-order differentials (or rather, the Wasserstein Hessians) also appeared in this form in \cite[Definition~3.1]{chow2019partial} and \cite[Proposition~19]{li_transport_2022}.

\begin{proposition} \label{prop:fdb_wasserstein:wassdiffs:expr_firstvars}
	Let a functional $\FFF: \PPP_2(\RR^d) \to \RR$, $\mu \in \PPP_2(\RR^d)$, and $\Phi \in L^2_\mu \cap \mathfrak{X}(\RR^d)$.
	Introduce the following shorthands:
	\begin{itemize}
		\item Denote by $\FFF^{(k)}[\mu]: (\RR^d)^k \to \RR$ the $k$-th variation in measure space.
		\item For any measures $\mu_1, \mu_2, ...$ over $\RR^d$, denote by $B_{n,k}(\mu_1, ..., \mu_{n-k+1})$ the measure over $(\RR^d)^k$ defined by interpreting the Bell polynomial $B_{n,k}$ as a polynomial over the ring of measures equipped with the tensor product. For example, $B_{4,2}(X_1, X_2, X_3) = 3 X_2^2 + 4 X_1 X_3$ and $B_{4,2}(\mu_1, \mu_2, \mu_3)(dx_1, dx_2) = 3 \mu_2(dx_1) \otimes \mu_2(dx_2) + 4 \mu_1(dx_1) \otimes \mu_3(dx_2)$. 
		(In fact the ordering of the tensor products will not matter because the $\FFF^{(k)}[\mu]: (\RR^d)^k \to \RR$ are permutation-invariant.)
		\item For any $k \geq 1$, for any measure $\bm{s}$ over $(\RR^d)^k$ and any $\bm{f}: (\RR^d)^k \to \RR$, denote
		$\bm{f} :_{\int} \bm{s} = \int_{\RR^d} ... \int_{\RR^d} \bm{f}(x_1, ..., x_k) \bm{s}(dx_1, ..., dx_k)$.
	\end{itemize}
	Then for any $n \geq 1$,
	\begin{equation*}
		D^n_\mu \FFF(\Phi)
		= \sum_{k=1}^n \FFF^{(k)}[\mu] :_{\int} B_{n,k}\left( -\nabla : [\mu \Phi], (-\nabla)^2 : [\mu \Phi^{\otimes 2}], (-\nabla)^3 : [\mu \Phi^{\otimes 3}], ... \right).
	\end{equation*}
	In particular,
	\begin{align*}
		D^1_\mu \FFF(\Phi) &= \int \Phi \cdot \nabla \FFF'[\mu] \,d\mu \\
		D^2_\mu \FFF(\Phi) &= \int \Phi^\top \nabla^2 \FFF'[\mu] \Phi \,d\mu
		+ \int d\mu(x) \Phi(x) \cdot \nabla_x \int d\mu(x') \Phi(x') \cdot \nabla_{x'} \FFF''[\mu](x,x') \\
		D^3_\mu \FFF(\Phi) &= \int d\mu\, \Phi_t^{\otimes 3} \cdot \nabla^3 \FFF'[\mu] \\
		&~~ + 3 \int d\mu(x) \Phi_t(x) \cdot \nabla_x \int d\mu(x') (\Phi_t(x')^{\otimes 2}) \cdot \nabla_{x'}^2 \FFF''[\mu](x,x') \\
		&~~ + \int d\mu(x_1) \Phi_t(x_1) \cdot \nabla_{x_1} \int d\mu(x_2) \Phi_t(x_2) \cdot \nabla_{x_2} \int d\mu(x_3) \Phi_t(x_3) \cdot \nabla_{x_3} \FFF^{(3)}[\mu](x_1,x_2,x_3).
	\end{align*}
\end{proposition}

\begin{proof}
	Let $\partial_t \mu_t = -\nabla \cdot (\mu_t \Phi_t)$ be a convective geodesic, with $\mu_0=\mu$ and $\Phi_0=\Phi$. Then one can check by induction that for any $n \geq 0$, 
	\begin{equation*}
		\forall f: \RR^d \to \RR,~
		\frac{d^n}{dt^n} \int f d\mu_t = \int (\nabla^n f)(x) : \Phi_t^{\otimes n}(x) d\mu_t(x),
		~~~~\text{i.e.,}~~~~
		\partial_t^n \mu_t = (-\nabla)^n : [\mu_t \Phi_t^{\otimes n}].
	\end{equation*}
	So by \autoref{lm:fdb_wasserstein:wassdiffs:meas_space_FdB} below,
	\begin{align*}
		D^n_{\mu_t} \FFF(\Phi_t) 
		&= \frac{d^n}{dt^n} \FFF(\mu_t)
		= \sum_{k=1}^n \FFF^{(k)}[\mu_t] :_{\int} B_{n,k}\left( \partial_t \mu_t, \partial_t^2 \mu_t, \partial_t^3 \mu_t, ... \right) \\
		&= \sum_{k=1}^n \FFF^{(k)}[\mu_t] :_{\int} B_{n,k}\left( -\nabla : [\mu_t \Phi_t], (-\nabla)^2 : [\mu_t \Phi_t^{\otimes 2}], (-\nabla)^3 : [\mu_t \Phi_t^{\otimes 3}], ... \right).
		\rqedhere
	\end{align*}
\end{proof}

\begin{lemma} \label{lm:fdb_wasserstein:wassdiffs:meas_space_FdB}
	Let $\FFF: \PPP(\RR^d) \to \RR$ and any curve $(\mu_t)_t$ in $\PPP(\RR^d)$ which is infinitely differentiable in~$t$ in the distributional sense (not necessarily absolutely continuous in the Wasserstein sense).
	Then using the same shorthands as in the proposition above,
	\begin{equation*}
		\frac{d^n}{dt^n} \FFF(\mu_t)
		= \sum_{k=1}^n \FFF^{(k)}[\mu_t] :_{\int} B_{n,k}\left( \partial_t \mu_t, \partial_t^2 \mu_t, \partial_t^3 \mu_t, ... \right).
	\end{equation*}
\end{lemma}

\begin{proof}
	This follows (up to regularity issues which are not considered in this work) from Faa di Bruno's formula applied to $\FFF: \PPP(\RR^d) \to \RR$, by viewing $\PPP(\RR^d)$ as a subspace of the Banach space of signed measures over $\RR^d$.
\end{proof}

\subsection{Convective derivatives} \label{subsec:fdb_wasserstein:convderiv}

The following notion, which is classical in fluid mechanics, will allow us to make the definition of Wasserstein differentials more ``operational''.
Note that with this definition, convective geodesics are precisely the transport couples $(\mu_t, \Phi_t)_t$ such that $(\cD \Phi)_t$, the convective derivative of the velocity field itself, is zero.
\begin{definition} \label{def:fdb_wasserstein:convderiv:convderiv}
    Given a velocity field $\Phi_t$,
    the \emph{convective} or \emph{material derivative} of a time-dependent tensor field $g_t$ is defined as $(\cD g)_t = \frac{d}{dt} g_t + \Phi_t \cdot \nabla g_t$, which is a time-dependent tensor field of the same type as $g_t$.
\end{definition}

The following lemmas explain the usefulness of the convective derivative in the context of probability flows.
They are all reformulations of classical facts about Eulerian calculus specialized to the Euclidean space~\cite[Chapter~14]{villani_optimal_2009} \cite[Chapter~3]{gigli_second_2012}.
\begin{lemma} \label{lm:fdb_wasserstein:convderiv:ddt_lambdat_mut}
    For any transport couple $\partial_t \mu_t = -\nabla \cdot (\mu_t \Phi_t)$ and any time-dependent tensor field $g_t$,
    \begin{equation*}
        \frac{d}{dt} \int g_t \,d\mu_t
        = \int (\cD g)_t d\mu_t.
    \end{equation*}
\end{lemma}

\begin{proof}
    We have by definition
    \begin{equation*}
        \frac{d}{dt} \int g_t \,d\mu_t
        = \int \left( \frac{d}{dt} g_t \right) d\mu_t
        + \int g_t \,d(\partial_t \mu_t)
        = \int \left( \frac{d}{dt} g_t \right) d\mu_t
        + \int \Phi_t \cdot \nabla g_t \,d\mu_t
        = \int (\cD g)_t d\mu_t.
        \rqedhere
    \end{equation*}
\end{proof}

\begin{lemma} \label{lm:fdb_wasserstein:convderiv:product_rule}
    For any velocity field $(\Phi_t)_t$, the convective derivative satisfies the Leibniz product rule. That is,
    \begin{equation*}
        \forall g=(g_s)_s, \forall h=(h_s)_s,~
        \left( \cD (g_s \cdot h_s)_s \right)_t
        = (\cD g)_t \cdot h_t + g_t \cdot (\cD h)_t
    \end{equation*}
    where $g, h$ are both time-dependent vector fields or both time-dependent scalar fields.
    More generally, the same identity holds for $g, h$ being any time-dependent tensor fields of compatible dimensions and ``$\cdot$'' being a contraction over any compatible subset of indices.
\end{lemma}

\begin{proof}
    We have
    \begin{align*}
        \left( \cD (g_s \cdot h_s)_s \right)_t
        &= \frac{d}{dt} \left[ g_t \cdot h_t \right]
        + \Phi_t \cdot \nabla \left[ g_t \cdot h_t \right] \\
        &= g_t \cdot \frac{d}{dt} h_t + h_t \cdot \frac{d}{dt} g_t
        + \Phi_t^i \cdot \left[ (\nabla_i g_t) \cdot h_t + (\nabla_i h_t) \cdot g_t \right] \\
        &= g_t \cdot \frac{d}{dt} h_t + h_t \cdot \frac{d}{dt} g_t
        + h_t \cdot \left( \Phi_t^i \nabla_i g_t \right) + g_t \cdot \left( \Phi_t^i \nabla_i h_t \right) \\
        &= h_t \cdot \left[ \frac{d}{dt} g_t + \Phi_t \cdot \nabla g_t \right] + g_t \cdot \left[ \frac{d}{dt} h_t + \Phi_t \cdot \nabla h_t \right] = h_t \cdot (\cD g)_t + g_t \cdot (\cD h)_t.
        \rqedhere
    \end{align*}
\end{proof}

\begin{lemma} 
    For any velocity field $(\Phi_t)_t$, the convective derivative satisfies the chain rule.
    We will only need its scalar version: 
    for any time-dependent scalar field $(\lambda_t)_t$ and any function $P_t(\lambda) = P(t, \lambda): \RR \times \RR \to \RR$,
        \begin{equation*}
            \left( \cD \left( P_s(\lambda_s) \right)_s \right)_t
            = (\partial_t P_t)(\lambda_t) + P_t'(\lambda_t) (\cD \lambda)_t.
        \end{equation*}
    Here ``$(P_s(\lambda_s))_s$'' refers to the time-dependent scalar field $x \mapsto P_s(\lambda_s(x))$.
    Note that $P_t$ is only a function over scalars, and not over scalar fields,
    so for example this lemma does not say anything about $\cD (\nabla \lambda)$.
\end{lemma}

\begin{proof}
    We have
    \begin{align*}
        \left( \cD \left( P_s(\lambda_s) \right)_s \right)_t
        &= \frac{d}{dt} \left[ P(t, \lambda_t) \right]
        + \Phi_t \cdot \nabla \left[ P_t(\lambda_t) \right] \\
        &= \partial_t P(t, \lambda_t) + P_t'(\lambda_t) \frac{d}{dt} \lambda_t
        + P_t'(\lambda_t) \Phi_t \cdot \nabla \lambda_t \\
        &= (\partial_t P_t)(\lambda_t) + P_t'(\lambda_t) (\cD \lambda)_t.
        \rqedhere
    \end{align*}
\end{proof}

\begin{lemma} 
\label{lm:fdb_wasserstein:convderiv:convderiv_nabla}
	For any velocity field $(\Phi_t)_t$, for any time-dependent tensor field $(g_t)_t$,
	\begin{align*}
		(\cD \nabla_i g)_t &= \nabla_i (\cD g)_t - (\nabla_i \Phi_t) \cdot (\nabla g_t).
    \end{align*}
\end{lemma}

\begin{proof}
	We have
	\begin{align*}
		(\cD \nabla_i g)_t &= \nabla_i \frac{d}{dt} g_t + \Phi_t \cdot \nabla (\nabla_i g_t)
		= \nabla_i \left[ (\cD g)_t - \Phi_t \cdot \nabla g_t \right]
		+ \Phi_t \cdot \nabla (\nabla_i g_t) \\
		&= \nabla_i (\cD g)_t
        - (\nabla_i \Phi) \cdot (\nabla g_t) - \Phi_t \cdot \nabla_i (\nabla g_t)
        + \Phi_t \cdot \nabla (\nabla_i g_t) \\
		&= \nabla_i (\cD g)_t
        - (\nabla_i  \Phi_t) \cdot (\nabla g_t).
        \rqedhere
    \end{align*}
    %
\end{proof}

\begin{remark} \label{rk:fdb_wasserstein:convderiv:convderiv_nabla_curved}
    All of the properties of the convective derivative reported above can be generalized to the case of flows over a Riemannian manifold $\MMM$ instead of $\RR^d$.
    However, \autoref{lm:fdb_wasserstein:convderiv:convderiv_nabla} must be adapted as follows.
    Firstly, the formula given in the lemma only holds for $(g_t)_t$ being a scalar field.
    Secondly, if $(g_t)_t$ is a vector field and $(\nabla_i g^j_t)_t$ denotes its Jacobian, then
    (suppressing the subscript~$t$ for brevity)
	\begin{equation*}
		(\cD \nabla_i g^j) = \nabla_i (\cD g^j) -  (\nabla_i \Phi^k) (\nabla_k g^j) - \Phi^k R^{~~j~}_{ik~l} g^l
    \end{equation*}
    where $R$ denotes the Riemann curvature tensor $R(V, W) U = \nabla_V \nabla_W U - \nabla_W \nabla_V U - \nabla_{[V,W]} U$.
    This is essentially a Eulerian formulation of (a refinement of) Bochner's formula.
    Here is a formal proof: for any vector field $\Psi$,
    \begin{align*}
        \Psi^i \cD \nabla_i g^j
        &= \Psi^i \left( \frac{d}{dt} \nabla_i g^j + \Phi^k \nabla_k (\nabla_i g^j) \right) 
        = \Psi^i \nabla_i \left( \cD g^j - \Phi^k \nabla_k g^j \right)
        + \Psi^i \left( \Phi^k \nabla_k \left( \nabla_i g^j \right) \right) \\
        &= \Psi^i \nabla_i \cD g^j
        - \Psi^i \nabla_i \left( \Phi^k \nabla_k g^j \right) 
        + \left[ \Phi^k \nabla_k \left( \Psi^i \nabla_i g^j \right) - \Phi^k (\nabla_k \Psi^i) (\nabla_i g^j) \right] \\
        &= \Psi^i \nabla_i \cD g^j
        - \nabla_\Psi \left( \nabla_\Phi g \right) 
        + \nabla_\Phi \left( \nabla_\Psi g \right)
        - \nabla_{(\nabla_\Phi \Psi)} g \\
        &= \Psi^i \nabla_i \cD g^j
        + \left[ R(\Phi, \Psi) g \right]^j
        - \nabla_{(\nabla_\Psi \Phi)} g^j \\
        &= \Psi^i \nabla_i \cD g^j
        + \Phi^k \Psi^i R^{~~j~}_{ki~l} g^l
        - \Psi^i (\nabla_i \Phi^k) (\nabla_k g^j).
    \end{align*}
    The announced formula follows since the above holds for all $\Psi$, noting that $R^{~~j~}_{ki~l} = -R^{~~j~}_{ik~l}$ by definition (that is, $R(V, W) U = -R(W, V) U$ for any $U, V, W$).
\end{remark}

\subsection{A Faa di Bruno's formula for derivatives along transport couples} \label{subsec:fdb_wasserstein:fdb}

We can now state the main result of this section.
\begin{samepage}
\begin{theorem} \label{thm:fdb_wasserstein:fdb_wasserstein}
    Consider a transport couple $\partial_t \mu_t = -\nabla \cdot (\mu_t \Phi_t)$ and a functional $\FFF: \PPP_2(\RR^d) \to \RR$.
    Introduce the following shorthands:
    \begin{itemize}
        \item For $n \geq 1$ and vector fields $\Phi_{(1)}, ..., \Phi_{(n)}$, let 
        \begin{equation*}
            D^n_\mu \FFF : (\Phi_{(1)} \otimes ... \otimes \Phi_{(n)})
            \coloneqq D^n_\mu \FFF(\Phi_{(1)}, ..., \Phi_{(n)}).
        \end{equation*}
        That is, we extend the (symmetric) $n$-linear form over vector fields $D^n_\mu \FFF$ into a (permutation-invariant) linear form over $(n, 0)$-tensor fields.
        \item For vector fields $\Psi_{(1)}, \Psi_{(2)}, ...$, denote by $B_{n,k}(\Psi_{(1)}, ..., \Psi_{(n-k+1)})$ the $(k, 0)$-tensor field defined by interpreting the Bell polynomial $B_{n,k}$ as a polynomial over the ring of tensor fields equipped with the tensor product. For example, $B_{4,2}(X_1, X_2, X_3) = 3 X_2^2 + 4 X_1 X_3$ and $B_{4,2}(\Psi_{(1)}, \Psi_{(2)}, \Psi_{(3)}) = 3 \Psi_{(2)}^{\otimes 2} + 4 \Psi_{(1)} \otimes \Psi_{(3)}$. 
        (In fact the ordering of the tensor products will not matter because Wasserstein differentials are permutation-invariant.)
        \item We denote by $(\cD^n \Phi)_t$ the iterated convective derivatives of the velocity field $\Phi_t$ itself, defined by induction by $(\cD^n \Phi)_t = (\cD (\cD^{n-1} \Phi_s)_s)_t$ for $n \geq 1$.
    \end{itemize}
    Then for any $n \geq 1$,
    \begin{equation*}
        \frac{d^n}{dt^n} \FFF(\mu_t)
        = \sum_{k=1}^n D^k_{\mu_t} \FFF : B_{n,k}(\Phi_t, (\cD \Phi)_t, (\cD^2 \Phi)_t, ..., (\cD^{n-k} \Phi)_t).
    \end{equation*}
\end{theorem}
\end{samepage}

\pagebreak

The proof of the proposition consists in manipulations of the Bell polynomials, and is deferred to \autoref{sec:apx_pf_fdbwass}.
As examples, the case $n=1$ is clear, since 
$\frac{d}{dt} \FFF(\mu_t) = \int \nabla \FFF'[\mu_t] \cdot \Phi_t d\mu_t = D^1_{\mu_t} \FFF(\Phi_t)$ by definition of the divergence. The case where the transport couple $(\mu_t, \Phi_t)_t$ is a convective geodesic is also clear, since $(\cD^n \Phi)_t = 0$ for any $n \geq 1$ and 
$B_{n,k}(\Phi_t, 0, ..., 0) = \Phi_t^{\otimes n}$~if $k=n$ and $0$~otherwise.

\subsection{Discussion: \texorpdfstring{$\PPP_2(\RR^d)$}{P2(Rd)} as a flat differential manifold} \label{subsec:fdb_wasserstein:discussion}

\autoref{thm:fdb_wasserstein:fdb_wasserstein} shows that a form of the univariate Faa di Bruno formula \eqref{eq:fdb_wasserstein:findim:flat:FdB} holds over $\PPP_2(\RR^d)$, with Wasserstein differentials playing the role of the differential tensors ($\nabla^k f$), and iterated convective accelerations playing the role of the time-derivatives of the curve ($\frac{d^l x}{dt^l}$).
On the other hand, it is known that, following the Riemannian analogy of Otto calculus, the Wasserstein space $\PPP_2(\RR^d)$ is not flat---and in fact much is known about its curvature \cite[Proposition~6.25]{ambrosio_users_2013}.
These two facts seem inconsistent with the result of \autoref{subsec:fdb_wasserstein:findim}, namely that, on a finite-dimensional Riemannian manifold, for the higher-order time-derivatives to be given by the univariate Faa di Bruno formula, it is necessary that the Riemann curvature tensor is zero.

This apparent mismatch is resolved by the fact that the convective derivative behaves as a connection on $\PPP_2(\RR^d)$ (in the sense of non-Riemannian differential geometry) which is different from the Levi-Civita connection following Otto calculus;
the latter is obtained by projecting the result of the former onto the space of gradient fields \cite[Definition~5.1]{gigli_second_2012}.
The difference fundamentally lies in that velocity fields which are gradient fields played no particular role in our derivations so far.
Furthermore, the convective derivative viewed as a connection on $\PPP_2(\RR^d)$ has zero Riemann curvature tensor, making \autoref{thm:fdb_wasserstein:fdb_wasserstein} intuitively consistent with \autoref{subsec:fdb_wasserstein:findim}.

Our claims that (a) the convective derivative formally endows $\PPP_2(\RR^d)$ with a differential manifold structure, and that (b) it is flat, are justified in \autoref{sec:apx_connection}.
There, we show that more generally, for a smooth manifold $\XXX$ equipped with a connection $\nabla$, the convective derivative formally endows $\PPP(\XXX)$ with a differential manifold structure which inherits some properties of $\XXX$ itself.
Note that $\XXX$ may not be Riemannian, and even if it is, $\nabla$ may not be the Levi-Civita connection.

\section{The higher-order Wasserstein differentials of the entropy} \label{sec:ent}

In this section, we compute the higher-order Wasserstein differentials of $H$, the (negative) differential entropy functional. Because it comes at little added cost and it may be of independent interest, we compute more generally the higher-order Wasserstein differentials of any internal energy functional.

Throughout this section, fix
$\nu = e^{-V(x)} dx$ a non-negative  measure on $\RR^d$---not necessarily a probability measure---and
$h: \RR_+ \to \RR$ $C^\infty$-smooth, and let
\begin{equation*}
	\EEE(\mu) = \EEE_{h,\nu}(\mu) = \int h\left(\frac{d\mu}{d\nu}(x) \right) d\nu(x).
\end{equation*}
We have $H = \EEE$ when $h(\rho) = \rho \log \rho$ and $V = 0$.

\subsection{Warm-up: second-order Wasserstein differentials} \label{subsec:ent:hess}

The following formula appeared in \cite[Formula~15.7]{villani_optimal_2009} and in \cite[Example~3]{li_hessian_2021}.
\begin{proposition} \label{prop:ent:hess:hess}
	The second-order Wasserstein differential of $\EEE$ is given by
	\begin{equation*}
		D^2_\mu \EEE(\Phi, \Phi)
		= \int \Gamma_2(\Phi, \Phi) p_1\left( \frac{d\mu}{d\nu} \right) d\nu
		+ \int (\AAA \Phi)^2 p_2\left( \frac{d\mu}{d\nu} \right) d\nu
	\end{equation*}
	where we defined the operators, 
	mapping vector fields to scalar fields,%
	\footnote{
		The operator $\AAA$ is called the Langevin-Stein operator.
		It is often used for the integration-by-parts formula:
		$\forall \lambda: \RR^d \to \RR,~
		\int (-\AAA \Phi) \lambda \,d\nu = \int \Phi \cdot \nabla \lambda \,d\nu$,
        but we never make use of this fact explicitly in this work.
		The operator $\Gamma_2$ is known as the iterated carr\'e du champ operator when applied to gradient fields \cite{bakry_diffusions_1985}.
	}
	\begin{align*}
		\AAA \Phi &= \nabla \cdot \Phi - \nabla V \cdot \Phi = \frac{1}{\nu} \nabla \cdot (\Phi \nu), 
		&
		\Gamma_2(\Phi, \Phi) 
		&= \trace(\nabla \Phi \cdot \nabla \Phi) + \Phi^\top \nabla^2 V \Phi,
	\end{align*}
	and where $p_1, p_2: \RR_+ \to \RR$ are the pressure and iterated pressure defined by
	\begin{equation*}
		p_1(\rho) = \rho h'(\rho) - h(\rho)
		~~~~\text{and}~~~~
		p_2(\rho) = \rho p_1'(\rho) - p_1(\rho).
	\end{equation*}   
	For example for
	$h(\rho) = \begin{cases}
		\frac{\rho^m-\rho}{m-1} ~~\text{if}~ m \neq 1 \\
		\rho \log \rho ~~\text{if}~m=1
	\end{cases}$
	then $p_1(\rho) = \rho^m$ and $p_2(\rho) = (m-1) \rho^m$.
\end{proposition}

We note that the two terms in the above expression of $D^2_\mu \EEE(\Phi, \Phi)$ do not correspond to the two terms in the expression of $D^2_\mu \FFF(\Phi, \Phi)$ in \autoref{prop:fdb_wasserstein:wassdiffs:expr_firstvars}, see \autoref{ex:ent:hess:hessian_KL} below.
Although the proposition would follow from the same computations as in \cite[Formula~15.7]{villani_optimal_2009} except that we manipulate vector fields $\Phi$ instead of gradient fields,
let us give here a significantly different and simpler proof.

\begin{proof}
	Consider a convective geodesic $\partial_t \mu_t = -\nabla \cdot (\mu_t \Phi_t)$ and denote $\rho_t = \frac{d\mu_t}{d\nu}$ for all $t$. The quantity we want to compute is
	$D^2_{\mu_t} \EEE(\Phi_t, \Phi_t) = \frac{d^2}{dt^2} \EEE(\mu_t)$.
	
	Let $g(s) = e^{-s} h(e^s)$, so that $h(\rho) = \rho g(\log \rho)$ for all $\rho > 0$.
	Then
	\begin{equation*}
		\EEE(\mu_t) = \int h(\rho_t) d\nu
		= \int g(\log \rho_t) d\mu_t.
	\end{equation*}
	By the properties of the convective derivative,
	\begin{align*}
		\frac{d}{dt} \EEE(\mu_t) 
		&= \int \left( \cD g(\log \rho) \right)_t d\mu_t 
		= \int g'(\log \rho_t) (\cD \log \rho)_t d\mu_t \\
		\text{and}~~
		\frac{d^2}{dt^2} \EEE(\mu_t) 
		&= \int g'(\log \rho_t) (\cD^2 \log \rho)_t d\mu_t 
		+ \int g''(\log \rho_t) (\cD \log \rho)_t^2 d\mu_t,
	\end{align*}
	where in the first line we used the chain rule, and in the second line we used the product rule and again the chain rule for the second term.
	Now
	\begin{align}
		g'(s) = \frac{d}{ds} [e^{-s} h(e^s)]
		&= h'(e^s) - e^{-s} h(e^s)
		& &\text{so} &
		g'(\log \rho) &= h'(\rho) - \frac{1}{\rho} h(\rho)
		= \frac{1}{\rho} p_1(\rho) \\
		\label{eq:hess:gp_p1}
		& & &\text{and so} &
		g'(s) &= e^{-s} p_1(e^s), \\
		\text{so}~~
		g''(s) = \frac{d}{ds} [e^{-s} p_1(e^s)]
		&= p_1'(e^s) - e^{-s} p_1(e^s)
		& &\text{so} &
		g''(\log \rho) &= p_1'(\rho) - \frac{1}{\rho} p_1(\rho)
		= \frac{1}{\rho} p_2(\rho). \nonumber
	\end{align}
	All that remains is to compute $(\cD \log \rho)_t$ and $(\cD^2 \log \rho)_t$, which is done in \autoref{lm:ent:hess:cDlogrho} just below.
	Substituting back into $\frac{d^2}{dt^2} \EEE(\mu_t)$ yields the announced expression.
\end{proof}

\begin{lemma} \label{lm:ent:hess:cDlogrho}
	Consider a convective geodesic $\partial_t \mu_t = -\nabla \cdot (\mu_t \Phi_t)$ and denote $\rho_t = \frac{d\mu_t}{d\nu}$ for all $t$. Then
	\begin{equation*}
		(\cD \log \rho)_t = -\AAA \Phi_t
		~~~~\text{and}~~~~
		(\cD^2 \log \rho)_t = \Gamma_2(\Phi_t, \Phi_t).
	\end{equation*}
\end{lemma}

\begin{proof}
	We have
	\begin{align*}
		(\cD \mu)_t
		&= \partial_t \mu_t + \Phi_t \cdot \nabla \mu_t 
        = -\nabla \cdot (\mu_t \Phi_t) + \Phi_t \cdot \nabla \mu_t
		= -\mu_t \nabla \cdot \Phi_t \\
		(\cD \log \mu)_t &= -\nabla \cdot \Phi_t \\
		(\cD \log \nu)_t &= -(\cD V)_t 
		= -\frac{d}{dt} V - \Phi_t \cdot \nabla V 
		= -\nabla V \cdot \Phi_t \\
		(\cD \log \rho)_t &= (\cD \log \mu)_t - (\cD \log \nu)_t
		= -\nabla \cdot \Phi_t +\nabla V \cdot \Phi_t
		= -\AAA \Phi_t.
	\end{align*}
	Finally let us compute $(\cD^2 \log \rho)_t = -(\cD \AAA \Phi)_t = -(\cD (\nabla \cdot \Phi))_t + (\cD (\nabla V \cdot \Phi))_t$.
	For the first term,
	\begin{equation*}
		(\cD (\nabla \cdot \Phi))_t
		= (\cD (\trace \nabla \Phi))_t
		= \trace (\cD (\nabla \Phi))_t
		= \trace \left[ 
		\nabla (\cD \Phi)_t
		- \nabla \Phi_t \cdot \nabla \Phi_t
		\right]
		= -\trace(\nabla \Phi_t \cdot \nabla \Phi_t)
	\end{equation*}
	where in the third equality we used \autoref{lm:fdb_wasserstein:convderiv:convderiv_nabla}
	and in the fourth one $(\cD \Phi)_t=0$.
	For the second term,
	\begin{equation*}
		(\cD (\nabla V \cdot \Phi))_t
		= \nabla V \cdot (\cD \Phi)_t + (\cD \nabla V) \cdot \Phi_t
		= 0 + (\Phi_t \cdot \nabla^2 V) \cdot \Phi_t
		= \Phi_t^\top \nabla^2 V \Phi_t.
	\end{equation*}
	Hence $-(\cD \AAA \Phi)_t = \trace(\nabla \Phi_t \cdot \nabla \Phi_t) + \Phi_t^\top \nabla^2 V \Phi_t = \Gamma_2(\Phi_t, \Phi_t)$.
\end{proof}

\begin{remark} \label{rk:ent:hess:Dlogrho_Poisson}
	That $(\cD \log \rho)_t$ does not depend explicitly on $\mu_t$ but only on $\Phi_t$, is a remarkable fact.
	It is a manifestation of the phenomenon that a transport couple $\partial_t \mu_t = -\nabla \cdot (\mu_t \Phi_t)$ is equivalent ``locally and at first order'' to a multiplicative dynamics in measure space.
	That is, formally, if $\Phi_t \approx \Phi$ for all $t \in [0,\eps]$, and if we denote $\nu = \mu_0$ and $\mu = \mu_\eps$, then $\mu \approx (\id + \eps \Phi)_\sharp \nu$ and so
	\begin{align*}
		\frac{d\mu}{dx}(x + \eps \Phi(x))
		\cdot \det \nabla (\id + \eps \Phi)(x)
		&\approx \frac{d\nu}{dx}(x) 
		= \frac{d\nu}{dx}(x + \eps \Phi(x)) \cdot e^{V(x+\eps \Phi(x)) - V(x)} \\
		\underbrace{
			\log \frac{d\mu}{d\nu}(x + \eps \Phi(x))
		}_{\approx \log \frac{d\mu}{d\nu}(x)}
		~+~ \underbrace{
			\log \det(I_d + \eps \nabla \Phi(x))
		}_{\approx \log(1+\eps \trace \nabla \Phi) \approx \eps \nabla \cdot \Phi}
		&\approx \underbrace{
			V(x + \eps \Phi(x))-V(x)
		}_{\approx \eps \Phi(x) \cdot \nabla V(x)} \\ 
		\restr{\frac{\partial_t \mu_t}{\mu_t}}{t=0}
		\approx
		\frac{1}{\eps} \log \frac{d\mu}{d\nu} 
		&\approx -\nabla \cdot \Phi + \nabla V \cdot \Phi.
	\end{align*}
	In optimal transport terms, this phenomenon is precisely the fact that the Monge-Amp{\`e}re equation ``linearizes'' into the Poisson equation, and the Talagrand inequality into the Poincar\'e inequality \cite[Section~7]{otto_generalization_2000}.
\end{remark}

\begin{example} \label{ex:ent:hess:hessian_KL}
	The relative entropy functional $\KLdiv{\mu}{\nu} = \int_{\RR^d} d\mu \log \frac{d\mu}{d\nu}$ is equal to $\EEE$ when $h(\rho) = \rho \log \rho$, in which case $p_1(\rho) = \rho$ and $p_2(\rho) = 0$.
	So
	\begin{equation*}
		D^2_\mu\left\{ \KLdiv{\cdot}{\nu} \right\}(\Phi, \Phi)
		= \int \left( \trace\left( \nabla \Phi \cdot \nabla \Phi \right) + \Phi^\top \nabla^2 V \Phi \right) d\mu.
	\end{equation*}
	For the sake of illustration, let us rederive this expression using the general formula from \autoref{prop:fdb_wasserstein:wassdiffs:expr_firstvars},
	\begin{equation*}
		D^2_\mu \FFF(\Phi,\Phi)
		= \int \Phi^\top \nabla^2 \FFF'[\mu] \Phi \,d\mu
		+ \int_x d\mu(x) \int_{x'} d\mu(x') \Phi(x)^\top \nabla_x \nabla_{x'} \FFF''[\mu](x,x') \Phi(x').
	\end{equation*}
	For $\FFF(\mu) = \KLdiv{\mu}{\nu}$, the first term in the general formula is
	\begin{align*}
		\int \Phi^\top \nabla^2 \log \frac{d\mu}{d\nu} \Phi \,d\mu 
		&= \int \Phi^\top \nabla^2 V \Phi \,d\mu 
		+ \int \Phi^\top \nabla \left( \frac{\nabla \mu}{\mu} \right) \Phi \,d\mu \\
		&= \int \Phi^\top \nabla^2 V \Phi \,d\mu 
		+ \int \Phi^\top \nabla^2 \mu\, \Phi
		- \int \Phi^\top \frac{(\nabla \mu) (\nabla \mu)^\top}{\mu^2} \Phi \,d\mu \\
		&= \int \Phi^\top \nabla^2 V \Phi \,d\mu 
		+ \int \Phi^\top \nabla^2 \mu\, \Phi
		- \int \frac{1}{\mu} \left( \Phi^\top \nabla \mu \right)^2
	\end{align*}
	and the term in the middle rewrites
	\begin{align*}
		\int \Phi^\top \nabla^2 \mu\, \Phi
		&= \int (\nabla^2_{ij} \mu) \Phi^i \Phi^j  = -\int (\nabla_j \mu) \nabla_i \left( \Phi^i \Phi^j \right) 
		= \int d\mu\, \nabla^2_{ji} \left( \Phi^i \Phi^j \right) \\
		&= \int d\mu\, \nabla_j \left[ 
		(\nabla_i \Phi^i) \Phi^j 
		+ \Phi^i (\nabla_i \Phi^j) 
		\right] \\
		&= \int d\mu\, \left[ 
		(\nabla^2_{ji} \Phi^i) \Phi^j
		+ (\nabla_i \Phi^i) (\nabla_j \Phi^j)
		+ (\nabla_j \Phi^i) (\nabla_i \Phi^j)
		+ \Phi^i (\nabla^2_{ji} \Phi^j) 
		\right] \\
		&= \int d\mu\, (\nabla_i \Phi^i)^2 
		+ \int d\mu\, (\nabla_j \Phi^i) (\nabla_i \Phi^j)
		+ 2 \int d\mu\, (\nabla^2_{ij} \Phi^i) \Phi^j \\
		&= \int d\mu\, (\nabla \cdot \Phi)^2 
		+ \int d\mu\, \trace\left( \nabla \Phi \nabla \Phi \right)
		+ 2 \int d\mu\, \Phi \cdot \nabla (\nabla \cdot \Phi),
	\end{align*}
	where in the fourth line we used that $\nabla^2_{ij} \Phi^k = \nabla^2_{ji} \Phi^k$ because $\RR^d$ is flat, so 
	\begin{equation*}
		\int \Phi^\top \nabla^2 \FFF'[\mu] \Phi \, d\mu
		= \int \! \left( 
		\Phi^\top \nabla^2 V \Phi
		+ (\nabla \cdot \Phi)^2 
		+ \trace\left( \nabla \Phi \nabla \Phi \right)
		+ 2 \Phi \cdot \nabla (\nabla \cdot \Phi)
		- \frac{1}{\mu^2} \left( \Phi^\top \nabla \mu \right)^2
		\right) d\mu.
	\end{equation*}
	The second term in the general formula is, denoting by $\delta_0$ the Dirac delta measure,
	\begin{align*}
		& \int_x d\mu(x) \int_{x'} d\mu(x')
		\Phi(x)^\top
		\nabla_x \nabla_{x'} \left( \frac{\delta_0(x-x')}{\mu(x)} \right)
		\Phi(x') \\
        &= \int_x \int_{x'} \left( \frac{\delta_0(x-x')}{\mu(x)} \right)\,
		\left[ \nabla \cdot (\mu \Phi) \right](x)\,
		\left[ \nabla \cdot (\mu \Phi) \right](x') \\
		&= 
		\int \frac{1}{\mu} \left( \nabla \cdot (\mu \Phi) \right)^2
		= \int \frac{1}{\mu} \left( \mu \nabla \cdot \Phi + \nabla \mu \cdot \Phi \right)^2 \\
        &= \int \mu \, (\nabla \cdot \Phi)^2
		+ 2 \int (\nabla \cdot \Phi) (\nabla \mu \cdot \Phi)
		+ \int \frac{1}{\mu} \left( \nabla \mu \cdot \Phi \right)^2
	\end{align*}
	and the term in the middle rewrites
	\begin{equation*}
		2 \int \nabla \mu \cdot \left[ (\nabla \cdot \Phi) \Phi \right]
		= -2\int d\mu\, \nabla \cdot \left[ (\nabla \cdot \Phi) \Phi \right]
		= -2\int d\mu\, \left[ (\nabla \cdot \Phi)^2 + \Phi \cdot \nabla \left( \nabla \cdot \Phi \right) \right],
	\end{equation*}
	so
	\begin{multline*}
		\int_x d\mu(x) \int_{x'} d\mu(x') \Phi(x)^\top \nabla_x \nabla_{x'} \FFF''[\mu](x,x') \Phi(x') \\
		= \int \left( (\nabla \cdot \Phi)^2
		- 2 \left[ (\nabla \cdot \Phi)^2 + \Phi \cdot \nabla \left( \nabla \cdot \Phi \right) \right]
		+ \frac{1}{\mu^2} \left( \nabla \mu \cdot \Phi \right)^2
		\right) d\mu.
	\end{multline*}
	Gathering all the terms, several cancellations occur, and we indeed end up with the expression given above:
	$D^2_\mu \KLdiv{\cdot}{\mu}(\Phi,\Phi)
	= \int \Phi^\top \nabla^2 V \Phi \,d\mu
	+ \int \trace\left( \nabla \Phi \nabla \Phi \right) \,d\mu$.
\end{example}

\begin{remark} \label{rk:ent:hess:riem_hard}
	All of the manipulations of this subsection can be generalized to the case of probability measures over a Riemannian manifold $\MMM$---and even over a smooth manifold equipped with a torsion-free connection $\nabla$---instead of $\RR^d$.
	Almost all of the computations are unchanged, except that the curvature of the manifold contributes an additional term in \autoref{lm:fdb_wasserstein:convderiv:convderiv_nabla}, as noted in \autoref{rk:fdb_wasserstein:convderiv:convderiv_nabla_curved}.
	So the formula for $D^2_\mu \EEE$ is almost unchanged, except for an additional term in the expression of $\Gamma_2$, namely
	$\Ric(\Phi, \Phi)$ where $\Ric$ denotes Ricci curvature, related to the Riemann curvature tensor by $\Ric_{kl} = R^{~~i~}_{ik~l}$~\cite[Chapter~14]{villani_optimal_2009}.
	
	It seems likely that the higher-order Wasserstein differentials of the entropy on a curved manifold will also involve the Riemann curvature tensor and its derivatives.
	This makes the generalization of the results of the next subsection to Riemannian manifolds, instead of $\RR^d$, seem difficult.
\end{remark}

\subsection{Higher-order Wasserstein differentials} \label{subsec:ent:higherord}

\begin{definition} \label{def:ent:higherord:pk}
	Given $h: \RR_+ \to \RR$, let $(p_k)_{k \geq 0}$ be the sequence of functions $p_k: \RR_+ \to \RR$ defined by
	$p_0 = h$ and $p_{k+1}(\rho) = \rho p_k'(\rho) - p_k(\rho)$, called the \emph{iterated pressures}.
	
	Also let $g(s) = e^{-s} h(e^s)$, so that $h(\rho) = \rho g(\log \rho)$ for all $\rho>0$.
	Then by similar computations as in \eqref{eq:hess:gp_p1}, one can check by induction that $\rho g^{(k)}(\log \rho) = p_k(\rho)$ for all $k \geq 0$.
%
\end{definition}

Consider a convective geodesic $\partial_t \mu_t = -\nabla \cdot (\mu_t \Phi_t)$ and denote $\rho_t(x) = \frac{d\mu_t}{d\nu}(x)$.
We have
\begin{align*}
	\EEE(\mu_t) &= \int h(\rho_t(x)) d\nu(x)
	= \int g(\log \rho_t(x)) d\mu_t(x) \\
	\text{and so}~~~~
	\frac{d^n}{dt^n} \EEE(\mu_t) &= \int \left[ \cD^n g(\log \rho) \right]_t d\mu_t \\
	&= \int \sum_{k=1}^n g^{(k)}(\log \rho_t) \cdot B_{n,k}\left( (\cD \log \rho)_t, ..., (\cD^{n-k+1} \log \rho)_t \right) d\mu_t \\
	&= \int \sum_{k=1}^n B_{n,k}\left( (\cD \log \rho)_t, ..., (\cD^{n-k+1} \log \rho)_t \right) p_k(\rho_t) d\nu
\end{align*}
by the scalar Faa di Bruno formula, since the operator $\cD$ behaves algebraically as a derivation as shown in \autoref{subsec:fdb_wasserstein:convderiv}.
So it only remains to compute the quantities $(\cD^n \log \rho)_t$ for all $n \geq 1$.

\begin{definition} \label{def:ent:higherord:Lambda}
	For any $n \geq 1$, let $\Lambda_n$ be the symmetric $n$-linear operator mapping vector fields to scalar fields defined by
	\begin{align*}
		\Lambda_n(\Phi, ..., \Phi)
		&= (-1)^n (n-1)!~ (\nabla_{i_1} \Phi^{i_2}) ... (\nabla_{i_{n-1}} \Phi^{i_n}) (\nabla_{i_n} \Phi^{i_1}) 
		+ \Phi^{i_1} ... \Phi^{i_n} \left( \nabla_{i_1} ... \nabla_{i_n} V \right) \\
		&= (-1)^n (n-1)!~ \trace\left( (\nabla \Phi)^n \right)
		+ \Phi^{\otimes n} : \nabla^n V.
	\end{align*}
	In particular, $\Lambda_1 = -\AAA$ and $\Lambda_2 = \Gamma_2$ as defined in \autoref{prop:ent:hess:hess}.
    We will also write $\Lambda_n(\Phi)$ for $\Lambda_n(\Phi, ..., \Phi)$.
	More explicitly, for any vector fields $g_{(1)}, ..., g_{(n)}$, denoting by $\frakS_n$ the set of permutations of $\{1, ..., n\}$,
	\begin{equation*}
		\Lambda_n(g_{(1)}, ..., g_{(n)})
		= (-1)^n \frac{1}{n} \sum_{\sigma \in \frakS_n} \trace\left( (\nabla g_{(\sigma(1))}) ... (\nabla g_{(\sigma(n))}) \right)
		+ \left( g_{(1)} \otimes ... \otimes g_{(n)} \right): \nabla^n V.
	\end{equation*}
	We also introduce the shorthand
	$\Lambda_n : \left( g_{(1)} \otimes ... \otimes g_{(n)} \right)
	\coloneqq \Lambda_n(g_{(1)}, ..., g_{(n)})$,
	i.e., we extend the $n$-linear map $\Lambda_n$ over vector fields into a linear map over $(n,0)$-tensors.
\end{definition}

\begin{proposition} \label{prop:ent:higherord:cDn=Lambdan}
	Consider $\partial_t \mu_t = -\nabla \cdot (\mu_t \Phi_t)$ a convective geodesic
	and let $\rho_t = \frac{d\mu_t}{d\nu}$.
	For any $n \geq 1$, we have
	$(\cD^n \log \rho_t) = \Lambda_n(\Phi_t)$.
\end{proposition}

\begin{proof}
	We proceed by induction. The case $n=1$ is shown by \autoref{lm:ent:hess:cDlogrho}.
	For $n \geq 1$, suppose $(\cD^n \log \rho)_t = \Lambda_n(\Phi_t)$,
	and let us compute separately the terms in $\nabla_i \Phi_t^j$ and in $V$ in $(\cD^{n+1} \log \rho)_t = \cD \Lambda_n(\Phi_t)$.
	For the terms in $\nabla_i \Phi^j$, dropping the subscript $t$ for brevity, by symmetry we have
	\begin{align*}
		& \cD \left[ (\nabla_{i_1} \Phi^{i_2}) ... (\nabla_{i_{n-1}} \Phi^{i_n}) (\nabla_{i_n} \Phi^{i_1}) \right]
		= \left( \frac{d}{dt} + \Phi \cdot \nabla \right) \left[ (\nabla_{i_1} \Phi^{i_2}) ... (\nabla_{i_{n-1}} \Phi^{i_n}) (\nabla_{i_n} \Phi^{i_1}) \right] \\
		&= n \, (\nabla_{i_1} \Phi^{i_2}) ... (\nabla_{i_{n-1}} \Phi^{i_n}) \left( \frac{d}{dt} \nabla_{i_n} \Phi^{i_1} \right)
		+ \Phi^{i_{n+1}} \nabla_{i_{n+1}} \left[ (\nabla_{i_1} \Phi^{i_2}) ... (\nabla_{i_{n-1}} \Phi^{i_n}) (\nabla_{i_n} \Phi^{i_1}) \right] \\
		&= n \, (\nabla_{i_1} \Phi^{i_2}) ... (\nabla_{i_{n-1}} \Phi^{i_n}) 
		\!
        \underbrace{
			\left( \nabla_{i_n} \left[ -\Phi \cdot \nabla \Phi \right]^{i_1} \right) 
		}_{\scriptsize
			\begin{aligned}
				&= -\nabla_{i_n} \left[ \Phi^{i_{n+1}} (\nabla_{i_{n+1}} \Phi^{i_1}) \right] \\
				&= -(\nabla_{i_n} \Phi^{i_{n+1}}) (\nabla_{i_{n+1}} \Phi^{i_1}) \\
				&~~~~ - \Phi^{i_{n+1}} \left( \nabla_{i_n} \nabla_{i_{n+1}} \Phi^{i_1} \right)
			\end{aligned}
		}
		\!\!
        +\, n \, \Phi^{i_{n+1}} (\nabla_{i_1} \Phi^{i_2}) ... (\nabla_{i_{n-1}} \Phi^{i_n}) \left( \nabla_{i_{n+1}} \nabla_{i_n} \Phi^{i_1} \right) \\
		&= -n \, (\nabla_{i_1} \Phi^{i_2})... (\nabla_{i_{n}} \Phi^{i_{n+1}}) (\nabla_{i_{n+1}} \Phi^{i_1})
	\end{align*}
    where for the underbraced term we used the fact that $(\mu_t, \Phi_t)_t$ is a convective geodesic so that $\frac{d}{dt} \Phi^{i_1} = (-\Phi \cdot \nabla \Phi)^{i_1}$.
	For the terms in $V$, likewise,
	\begin{align*}
		\MoveEqLeft
		\cD \left[ \Phi^{i_1} ... \Phi^{i_n} \left( \nabla_{i_1} ... \nabla_{i_n} V \right) \right]
		= \left( \frac{d}{dt} + \Phi \cdot \nabla \right) \left[ \Phi^{i_1} ... \Phi^{i_n} \left( \nabla_{i_1} ... \nabla_{i_n} V \right) \right] \\
		&= n\, \Phi^{i_1} ... \Phi^{i_{n-1}} \left( \frac{d}{dt} \Phi^{i_n} \right) \left( \nabla_{i_1} ... \nabla_{i_n} V \right)
		+ \Phi^{i_{n+1}} \nabla_{i_{n+1}} \left[ \Phi^{i_1} ... \Phi^{i_n} \left( \nabla_{i_1} ... \nabla_{i_n} V \right) \right] \\
		&= n\, \Phi^{i_1} ... \Phi^{i_{n-1}} 
		\underbrace{
			\left[ -\Phi \cdot \nabla \Phi \right]^{i_n}
		}_{= -(\nabla_{i_{n+1}} \Phi^{i_n}) \Phi^{i_{n+1}}}
		\left( \nabla_{i_1} ... \nabla_{i_n} V \right) \\
		&~~~~ + n\, \Phi^{i_{n+1}} \Phi^{i_1} ... \Phi^{i_{n-1}} \left( \nabla_{i_{n+1}} \Phi^{i_n} \right) \left( \nabla_{i_1} ... \nabla_{i_n} V \right)
		+ \Phi^{i_{n+1}} \Phi^{i_1} ... \Phi^{i_n} \left( \nabla_{i_1} ... \nabla_{i_n} \nabla_{i_{n+1}} V \right) \\
		&= \Phi^{i_{n+1}} \Phi^{i_1} ... \Phi^{i_n} \left( \nabla_{i_1} ... \nabla_{i_n} \nabla_{i_{n+1}} V \right).
	\end{align*}
	Putting the terms together, we indeed obtain
	\begin{align*}
		(\cD^{n+1} \log \rho)_t
		&= (-1)^{n+1} n! ~ (\nabla_{i_1} \Phi^{i_2})... (\nabla_{i_{n}} \Phi^{i_{n+1}}) (\nabla_{i_{n+1}} \Phi^{i_1}) 
		+ \Phi^{i_1} ... \Phi^{i_n} \Phi^{i_{n+1}} \left( \nabla_{i_1} ... \nabla_{i_n} \nabla_{i_{n+1}} V \right) \\
		&= \Lambda_{n+1}(\Phi),
	\end{align*}
	which concludes the induction.
\end{proof}

In summary, in this subsection we have shown the following.
\begin{theorem} \label{thm:ent:difftls_EEE}
    Let $\nu = e^{-V(x)} dx$ a non-negative  measure on $\RR^d$---not necessarily a probability measure---and
    $h: \RR_+ \to \RR$ $C^\infty$-smooth, and
    $
    	\EEE(\mu) = \EEE_{h,\nu}(\mu) = \int h\left(\frac{d\mu}{d\nu}(x) \right) d\nu(x)
    $.
    
	The $n$-th order Wasserstein differential of $\EEE$ is given by
	\begin{equation*}
		D^n_\mu \EEE(\Phi) = \int \sum_{k=1}^n B_{n,k}\left( \Lambda_1(\Phi), \Lambda_2(\Phi), ..., \Lambda_{n-k+1}(\Phi) \right) \, p_k\left( \frac{d\mu}{d\nu} \right) d\nu
	\end{equation*}
	where the $p_k$ are defined in \autoref{def:ent:higherord:pk}
	and the $\Lambda_k$ are defined in \autoref{def:ent:higherord:Lambda}.
\end{theorem}

As a particular case when $V=0$ and $h(\rho) = \rho \log \rho$, so that $p_1(\rho) = \rho$ and $p_k(\rho) = 0$ for $k \geq 2$, we obtain the following, since $B_{n,1}(X_1, ..., X_n) = X_n$.
\begin{corollary}
	The $n$-th order Wasserstein differential of $H$ is given by
	\begin{equation*}
		D^n_\mu H(\Phi) 
        = \int \Lambda_n(\Phi) \, d\mu
        = (-1)^n (n-1)! \int \trace\left( (\nabla \Phi)^n \right) \, d\mu.
	\end{equation*}
\end{corollary}

\begin{remark}
	For any $n \geq 1$, symmetric matrix $S \succeq 0$ and function $u: \RR^d \to \RR$, for $\Phi = S \nabla u$,
	\begin{equation*}
		\trace\left( (\nabla \Phi)^n \right)
		= \trace\left( (S \nabla^2 u)^n \right)
		= \trace\left( (S^\half \nabla^2 u S^\half)^n \right),
	\end{equation*}
	which is non-negative when $n$ is even or $u$ is convex.
\end{remark}

\subsection{``Pointwise in space'' expansions and relation to Ledoux's \texorpdfstring{operators~$\tGamma_n$}{operators tildeGamman}} \label{subsec:ent:ptwise}

In \cite[Section~3]{ledoux_algebre_1995}, Ledoux introduced $n$-linear operators~$\tGamma_n$ which map scalar fields to scalar fields and such that,
for the Langevin dynamics $\partial_t \mu_t = \nabla \cdot \left( \mu_t \nabla \log \frac{d\mu_t}{d\nu} \right)$ and $h(\rho) = \rho \log \rho$, it holds $\frac{d^n}{dt^n} \EEE_{h,\nu}(\mu_t) = (-1)^n 2^{n-1} \int \tGamma_n\left( \log \rho_t \right) d\mu_t$,
where $\rho_t = \frac{d\mu_t}{d\nu}$.
This leads to computations at a ``pointwise'' level, i.e., manipulating scalar fields instead of scalars.
In this subsection we show that the iterated convective derivatives of $\log \rho_t$ furnish a different pointwise expansion.
We did not manage to make explicit the connection between the $\tGamma_n$ and our alternative pointwise expansion.

Consider any transport couple $\partial_t \mu_t = -\nabla \cdot (\mu_t \Phi_t)$, with potentially non-zero convective accelerations, and a reference measure $\nu$. Combining \autoref{thm:fdb_wasserstein:fdb_wasserstein} and \autoref{thm:ent:difftls_EEE} would yield an expression for $\frac{d^n}{dt^n} \EEE_{h,\nu}(\mu_t) = \frac{d^n}{dt^n} \int \log \rho_t d\mu_t = \int (\cD^n \log \rho)_t d\mu_t$.
The following result gives a ``pointwise in space'' refinement of it, by computing the scalar fields $(\cD^n \log \rho)_t$ themselves.

\begin{proposition}
	Consider any transport couple $\partial_t \mu_t = -\nabla \cdot (\mu_t \Phi_t)$,
	and let $\rho_t = \frac{d\mu_t}{d\nu}$.
	For any~$n \geq 1$,
	\begin{equation*}
		(\cD^n \log \rho)_t
		= \sum_{k=1}^n \Lambda_k : B_{n,k}(\Phi_t, (\cD \Phi)_t, ..., (\cD^{n-k} \Phi)_t).
	\end{equation*}
\end{proposition}

\begin{proof}
	Given the lemma just below, this is a direct application of \autoref{prop:apx_bell:abstract_vectorspaces}, an abstract result on derivations and Bell polynomials.
\end{proof}

\begin{lemma} \label{lm:ent:ptwise:DLambdan}
	For any velocity field $\Phi_t$ and any time-dependent vector fields $g_{(1)} = (g_{(1)t})_t, g_{(2)}, \allowbreak ..., \allowbreak g_{(n)}$,
	it holds
	\begin{equation*}
		(\cD \Lambda_n(g_{(1)}, ..., g_{(n)}))_t
		= \Lambda_{n+1}(g_{(1)}, ..., g_{(n)}, \Phi_t) 
		+ \sum_{h_1+...+h_n=1} \Lambda_n\left( \cD^{h_1} g_{(1)}, ..., \cD^{h_n} g_{(n)} \right).
	\end{equation*}
\end{lemma}

\begin{proof}
	By symmetry and Leibniz product rule, 
	for any permutation $\sigma \in \frakS_n$,
	\begin{align*}
		\left( \cD \trace\left( (\nabla g_{\sigma(1)}) ... (\nabla g_{\sigma(n)}) \right) \right)_t
		&= \sum_{h_1+...+h_n=1} \trace \left( (\cD^{h_1} \nabla g_{\sigma(1)}) ~ ... ~ (\cD^{h_n} \nabla g_{\sigma(n)}) \right) \\
		&= 
		\sum_{h_1+...+h_n=1} \trace \left( (\nabla \cD^{h_1} g_{\sigma(1)}) ~ ... ~ (\nabla \cD^{h_1} g_{\sigma(n)}) \right) \\
		&~~~~ ~~~~ - \sum_{i=1}^n \trace \left( (\nabla g_{\sigma(1)}) ~ ... ~ (\nabla g_{\sigma(i)}) (\nabla \Phi_t) (\nabla g_{\sigma(i+1)}) ~ ... ~ (\nabla g_{\sigma(n)}) \right),
	\end{align*}
	since
	$(\cD \nabla g)_t = \nabla (\cD g)_t - (\nabla \Phi_t) \cdot (\nabla g_t)$ 
	for any $(g_t)_t$ by \autoref{lm:fdb_wasserstein:convderiv:convderiv_nabla}.
	Now
	\begin{align*}
		& \frac{1}{n} \sum_{\sigma \in \frakS_n} \sum_{i=1}^n
		\trace \left( (\nabla g_{\sigma(1)}) ~ ... ~ (\nabla g_{\sigma(i)}) (\nabla \Phi_t) (\nabla g_{\sigma(i+1)}) ~ ... ~ (\nabla g_{\sigma(n)}) \right) \\
		&= \frac{1}{n} \sum_{\sigma \in \frakS_n} \sum_{i=1}^n
		\trace \left( (\nabla g_{\sigma(i+1)}) ~...~ (\nabla g_{\sigma(n)}) (\nabla g_{\sigma(1)}) ~...~ (\nabla g_{\sigma(i)}) (\nabla \Phi_t) \right) \\
		&= \sum_{\sigma \in \frakS_n}
		\trace \left( (\nabla g_{\sigma(1)}) ~...~ (\nabla g_{\sigma(n)}) (\nabla \Phi_t) \right) \\
		&= \frac{1}{n+1} \sum_{\sigma \in \frakS_{n+1}}
		\trace \left( (\nabla g_{\sigma(1)}) ~...~ (\nabla g_{\sigma(n+1)}) \right),
	\end{align*}
	where in the third line we exchanged the two summation symbols and noticed that the inner sum was independent of $i$,
	and in the last line we set $g_{n+1} = \nabla \Phi_t$ and counted the number of possible positions of ``$\nabla g_{n+1}$'' within an expression of the form $(\nabla g_{\sigma(1)}) ~...~ (\nabla g_{\sigma(n+1)})$.
	Likewise for the term in $V$,
	by symmetry and Leibniz product rule,
	\begin{align*}
		\MoveEqLeft \left( \cD \left[ (g_{(1)} \otimes ... \otimes g_{(n)}) : \nabla^n V \right] \right)_t \\
		&= \sum_{h_1+...+h_n=1} \left( (\cD^{h_1} g_{(1)} \otimes ... \otimes \cD^{h_n} g_{(n)}) \right) : \nabla^n V 
		+ (g_{(1)} \otimes ... \otimes g_{(n)}) : \left( \cD \nabla^n V \right)_t \\
		&= \sum_{h_1+...+h_n=1} \left( (\cD^{h_1} g_{(1)} \otimes ... \otimes \cD^{h_n} g_{(n)}) \right) : \nabla^n V 
		+ (g_{(1)} \otimes ... \otimes g_{(n)} \otimes \Phi_t) : \nabla^{n+1} V
	\end{align*}
	since
	$
		\left( \cD \nabla^n V \right)_t
		= 0 + \Phi_t \cdot \nabla \left[ \nabla^n V \right]
	$.
	Gathering the terms, we indeed have
	\begin{equation*}
		(\cD \Lambda_n(g_{(1)}, ..., g_{(n)}))_t
		= \Lambda_{n+1}(g_{(1)}, ..., g_{(n)}, \Phi_t) 
		+ \sum_{h_1+...+h_n=1} \Lambda_n\left( \cD^{h_1} g_{(1)}, ..., \cD^{h_n} g_{(n)} \right).
		\rqedhere
	\end{equation*}
\end{proof}

For the purpose of the GCM conjecture, it turns out that these pointwise expansions are not particularly helpful.
Indeed, for the heat flow $\partial_t \mu_t = \Delta \mu_t$, we have
$\frac{d^n}{dt^n} H(\mu_t) = \frac{d^n}{dt^n} \int (\log \mu_t) d\mu_t = \int (\cD^n \log \mu)_t d\mu_t$,
and one could hope that $(\cD^n \log \mu)_t(x)$ has a sign uniformly for all $x \in \RR^d$.
But it is not the case, as one can show that
\begin{itemize}
    \item For $n=1$, $(\cD \log \mu)_t = \trace \nabla^2 \log \mu_t$, which is non-negative only when $\mu_t$ is log-concave, whereas $\frac{d}{dt} H(\mu_t) = -\int \norm{\nabla \log \mu_t}^2 d\mu_t \leq 0$ unconditionally.
    \item For $n=2$, $(\cD^2 \log \mu)_t = 2 \trace((\nabla^2 \log \mu_t)^2) + \Delta(\Delta \log \mu_t) + (\nabla \log \mu_t)^\top (\nabla (\Delta \log \mu_t))$,
    which can take any sign a priori, whereas $\frac{d^2}{dt^2} H(\mu_t) = 2 \int \trace((\nabla^2 \log \mu_t)^2) d\mu_t \geq 0$.
\end{itemize}
Thus, as the case $n=2$ illustrates, the GCM conjecture relies on some cancellations occuring upon integration of $(\cD^n \log \mu)_t$ against $\mu_t$, using integrations by parts.

\pagebreak

For comparison, \cite[Theorem~1, Proposition~2]{ledoux_heat_2016} shows that for the heat flow in dimension $d=1$, then $\frac{d^n}{dt^n} H(\mu_t) = (-1)^n 2^{n-1} \int \tGamma_n\left( \log \mu_t \right) d\mu_t$, and
\begin{itemize}[itemsep=2pt]
    \item $\tGamma_1(u) = \norm{\nabla u}^2 \geq 0$
    and $\tGamma_2(u) = \trace((\nabla^2 u)^2) \geq 0$ pointwise (this is true for any $d$ \cite{ledoux_algebre_1995}).
    \item $\tGamma_3(u) = u^{\prime\prime\prime 2} - 2 u^{\prime\prime 3}, 
    \tGamma_4(u) = \left( u^{(4)} \right)^2 - 12 u'' u^{\prime\prime\prime 2} + 6 u^{\prime\prime 4} \geq 0$
    pointwise when $u$ is concave.
    \item $\tGamma_5(u)$ can take positive and negative values when $u$ is not concave.%
    \footnote{An algorithm to efficiently compute the $\tGamma_n$ in closed form for any $n \geq 1$, for the heat flow in dimension~$1$, is given in \cite{mansanarez_derivatives_2024}.}
\end{itemize}
Hence, the relation between the computations in this paper and the operators $\tGamma_n$ of \cite{ledoux_algebre_1995,ledoux_heat_2016} remains unclear.
Let us only emphasize that our $\Lambda_n$ arise naturally in the context of convective geodesics (straight-line mass transport), whereas the $\tGamma_n$ arise naturally in the context of the Langevin dynamics or the heat flow.
One may observe that $\Lambda_1(\nabla u) = -L u$, where $L$ is the Markov generator of the Langevin dynamics or the Laplacian in the case of the heat flow,
and that $\Lambda_2(\nabla u) = \tGamma_2(u)$.
We believe that the former is a deep-reaching fact (\autoref{rk:ent:hess:Dlogrho_Poisson}), and that the latter is only a coincidence.

\section{The formula for time-derivatives along gradient flows} \label{sec:gf}

\subsection{The formula for gradient flows on \texorpdfstring{$\RR^d$}{Rd}} \label{subsec:gf:findim}

For clarity of exposition, first consider the scalar case: let $f, x: \RR \to \RR$ and denote $x_{nt} = \frac{d^n}{dt^n} x_t$ and $f_{nt} = f^{(n)}(x_t)$ for any $n \geq 0$.
By Faa di Bruno's formula,
\begin{equation*}
    \forall n \geq 0,~
    \frac{d^n}{dt^n} f(x_t)
    = \sum_{k=0}^n B_{n,k}(x_{1t}, x_{2t}, ..., x_{(n-k+1)t}) \, f_{kt}
\end{equation*}
where $B_{0,0}=1$. Suppose that, on the other hand, the curve $(x_t)_t$ is a gradient flow, i.e.,
$x'_t = -g'(x_t)$ for some $g: \RR \to \RR$.
Then, denoting $g_{nt} = g^{(n)}(x_t)$ for any $n \geq 0$,
$x_{1t} = -g_{1t}$ and
\begin{equation*}
    \forall n \geq 0,~
    x_{(n+1)t} = -(g' \circ x)^{(n)}(t)
    = -\sum_{k=0}^n B_{n,k}(x_{1t}, x_{2t}, ..., x_{(n-k+1)t}) \, g_{(k+1)t}.
\end{equation*}
This is not of the correct format to apply the composition formula of~\autoref{subsec:apx_bell:compos}:
we expressed $x_{(n+1)t}$, and not $x_{nt}$, as a sum of Bell polynomials of order $n$, so the index is off by $1$.

In this subsection, we show how to compute the $\frac{d^n}{dt^n} f(x_t)$ when $g=f: \RR^d \to \RR$. That is, we compute the higher-order time-derivatives of a (multivariate) function along its own gradient flow.
Fix henceforth $f: \RR^d \to \RR$, $(x_t)_t$ a curve in $\RR^d$ such that $\frac{dx_t}{dt} = -\nabla f(x_t)$,
and denote $y^{(0)} = f(x_t)$ and $y^{(k)}_{i_1 i_2 ... i_k} = \left[ \nabla^k_{i_1 ... i_k} f \right](x_t)$.

\paragraph{The orders $n=1, 2, 3, 4$.}
We have, using that the tensors $y^{(2)}_{ij}$ resp.\ $y^{(3)}_{ijk}$ are symmetric,
\begin{align*}
	\frac{d}{dt} f(x_t) &= \nabla f(x_t) \cdot \frac{dx_t}{dt} = -\norm{\nabla f(x_t)}^2
	= -y^{(1)}_i y^{(1)}_i \\
	\frac{d^2}{dt^2} f(x_t) &= +2 y^{(1)}_i \cdot y^{(2)}_{ij} y^{(1)}_j = 2 y^{(2)}_{ij} y^{(1)}_i y^{(1)}_j \\
    \frac{d^3}{dt^3} f(x_t) &= -2 y^{(3)}_{ijk} y^{(1)}_k \cdot y^{(1)}_i y^{(1)}_j - 4 y^{(2)}_{ij} y^{(1)}_i \cdot y^{(2)}_{jk} y^{(1)}_k
    = -2 y^{(3)}_{ijk} y^{(1)}_i y^{(1)}_j y^{(1)}_k - 4 y^{(1)}_i y^{(2)}_{ij} y^{(2)}_{jk} y^{(1)}_k \\
    \frac{d^4}{dt^4} f(x_t) &= +2 y^{(4)}_{ijkl} y^{(1)}_l \cdot y^{(1)}_i y^{(1)}_j y^{(1)}_k 
    + 6 y^{(3)}_{ijk} y^{(1)}_i y^{(1)}_j \cdot y^{(2)}_{kl} y^{(1)}_l \\
    &~~~~ + 8 y^{(1)}_i y^{(2)}_{ij} \cdot y^{(3)}_{jkl} y^{(1)}_l \cdot y^{(1)}_k
    + 8 y^{(1)}_i y^{(2)}_{ij} y^{(2)}_{jk} y^{(2)}_{kl} y^{(1)}_l \\
    &= 2 y^{(4)}_{ijkl} y^{(1)}_i y^{(1)}_j y^{(1)}_k y^{(1)}_l 
    + 14 y^{(3)}_{ijk} y^{(2)}_{kl} y^{(1)}_i y^{(1)}_j y^{(1)}_l
    + 8 y^{(1)}_i y^{(2)}_{ij} y^{(2)}_{jk} y^{(2)}_{kl} y^{(1)}_l. \\[-0.9em]
\end{align*}

\paragraph{The tensor diagrams arising from gradient flows.}
Taking a step back, let us think about what is the seed that generated the above expressions.
What we used is really that $\frac{d}{dt}$ satisfies the Leibniz product rule and that
\begin{equation} \label{eq:gf:findim:rec_rule_yk}
    \frac{d}{dt} y^{(k)}_{i_1 ... i_k} = -y^{(k+1)}_{i_1 ... i_k i_{k+1}} ~ y^{(1)}_{i_{k+1}}.
\end{equation}

We can use this recurrence rule to compute explicitly the $\frac{d^n}{dt^n} f(x_t)$ for any $n \geq 1$.
Namely, $(-1)^n \frac{d^n}{dt^n} f(x_t)$ is a sum of monomials of the $y^{(1)}, y^{(2)}, ...$ with positive integer coefficients, and each monomial can be represented in tensor diagram notation by a tree with $n+1$ nodes.
In these tensor diagrams, every node with $k$ edges should be interpreted as the symmetric tensor $y^{(k)}$.
For ease of exposition, we extend the tensor diagram notation to represent $(-1)^n \frac{d^n}{dt^n} f(x_t)$ itself as a forest of trees, to be interpreted as the sum of the scalars represented by each tree.%
\footnote{See, e.g., the webpage \fnsurl{https://tensornetwork.org/diagrams/} for an introduction to tensor diagram notation. Here, the monomials appearing in $(-1)^n \frac{d^n}{dt^n} f(x_t)$ are scalars, so there is no ``dangling edge'' in their tensor diagram representations.
The fact that they are trees with $n$ edges and $n+1$ nodes reflects the fact that, if $f(x) = \alpha g(\beta x)$, then $\frac{d^n}{dt^n} f(x_t)$ is $n+1$-homogeneous in $\alpha$ and $2n$-homogeneous in $\beta$.
Furthermore, as can be deduced from the recurrence rule, it turns out that each tree with $n$ edges appears at least once in the representation of $(-1)^n \frac{d^n}{dt^n} f(x_t)$.}
Then the aforementioned recurrence rule \eqref{eq:gf:findim:rec_rule_yk} translates to the following rule for computing the tensor diagram representation, denoted $F^n$, of $(-1)^n \frac{d^n}{dt^n} f(x_t)$.

\begin{definition} \label{def:gf:findim:Fn}
	The undirected unlabeled forests $F^n$ ($n \in \NN$) are defined as follows.
\begin{samepage}
	\begin{itemize}[noitemsep]
	    \item $F^0$ is the graph with a single node.
	    \item For any $n \geq 0$:
	    for each node $X \in F^n$, denoting by $T$ the connected component containing~$X$,
	    \begin{itemize}[noitemsep]
	        \item make a copy $T_X$ of $T$,
	        \item add a leaf to $X$ in $T_X$;
	    \end{itemize}
	    and this defines the next forest $F^{n+1}$.
	\end{itemize}
\end{samepage}
    In diagrams we will write ``$n \times$'' to denote $n$ copies of the same graph and ``$+$'' to denote juxtaposition of graphs.
	The forests $F^1, F^2, F^3$ are shown in \autoref{fig:gf:findim:F123}
	and $F^4, F^5$ are shown in \autoref{fig:main_pf:F4},~\autoref{fig:main_pf:F5}.
\end{definition}

Let us summarize our derivation into a proposition for ease of future reference.
\begin{proposition} \label{prop:gf:findim:expr}
	Consider $f: \RR^d \to \RR$ and $(x_t)_t$ a curve in $\RR^d$ such that $\frac{dx_t}{dt} = -\nabla f(x_t)$,
	and denote the tensor $y^{(k)}_{i_1 ... i_k} = \left[ \nabla^k_{i_1 ... i_k} f \right](x_t)$.
	
	To any undirected tree $T$ with $n+1$ nodes, associate the scalar obtained by viewing it as a tensor network where every node of degree $k$ represents the symmetric tensor $y^{(k)}$. I.e., more explicitly: label the edges of $T$ from $1$ to $n$ and interpret each node with edges $i_1, ..., i_k$ as the symbol $y^{(k)}_{i_1...i_k}$. Then the scalar associated to $T$ is the product of all of these $n+1$ symbols. 
    Moreover, to any forest, associate the sum of the scalars associated to its tree components.
	
	Then, for any $n \geq 1$, $(-1)^n \frac{d^n}{dt^n} f(x_t)$ is equal to the scalar associated to $F^n$.
\end{proposition}

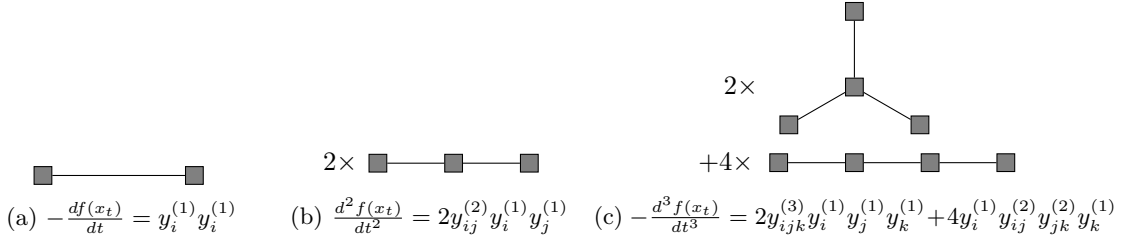
\begin{figure}[t]
	\begin{subfigure}{0.25\textwidth}
		\centering
		\adjustbox{scale=1}{%
			\begin{tikzpicture}
				\node[rectangle,fill=gray,draw] (A) at (0,0) {};
				\node[rectangle,fill=gray,draw] (B) at (2,0) {};
				\draw (A) -- (B);
			\end{tikzpicture}
		}
		\caption{$-\frac{d f(x_t)}{dt} = y^{(1)}_i y^{(1)}_i$}
		\label{subfig:F1}
	\end{subfigure}
	\hfill
	\begin{subfigure}{0.25\textwidth}
		\centering
		\adjustbox{scale=1}{%
			\begin{tikzpicture}
				\node[rectangle,fill=gray,draw] (A1) at (-1, 0) {};
				\node[rectangle,fill=gray,draw] (B1) at ( 0, 0) {};
				\node[rectangle,fill=gray,draw] (C1) at ( 1, 0) {};
				\draw (A1) -- (B1);
				\draw (B1) -- (C1);
				\node at (-1.5, 0) {$2 \times$};
			\end{tikzpicture}
		}
		\caption{$\frac{d^2 f(x_t)}{dt^2} = 2 y^{(2)}_{ij} y^{(1)}_i y^{(1)}_j$}
		\label{subfig:gf:findim:F2}
	\end{subfigure}
	\hfill
	\begin{subfigure}{0.458\textwidth}
		\centering
		\adjustbox{scale=1}{%
			\begin{tikzpicture}
				\node[rectangle,fill=gray,draw] (A1) at (1,0) {};
				\node[rectangle,fill=gray,draw] (B1) at (1,1) {};
				\node[rectangle,fill=gray,draw] (C1) at (1.866,-0.5) {};
				\node[rectangle,fill=gray,draw] (D1) at (1-0.866,-0.5) {};
				\draw (A1) -- (B1);
				\draw (A1) -- (C1);
				\draw (A1) -- (D1);
				\node at (-0.5, 0) {$2 \times$};
				\node[rectangle,fill=gray,draw] (A2) at (0,-1) {};
				\node[rectangle,fill=gray,draw] (B2) at (1,-1) {};
				\node[rectangle,fill=gray,draw] (C2) at (2,-1) {};
				\node[rectangle,fill=gray,draw] (D2) at (3,-1) {};
				\draw (A2) -- (B2);
				\draw (B2) -- (C2);
				\draw (C2) -- (D2);
				\node at (-0.7, -1) {$+4 \times$};
			\end{tikzpicture}
		}
		\caption{$-\frac{d^3 f(x_t)}{dt^3} = 2 y^{(3)}_{ijk} y^{(1)}_i y^{(1)}_j y^{(1)}_k + 4 y^{(1)}_i y^{(2)}_{ij} y^{(2)}_{jk} y^{(1)}_k$}
		\label{subfig:gf:findim:F3}
	\end{subfigure}
	\caption{The tensor diagram representations $F^n$ of $(-1)^n \frac{d^n}{dt^n} f(x_t)$ for $n \in \{1, 2, 3\}$}
	\label{fig:gf:findim:F123}
\end{figure}

\subsection{The formula for Wasserstein gradient flows on \texorpdfstring{$\RR^d$}{Rd}}
\label{subsec:gf:wass}

The discussion of \autoref{subsec:fdb_wasserstein:discussion} suggests that computations for Wasserstein gradient flows on $\RR^d$ follow the same algebraic rules as gradient flows on $\RR^d$, and so, that the time-derivatives of a functional along its own Wasserstein gradient flow are also given by the tensor diagrams $F^n$.
In this subsection we show that it is indeed the case, for the entropy.

The fact that $L^2_\mu$ equipped with the inner product $\innerprod{\Phi}{\Psi}_{L^2_\mu} = \int \Phi^\top \Psi \,d\mu$ is a Hilbert space,
suggests the following formal claim.
For any smooth enough functional $\FFF: \PPP_2(\RR^d) \to \RR$, $\mu \in \PPP_2(\RR^d)$, $n \geq 1$, and any vector fields $\Phi_{(1)}, ..., \Phi_{(n-1)} \in L^2_\mu$, there exists a vector field $\Psi \in L^2_\mu$ such that
\begin{equation*}
    \forall \xi \in L^2_\mu,~
    D_\mu^{n} \FFF(\Phi_{(1)}, ..., \Phi_{(n-1)}, \xi) = \int \Psi^\top \xi \,\,d\mu.
\end{equation*}
For the functional of interest in this paper, $\FFF = H$ the (negative) differential entropy, we indeed have the following result, although we note that the proof relies on an integration by parts rather than Hilbertness of~$L^2_\mu$.

\begin{proposition} \label{prop:gf:wass:innerprod}
    For any $\mu \in \PPP_2(\RR^d)$ and vector fields $\Phi_{(1)}, ..., \Phi_{(n-1)} \in L^2_\mu \cap \mathfrak{X}(\RR^d)$, there exists $\Psi \in L^2_\mu \cap \mathfrak{X}(\RR^d)$ such that
	\begin{equation*}
		\forall \xi \in L^2_\mu \cap \mathfrak{X}(\RR^d),~
		\int \Lambda_n(\Phi_{(1)}, ..., \Phi_{(n-1)}, \xi) \,d\mu = \int \Psi^\top \xi \,\,d\mu.
	\end{equation*}
\end{proposition}

\begin{proof}
    Denoting $M_i^{~j} = (-1)^n \sum_{\sigma \in \frakS_{n-1}} \nabla_i \Phi_{(\sigma(1))}^{i_1} ... \nabla_{i_{n-1}} \Phi_{(\sigma(n-1))}^j$,
    we have by definition of $\Lambda_n$ and by integration by parts
    \begin{align*}
        \int \Lambda_n(\Phi_{(1)}, ..., \Phi_{(n-1)}, \xi) \,d\mu 
        &= \int \trace(M(x) \nabla \xi(x)) \, d\mu(x)
        = \int M_i^{~j}~ \nabla_j \xi^i  ~d\mu \\
        &= -\int \xi^i \,\nabla_j \left( M_i^{~j} \mu \right)
        = -\int \xi^i \left( \nabla_j M_i^{~j} + M_i^{~j} \nabla_j \log\mu \right) d\mu.
        \rqedhere
    \end{align*}
\end{proof}

The above proposition justifies the following shorthands.
\begin{definition} \label{def:gf:wass:Lambdan_i_shorthand}
	For any $\mu \in \PPP_2(\RR^d)$ and vector fields $\Phi_{(1)}, ..., \Phi_{(n)} \in L^2_\mu \cap \mathfrak{X}(\RR^d)$, let
	\begin{equation*}
		\Lambda^{(n,\mu)}_{i_1...i_n}~ \Phi_{(1)}^{i_1} ... \Phi_{(n)}^{i_n} 
		\coloneqq \int \Lambda_n(\Phi_{(1)}, ..., \Phi_{(n)}) \,d\mu,
	\end{equation*}
	and let $\Lambda^{(n,\mu)}_{i_1...i_n}~ \Phi_{(1)}^{i_1} ... \Phi_{(n-1)}^{i_{n-1}}$ be the vector field in $L^2_\mu \cap \mathfrak{X}(\RR^d)$, indexed by $i_n$, such that
	\begin{equation*}
		\forall \xi \in L^2_\mu \cap \mathfrak{X}(\RR^d),~
		\int \xi^\top \left( \Lambda^{(n,\mu)}_{i_1...i_n}~ \Phi_{(1)}^{i_1} ... \Phi_{(n-1)}^{i_{n-1}} \right) d\mu
        = \int \Lambda_n(\Phi_{(1)}, ..., \Phi_{(n-1)}, \xi) \,d\mu.
	\end{equation*}
	For example, 
    $\Lambda^{(1,\mu)}_i \Phi^i = \int \Lambda_1(\Phi) \,d\mu = \int (-\nabla \cdot \Phi) \,d\mu = \int \Phi^i (\nabla_i \log \mu) \,d\mu$ by integration by parts,
    and $\Lambda^{(1,\mu)}_i = \nabla_i \log \mu$.
\end{definition}

\begin{proposition} \label{prop:gf:wass:expr}
	For any fixed $\mu$, to any undirected tree $T$ with $n+1$ nodes, associate a scalar as follows.
	Label the edges from $1$ to $n$ and interpret each node with edges $i_1, ..., i_k$ as the symbol $\Lambda^{(k,\mu)}_{i_1,...,i_k}$.
	The scalar associated to $T$ is the product of all of these $n+1$ symbols.
    Then, this procedure is well-defined, i.e., any compatible parenthesization of the product of symbols defines the same quantity via \autoref{def:gf:wass:Lambdan_i_shorthand}.

    Moreover, to any forest $F$, associate the sum of the scalars associated to its tree components, and denote it henceforth as $\val_\mu(F)$ (or simply $\val(F)$ when $\mu$ is clear from context).
	Consider the transport couple $(\mu_t, \Phi_t)_t$ given by $\Phi_t = -\nabla H'[\mu_t] = -\nabla \log \mu_t$ and $\partial_t \mu_t = -\nabla \cdot (\mu_t \Phi_t) = \Delta \mu_t$.
	Then for any $n \geq 1$, $(-1)^n \frac{d^n}{dt^n} H(\mu_t) = \val_{\mu_t}(F^n)$.
\end{proposition}

For example, for the forests $F^1, F^2, F^3$ given in \autoref{fig:gf:findim:F123},
\begin{itemize}
	\item The scalar associated to $F^1$ is
	$\val(F^1) = \Lambda^{(1,\mu)}_i \Lambda^{(1,\mu)}_i
	= \int \norm{\nabla \log \mu}^2 d\mu$,
	which is indeed equal to $-\frac{d}{dt} H(\mu_t)$ (for $\mu = \mu_t$) by De Bruijn's identity.
	\item The scalar associated to $F^2$ is
	\begin{equation*}
		\val(F^2) = 2 \Lambda^{(2,\mu)}_{i,j} \Lambda^{(1,\mu)}_{i} \Lambda^{(1,\mu)}_{j}
		= 2 \int \Lambda_2(\nabla \log \mu, \nabla \log \mu) d\mu
		= 2 \int \trace((\nabla^2 \log \mu)^2) d\mu,
	\end{equation*}
	which is indeed equal to $\frac{d^2}{dt^2} H(\mu_t)$ by \cite{bakry_diffusions_1985,villani_short_2000}.
	\item The scalar associated to $F^3$ is
	$\val(F^3) = 2 \Lambda^{(3,\mu)}_{ijk} \Lambda^{(1,\mu)}_i \Lambda^{(1,\mu)}_j \Lambda^{(1,\mu)}_k
	+ 4 \Lambda^{(1,\mu)}_i \Lambda^{(2,\mu)}_{ij} \Lambda^{(2,\mu)}_{jk} \Lambda^{(1,\mu)}_k$.
	Observe that the first term is equal to
	\begin{equation*}
		2 \Lambda^{(3,\mu)}_{ijk} \Lambda^{(1,\mu)}_i \Lambda^{(1,\mu)}_j \Lambda^{(1,\mu)}_k
		= 2 \int \Lambda_3(\nabla \log \mu) d\mu
		= 2 \cdot (-2) \int \trace((\nabla^2 \log \mu)^3) d\mu,
	\end{equation*}
	and that the second term is non-negative since
	\begin{equation*}
		4 \Lambda^{(1,\mu)}_i \Lambda^{(2,\mu)}_{ij} \Lambda^{(2,\mu)}_{jk} \Lambda^{(1,\mu)}_k
		= 4 \int \norm{\Psi}^2 d\mu
		~~\text{where}~~
		\Psi_j = \Lambda^{(2,\mu)}_{ji} \Lambda^{(1,\mu)}_i.
	\end{equation*}
	Note that if $\mu_0$ is log-concave then $\mu_t$ is too for all $t \geq 0$, and then ${\trace((\nabla^2 \log \mu_t)^3) \leq 0}$, and so $-\frac{d^3}{dt^3} H(\mu_t) = \val_{\mu_t}(F^3) \geq 0$.
	This provides a very short proof of~\eqref{eq:GCM_d_m} for any $d$ and $m=3$ provided that $\mu_0$ is log-concave, previously shown using information-theoretic arguments in \cite{toscani_concavity_2015} (where actually a much finer bound is obtained).
\end{itemize}

The proof of \autoref{prop:gf:wass:expr} follows from the exact same derivation as for \autoref{prop:gf:findim:expr}, with the help of \autoref{lm:ent:ptwise:DLambdan} in place of \eqref{eq:gf:findim:rec_rule_yk}.

\section{Proof of the main theorem} \label{sec:main_pf}

This section is dedicated to proving \autoref{thm:intro:GCMC_45}.
Namely, let $\mu_0 \in \PPP_2(\RR^d)$ and let $(\mu_t)_{t \geq 0}$ be given by $\partial_t \mu_t = \Delta \mu_t$;
assume that $\mu_0$ is log-concave, so that $\mu_t$ is log-concave too for all $t \geq 0$;
then we claim that $(-1)^n \frac{d^n}{dt^n} H(\mu_t) \geq 0$ for $n \in \{4, 5\}$.
By \autoref{prop:gf:wass:expr}, expressions for those two quantities are given by the tensor diagrams $F^4$ and $F^5$ from  \autoref{def:gf:findim:Fn}.
By applying the definition, one finds that $F^4$ is as shown in \autoref{fig:main_pf:F4}, and that $F^5$ is as shown in \autoref{fig:main_pf:F5}.
(As a sanity check, observe that $F^4$ consists of $4! = 24$ trees and $F^5$ of $5! = 120$ trees, consistent with \autoref{def:gf:findim:Fn}.)
We also introduce notations for the distinct tree components $S_1, S_2, S_3$ of $F^4$, respectively $T_1, ..., T_6$ of $F^5$, as shown in the figures.
\autoref{thm:intro:GCMC_45} is then a direct consequence of the following two propositions, where we recall that the notation $\val$ is defined in~\autoref{prop:gf:wass:expr}.

\begin{proposition} \label{prop:main_pf:n=4}
    Consider the trees $S_i$ ($i \in \{1, 2, 3\}$) shown in \autoref{fig:main_pf:F4}
    and $\mu \in \PPP_2(\RR^d)$ log-concave.
    Then for each $i$, $\val_\mu(S_i) \geq 0$.
\end{proposition}

\begin{proposition} \label{prop:main_pf:n=5}
    Consider the trees $T_i$ ($i \in \{1, ..., 6\}$) shown in \autoref{fig:main_pf:F5}
    and $\mu \in \PPP_2(\RR^d)$ log-concave.
    Then for each $i$, $\val_\mu(T_i) \geq 0$.
\end{proposition}

The remainder of this section, and of this paper, is dedicated to the proof of these two propositions.

\begin{remark} \label{rk:main_pf:allnonneg}
    The above two results, as well as the last bullet point of \autoref{subsec:gf:wass}, show that for any log-concave $\mu$ and for any tree $T$ with at most $6$ nodes, $\val_\mu(T) \geq 0$.
    It is natural to conjecture that the same holds for all trees.
    This would of course imply that \eqref{eq:GCM_d_m} holds for all $d, m$ provided that $\mu_0$ is log-concave.
\end{remark}

\begin{figure}[t]
	\centering
	\adjustbox{scale=1}{%
		\begin{tikzpicture}
			\node[rectangle,fill=gray,draw] (A1) at (0, 0) {};
			\node[rectangle,fill=gray,draw] (B1) at (1, 0) {};
			\node[rectangle,fill=gray,draw] (C1) at (-1, 0) {};
			\node[rectangle,fill=gray,draw] (D1) at (0, 1) {};
			\node[rectangle,fill=gray,draw] (E1) at (0, -1) {};
			\draw (A1) -- (B1);
			\draw (A1) -- (C1);
			\draw (A1) -- (D1);
			\draw (A1) -- (E1);
            \node at (-1.75, 0) {$2 \times$};
            \draw [decorate,decoration={brace,amplitude=5pt,mirror,raise=4ex}] (-1.1,-0.6) -- (1.1,-0.6) node[midway,yshift=-3em]{$S_1$};
			\node[rectangle,fill=gray,draw] (A2) at (3.2, 0.7) {};
			\node[rectangle,fill=gray,draw] (B2) at (3.2, -0.7) {};
			\node[rectangle,fill=gray,draw] (C2) at (4, 0) {};
			\node[rectangle,fill=gray,draw] (D2) at (5, 0) {};
			\node[rectangle,fill=gray,draw] (E2) at (6, 0) {};
			\draw (A2) -- (C2);
			\draw (B2) -- (C2);
			\draw (C2) -- (D2);
			\draw (D2) -- (E2);
			\node at (2.1, 0) {$+~ 14 \times$};
            \draw [decorate,decoration={brace,amplitude=5pt,mirror,raise=4ex}] (3.1,-0.5) -- (6.1,-0.5) node[midway,yshift=-3em]{$S_2$};
			\node[rectangle,fill=gray,draw] (A3) at (8, 0) {};
			\node[rectangle,fill=gray,draw] (B3) at (9, 0) {};
			\node[rectangle,fill=gray,draw] (C3) at (10, 0) {};
			\node[rectangle,fill=gray,draw] (D3) at (11, 0) {};
			\node[rectangle,fill=gray,draw] (E3) at (12, 0) {};
			\draw (A3) -- (B3);
			\draw (B3) -- (C3);
			\draw (C3) -- (D3);
			\draw (D3) -- (E3);
			\node at (7, 0) {$+~ 8 \times$};
            \draw [decorate,decoration={brace,amplitude=5pt,mirror,raise=4ex}] (7.9,0) -- (12.1,0) node[midway,yshift=-3em]{$S_3$};
		\end{tikzpicture}
	}
	\caption{The forest $F^4$}
	\label{fig:main_pf:F4}
\end{figure}
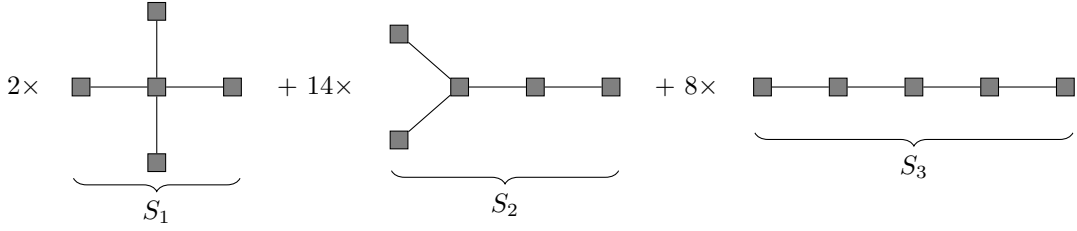

\begin{figure}[t]
	\centering
	\adjustbox{scale=0.9}{%
		\begin{tikzpicture}
			\node[rectangle,fill=gray,draw] (A1) at (0, 0) {};
			\node[rectangle,fill=gray,draw] (B1) at (90:1) {};
			\node[rectangle,fill=gray,draw] (C1) at (162:1) {};
			\node[rectangle,fill=gray,draw] (D1) at (234:1) {};
			\node[rectangle,fill=gray,draw] (E1) at (306:1) {};
			\node[rectangle,fill=gray,draw] (F1) at (18:1) {};
			\draw (A1) -- (B1);
			\draw (A1) -- (C1);
			\draw (A1) -- (D1);
			\draw (A1) -- (E1);
			\draw (A1) -- (F1);
			\node at (-1.75, 0) {$2 \times$};
            \draw [decorate,decoration={brace,amplitude=5pt,mirror,raise=4ex}] (-1.1,-0.5) -- (1.1,-0.5) node[midway,yshift=-3em]{$T_1$};
            \node[rectangle,fill=gray,draw] (A2) at (3, 0.5) {};
            \node[rectangle,fill=gray,draw] (B2) at (4, 0.5) {};
            \node[rectangle,fill=gray,draw] (C2) at (5, 0.5) {};
            \node[rectangle,fill=gray,draw] (D2) at (6, 0.5) {};
            \node[rectangle,fill=gray,draw] (E2) at (4, -0.5) {};
            \node[rectangle,fill=gray,draw] (F2) at (5, -0.5) {};
            \draw (A2) -- (B2);
            \draw (B2) -- (C2);
            \draw (C2) -- (D2);
            \draw (B2) -- (E2);
            \draw (C2) -- (F2);
			\node at (2, 0) {$+~ 14 \times$};
            \draw [decorate,decoration={brace,amplitude=5pt,mirror,raise=4ex}] (3.1,-0.3) -- (6.1,-0.3) node[midway,yshift=-3em]{$T_2$};
			\node[rectangle,fill=gray,draw] (A3) at (8, 0) {};
			\node[rectangle,fill=gray,draw] (B3) at (9, 0) {};
			\node[rectangle,fill=gray,draw] (C3) at (10, 0) {};
			\node[rectangle,fill=gray,draw] (D3) at (11, 0) {};
			\node[rectangle,fill=gray,draw] (E3) at (12, 0) {};
			\node[rectangle,fill=gray,draw] (F3) at (13, 0) {};
			\draw (A3) -- (B3);
			\draw (B3) -- (C3);
			\draw (C3) -- (D3);
			\draw (D3) -- (E3);
			\draw (E3) -- (F3);
			\node at (7, 0) {$+~ 16 \times$};
            \draw [decorate,decoration={brace,amplitude=5pt,mirror,raise=4ex}] (7.9,0) -- (13.1,0) node[midway,yshift=-3em]{$T_3$};
            \node[rectangle,fill=gray,draw] (A4) at (-1, -3) {};
            \node[rectangle,fill=gray,draw] (B4) at (0, -3) {};
            \node[rectangle,fill=gray,draw] (C4) at (1, -3) {};
            \node[rectangle,fill=gray,draw] (D4) at (2, -3) {};
            \node[rectangle,fill=gray,draw] (E4) at (0, -3.8) {};
            \node[rectangle,fill=gray,draw] (F4) at (0, -2.2) {};
            \draw (A4) -- (B4);
            \draw (B4) -- (C4);
            \draw (C4) -- (D4);
            \draw (B4) -- (E4);
            \draw (B4) -- (F4);
			\node at (-2, -3) {$+~ 22 \times$};
            \draw [decorate,decoration={brace,amplitude=5pt,mirror,raise=4ex}] (-1.1,-3.5) -- (2.1,-3.5) node[midway,yshift=-3em]{$T_4$};
            \node[rectangle,fill=gray,draw] (A5) at (4, -2.5) {};
            \node[rectangle,fill=gray,draw] (B5) at (5, -2.5) {};
            \node[rectangle,fill=gray,draw] (C5) at (6, -2.5) {};
            \node[rectangle,fill=gray,draw] (D5) at (7, -2.5) {};
            \node[rectangle,fill=gray,draw] (E5) at (8, -2.5) {};
            \node[rectangle,fill=gray,draw] (F5) at (6, -3.5) {};
            \draw (A5) -- (B5);
            \draw (B5) -- (C5);
            \draw (C5) -- (D5);
            \draw (D5) -- (E5);
            \draw (C5) -- (F5);
			\node at (3, -3) {$+~ 36 \times$};
            \draw [decorate,decoration={brace,amplitude=5pt,mirror,raise=4ex}] (3.9,-3.3) -- (8.1,-3.3) node[midway,yshift=-3em]{$T_5$};
            \node[rectangle,fill=gray,draw] (A6) at (10, -3-0.7) {};
            \node[rectangle,fill=gray,draw] (B6) at (10, -3+0.7) {};
            \node[rectangle,fill=gray,draw] (C6) at (10.8, -3) {};
            \node[rectangle,fill=gray,draw] (D6) at (11.8, -3) {};
            \node[rectangle,fill=gray,draw] (E6) at (12.8, -3) {};
            \node[rectangle,fill=gray,draw] (F6) at (13.8, -3) {};
            \draw (A6) -- (C6);
            \draw (B6) -- (C6);
            \draw (C6) -- (D6);
            \draw (D6) -- (E6);
            \draw (E6) -- (F6);
			\node at (9, -3) {$+~ 30 \times$};
            \draw [decorate,decoration={brace,amplitude=5pt,mirror,raise=4ex}] (9.7,-3.5) -- (13.9,-3.5) node[midway,yshift=-3em]{$T_6$};
		\end{tikzpicture}
	}
	\caption{The forest $F^5$}
	\label{fig:main_pf:F5}
\end{figure}
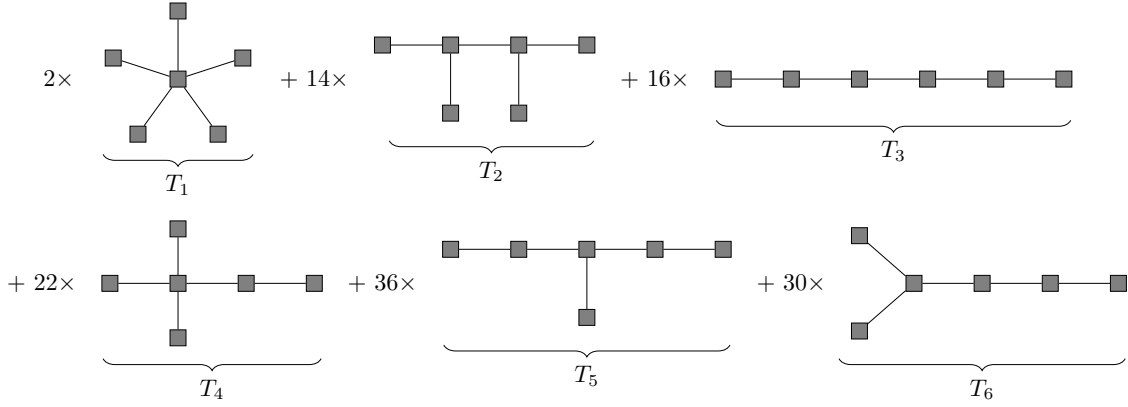

\subsection{Two diagrammatic representations for vector and tensor fields}

In order to present our computations concisely, it will be beneficial to introduce two shorthands.
The first one is a natural extension of \autoref{prop:gf:wass:expr} to vector fields.
\begin{definition} \label{def:main_pf:repr:expr_vect}
    Let us call ``dangling tree'' a pair $(T, x)$ where $T$ is a tree and $x$ is a distinguished node of $T$.
    For two dangling trees $\ulT, \ulS$, denote by $\ulT\text{---}\ulS$ the tree obtained by connecting them along their distinguished nodes.
    
    For any fixed $\mu$, to any dangling tree $\ulT = (T,x)$ with $n+1$ nodes, associate a vector field indexed by $j$ as follows.
    Label the edges of $T$ from $1$ to $n$ and interpret each node with edges $i_1, ..., i_k$ as the symbol $\Lambda^{(k,\mu)}_{i_1,...,i_k}$, except for the node $x$ which is interpreted as the symbol $\Lambda^{(k+1,\mu)}_{i_1,...,i_k,j}$.
    The vector field associated to $\ulT$ is the product of all of these $n+1$ symbols.
    Then, by the same arguments as for \autoref{prop:gf:wass:expr}, this procedure is well-defined, i.e., any compatible parenthesization of the product of symbols defines the same vector field via \autoref{def:gf:wass:Lambdan_i_shorthand}.
    
    We denote henceforth this vector field by $\val_\mu(\ulT)^j$ (or simply $\val(\ulT)^j$ when $\mu$ is clear from context).
    Note that if $\val(\ulT)^j = \xi^j$ and $\val(\ulS)^j = \Psi^j$, then $\val(\ulT\text{---}\ulS) = \int \xi^\top \Psi \,\,d\mu$.
\end{definition}

It will be desirable to manipulate quantities expressed directly in terms of the log-density, so we also introduce the following shorthand. 
\begin{definition} \label{def:main_pf:repr:multigraphs}
    For any fixed $\mu$, a connected multi-graph whose nodes are drawn as white circles should be interpreted as a tensor network where every node of degree $k$ represents the symmetric tensor $\nabla^k \log \mu$.
    The same convention applies for connected multi-graphs with dangling edges, where each dangling edge then represents a free index of the resulting tensor.
    In expressions involving such tensor networks, for readability, we will write subtraction as $\ominus$ instead of $-$.
\end{definition}

For example,
\begin{align*}
    \val\left(
        ~\begin{tikzpicture}[baseline=-0.5ex]
            \node[rectangle,fill=gray,draw] (A) at (0,0) {};
            \node (B) at (0.7,0) {};
            \draw (A) -- (B);
        \end{tikzpicture}
    \right)^j
    = 
    ~\begin{tikzpicture}[baseline=-0.5ex]
        \node[circle,draw] (A) at (0,0) {};
        \node (B) at (0.7,0) {$j$};
        \draw (A) -- (B);
    \end{tikzpicture}
    ~~
    &= \nabla_j \log \mu \\
    \text{and}~~~~
    \val\left(
        ~\begin{tikzpicture}[baseline=-0.5ex]
            \node[rectangle,fill=gray,draw] (A) at (0,0) {};
            \node[rectangle,fill=gray,draw] (B) at (1,0) {};
            \draw (A) -- (B);
        \end{tikzpicture}~
    \right)
    = 
    \int
    ~\begin{tikzpicture}[baseline=-0.5ex]
        \node[circle,draw] (A) at (0,0) {};
        \node[circle,draw] (B) at (1,0) {};
        \draw (A) -- (B);
    \end{tikzpicture}
    ~d\mu
    &= \int \norm{\nabla \log \mu}^2 \,d\mu.
\end{align*}

The following proposition allows to manipulate and rearrange the tensor diagram representations.
\begin{proposition}[Integration by parts] \label{prop:main_pf:repr:IPP}
    Consider a fixed $\mu$ and $G$ a connected multi-graph.
    Consider any node~$g$ and any edge~$i$ adjacent to it.

    Suppose~$i$ is not a self-loop and denote by~$h$ the other node adjacent to it.
    \begin{itemize}[nosep]
        \item Let~$G_g$ be the graph obtained from~$G$ by removing~$i$ and adding a self-loop to~$g$ (i.e., detaching~$i$ from~$h$ and reattaching the loose end to $g$).
        \item Let~$G_0$ be the graph obtained from~$G$ by adding a node, removing~$i$ and adding an edge between~$g$ and the new node (i.e., detaching~$i$ from~$h$ and reattaching the loose end to a new node).
        \item For every other node $f \not\in \{g, h\}$ of~$G$, let~$G_f$ be the graph obtained from~$G$ by removing~$i$ and adding an edge between~$g$ and~$f$ (i.e., detaching~$i$ from~$h$ and reattaching the loose end to $f$).
    \end{itemize}
    Then $\int G \,d\mu = -\int \left( G_g + G_0 + \sum_{f \not\in \{g,h\}} G_f \right) d\mu$.
    
    Now suppose $i$ is a self-loop.
    \begin{itemize}[nosep]
        \item Let~$G_0$ be the graph obtained from~$G$ by adding a node, removing~$i$ and adding an edge between~$g$ and the new node (i.e., detaching one end of~$i$ from~$g$ and reattaching it to a new node).
        \item For every other node $f \neq g$ of~$G$, let~$G_f$ be the graph obtained from~$G$ by removing~$i$ and adding an edge between~$g$ and~$f$ (i.e., detaching one end of~$i$ from~$g$ and reattaching it to $f$).
    \end{itemize}
    Then $\int G \,d\mu = -\int \left( G_0 + \sum_{f \neq g} G_f \right) d\mu$.
\end{proposition}

\begin{proof}
    For any tensor fields $f, g, h$, by integration by parts,
    \begin{multline*}
        \int (\nabla_i g^{\bm{j}\bm{k}}) (\nabla_i h^{\bm{j}\bm{l}}) f^{\bm{k}\bm{l}} \,d\mu
        = -\int h^{\bm{j}\bm{l}} \nabla_i \left[ (\nabla_i g^{\bm{j}\bm{k}}) f^{\bm{k}\bm{l}} \mu \right] \\
        = -\int h^{\bm{j}\bm{l}} \left[
            (\nabla^2_{ii} g^{\bm{j}\bm{k}}) f^{\bm{k}\bm{l}} d\mu 
            + (\nabla_i g^{\bm{j}\bm{k}}) (\nabla_i f^{\bm{k}\bm{l}}) d\mu
            + (\nabla_i g^{\bm{j}\bm{k}}) f^{\bm{k}\bm{l}} (\nabla_i \log \mu) d\mu
        \right].
    \end{multline*}
    The first part of the proposition statement is a translation of this computation when 
    $g, h$ are of the form $\nabla^k \log \mu$ and $f$ is a product thereof,
    using the tensor diagram shorthand introduced in \autoref{def:main_pf:repr:multigraphs}.

    Likewise, the second part of the proposition follows from the following computation: for any tensor fields $f, g$, by integration by parts,
    \begin{equation*}
        \int (\nabla^2_{ii} g^{\bm{j}}) f^{\bm{j}} \,d\mu
        = -\int (\nabla_i g^{\bm{j}}) \nabla_i \left[ f^{\bm{j}} \mu \right]
        = -\int (\nabla_i g^{\bm{j}}) \left[ 
            (\nabla_i f^{\bm{j}}) d\mu
            + f^{\bm{j}} (\nabla_i \log \mu) d\mu
        \right].
        \rqedhere
    \end{equation*}
\end{proof}

We will use several times the following fact.
\begin{lemma} \label{lm:main_pf:repr:val_rect_rect}
    We have, for $\mu \in \PPP_2(\RR^d)$ and denoting $U = \log \mu$,
    \begin{align*}
        \val\left(
            ~\begin{tikzpicture}[baseline=-0.5ex]
                \node[rectangle,fill=gray,draw] (A) at (0,0) {};
                \node[rectangle,fill=gray,draw] (B) at (1,0) {};
                \node (C) at (1.7,0) {};
                \draw (A) -- (B);
                \draw (B) -- (C);
            \end{tikzpicture}
        \right)^i
        &= -\nabla^3_{ikk} U - (\nabla^2_{ik} U) (\nabla_k U)
        = -\nabla_i \left[ \Delta U + \frac12 \norm{\nabla U}^2 \right] \\
        &= 
        \ominus \begin{tikzpicture}[baseline=-0.5ex]
            \node[circle,draw] (A) at (0,0) {};
            \path[every loop/.style={looseness=10, in=210, out=150}] (A) edge [loop] node {} (A);
            \node (B) at (0.7,0) {$i$};
            \draw (A) -- (B);
        \end{tikzpicture}
        \ominus \begin{tikzpicture}[baseline=-0.5ex]
            \node[circle,draw] (A) at (0,0) {};
            \node[circle,draw] (B) at (1,0) {};
            \node (C) at (1.7,0) {$i$};
            \draw (A) -- (B);
            \draw (B) -- (C);
        \end{tikzpicture} \\
        \text{and}~~~~
        \nabla_j~ 
        \val\left(
            ~\begin{tikzpicture}[baseline=-0.5ex]
                \node[rectangle,fill=gray,draw] (A) at (0,0) {};
                \node[rectangle,fill=gray,draw] (B) at (1,0) {};
                \node (C) at (1.7,0) {};
                \draw (A) -- (B);
                \draw (B) -- (C);
            \end{tikzpicture}
        \right)^i
        &=  
        \ominus \begin{tikzpicture}[baseline=-0.5ex]
            \node (Z) at (-0.7,0) {$i$};
            \node[circle,draw] (A) at (0,0) {};
            \path[every loop/.style={looseness=10, in=120, out=60}] (A) edge [loop] node {} (A);
            \node (B) at (0.7,0) {$j$};
            \draw (Z) -- (A);
            \draw (A) -- (B);
        \end{tikzpicture}
        \ominus \begin{tikzpicture}[baseline=-0.5ex]
            \node (Z) at (-0.7,0) {$i$};
            \node[circle,draw] (A) at (0,0) {};
            \node[circle,draw] (C) at (0,0.8) {};
            \node (B) at (0.7,0) {$j$};
            \draw (Z) -- (A);
            \draw (A) -- (B);
            \draw (A) -- (C);
        \end{tikzpicture}
        \ominus \begin{tikzpicture}[baseline=-0.5ex]
            \node (Z) at (-0.7,0) {$i$};
            \node[circle,draw] (A) at (0,0) {};
            \node[circle,draw] (B) at (1,0) {};
            \node (C) at (1.7,0) {$j$};
            \draw (Z) -- (A);
            \draw (A) -- (B);
            \draw (B) -- (C);
        \end{tikzpicture}.
    \end{align*}
\end{lemma}

\begin{proof}
    By definition (\autoref{def:main_pf:repr:expr_vect}),
    $
        \val\left(
            ~\begin{tikzpicture}[baseline=-0.5ex]
                \node[rectangle,fill=gray,draw] (A) at (0,0) {};
                \node[rectangle,fill=gray,draw] (B) at (1,0) {};
                \node (C) at (1.7,0) {};
                \draw (A) -- (B);
                \draw (B) -- (C);
            \end{tikzpicture}
        \right)^i
        = \Lambda^{(1,\mu)}_k \Lambda^{(2,\mu)}_{ki}
    $.
    Now for any vector field $\xi \in L^2_\mu \cap \mathfrak{X}(\RR^d)$, by definition (\autoref{def:gf:wass:Lambdan_i_shorthand}) and by integration by parts,
    \begin{align*}
        \int \xi^i \left( \Lambda^{(1,\mu)}_k \Lambda^{(2,\mu)}_{ki} \right) d\mu 
        = \Lambda^{(2,\mu)}_{ik}~ \xi^i (\nabla \log \mu)^k 
        &= \int \trace\left( (\nabla_k \xi^i) (\nabla_i (\nabla \log \mu)^k) \right) d\mu \\
        &= \int (\nabla_k \xi^i) (\nabla^2_{ik} \log \mu) \,d\mu
        = -\int \xi^i \,\nabla_k \left[ (\nabla^2_{ik} \log \mu) \mu \right].
    \end{align*}
    So
    \begin{equation*}
        \Lambda^{(1,\mu)}_k \Lambda^{(2,\mu)}_{ki}
        = -\frac{1}{\mu} \nabla_k \left[ (\nabla^2_{ik} \log \mu) \mu \right]
        = -\nabla^3_{ikk} \log \mu
        - (\nabla^2_{ik} \log \mu) (\frac{1}{\mu} \nabla_k \mu)
        = -\nabla^3_{ikk} U
        - (\nabla^2_{ik} U) (\nabla_k U)
    \end{equation*}
    where $U = \log \mu$.
    Moreover,
    \begin{align*}
        \nabla_j~ 
        \val\left(
            ~\begin{tikzpicture}[baseline=-0.5ex]
                \node[rectangle,fill=gray,draw] (A) at (0,0) {};
                \node[rectangle,fill=gray,draw] (B) at (1,0) {};
                \node (C) at (1.7,0) {};
                \draw (A) -- (B);
                \draw (B) -- (C);
            \end{tikzpicture}
        \right)^i
        &= \nabla_j \left[ -\nabla^3_{ikk} U - (\nabla^2_{ik} U) (\nabla_k U) \right] \\
        &= -\nabla^4_{ijkk} U - (\nabla^3_{ijk} U) (\nabla_k U) - (\nabla^2_{ik} U) (\nabla^2_{jk} U)
    \end{align*}
    as announced.
\end{proof}

\subsection{Proof of \autoref{prop:main_pf:n=4}}
For the remainder of this section, fix $\mu \in \PPP_2(\RR^d)$ log-concave and denote $U = \log \mu$.
Recall that $S_1, S_2$ and $S_3$ are as shown in \autoref{fig:main_pf:F4}.

\paragraph{Proof that \texorpdfstring{$\val(S_1) \geq 0$}{val(S1)>=0}.}
By definition,
\begin{equation*}
    \val(S_1)
    = \int \Lambda_4(\nabla U) \,d\mu
    = \int (-1)^4 (3!) \trace((\nabla^2 U)^4) \,\,d\mu
    = 6 \int \trace((\nabla^2 U)^4) \,d\mu
    \geq 0.
\end{equation*}

\paragraph{Proof that \texorpdfstring{$\val(S_2) \geq 0$}{val(S2)>=0}.}
By definition and by \autoref{lm:main_pf:repr:val_rect_rect},
\begin{align*}
    \val(S_2) 
    &= \int \Lambda_3\left(
        \nabla U,~
        \nabla U,~
        \val\left(
            ~\begin{tikzpicture}[baseline=-0.5ex]
                \node[rectangle,fill=gray,draw] (A) at (0,0) {};
                \node[rectangle,fill=gray,draw] (B) at (1,0) {};
                \node (C) at (1.7,0) {};
                \draw (A) -- (B);
                \draw (B) -- (C);
            \end{tikzpicture}
        \right)
    \right) d\mu \\
    &= (-1)^3 (2!) \int 
    \trace\left(
        (\nabla^2_{ij} U)
        (\nabla^2_{jk} U)
        \nabla_k \val\left(
            ~\begin{tikzpicture}[baseline=-0.5ex]
                \node[rectangle,fill=gray,draw] (A) at (0,0) {};
                \node[rectangle,fill=gray,draw] (B) at (1,0) {};
                \node (C) at (1.7,0) {};
                \draw (A) -- (B);
                \draw (B) -- (C);
            \end{tikzpicture}
        \right)^i
    \right) d\mu \\
    &= 2 \int \left(~
        \begin{tikzpicture}[baseline=-0.5ex]
            \node[circle,draw] (A) at (0, -0.5) {};
            \node[circle,draw] (B) at (0, 0.5) {};
            \node[circle,draw] (C) at (0.866, 0) {};
            \path[every loop/.style={looseness=10, in=30, out=-30}] (C) edge [loop] node {} (C);
            \draw (A) -- (B);
            \draw (A) -- (C);
            \draw (B) -- (C);
        \end{tikzpicture}
        ~+~
        \begin{tikzpicture}[baseline=-0.5ex,nodes={circle,draw}]
            \node (A) at (0, -0.5) {};
            \node (B) at (0, 0.5) {};
            \node (C) at (0.866, 0) {};
            \node (D) at (1.866, 0) {};
            \draw (A) -- (B);
            \draw (A) -- (C);
            \draw (B) -- (C);
            \draw (C) -- (D);
        \end{tikzpicture}
        ~+~
        \begin{tikzpicture}[baseline=-0.5ex,nodes={circle,draw}]
            \node (A) at (0, -0.5) {};
            \node (B) at (1, -0.5) {};
            \node (C) at (0, 0.5) {};
            \node (D) at (1, 0.5) {};
            \draw (A) -- (B);
            \draw (A) -- (C);
            \draw (B) -- (D);
            \draw (C) -- (D);
        \end{tikzpicture}
    ~\right) d\mu.
\end{align*}
The third term in the integrand is equal to $\trace((\nabla^2 U)^4) \geq 0$.
For the first two terms, an application of \autoref{prop:main_pf:repr:IPP} shows that
\begin{equation*}
    \int \left(~
        \begin{tikzpicture}[baseline=-0.5ex]
            \node[circle,draw] (A) at (0, -0.5) {};
            \node[circle,draw] (B) at (0, 0.5) {};
            \node[circle,draw] (C) at (0.866, 0) {};
            \path[every loop/.style={looseness=10, in=30, out=-30}] (C) edge [loop] node {} (C);
            \draw (A) -- (B);
            \draw (A) -- (C);
            \draw (B) -- (C);
        \end{tikzpicture}
        ~+~
        \begin{tikzpicture}[baseline=-0.5ex,nodes={circle,draw}]
            \node (A) at (0, -0.5) {};
            \node (B) at (0, 0.5) {};
            \node (C) at (0.866, 0) {};
            \node (D) at (1.866, 0) {};
            \draw (A) -- (B);
            \draw (A) -- (C);
            \draw (B) -- (C);
            \draw (C) -- (D);
        \end{tikzpicture}
    ~\right) d\mu
    = -\int ~2~ 
    \begin{tikzpicture}[baseline=-0.5ex,nodes={circle,draw}]
        \node (A) at (0, -0.5) {};
        \node (B) at (0, 0.5) {};
        \node (C) at (0.866, 0) {};
        \draw (A) -- (B);
        \draw (A) -- (C);
        \draw[double distance=2pt] (B) -- (C);
    \end{tikzpicture}
    ~d\mu.
\end{equation*}
Now
\begin{equation*}
    \begin{tikzpicture}[baseline=-0.5ex]
        \node[circle,draw] (A) at (0, -0.5) {};
        \node[circle,draw] (B) at (0, 0.5) {};
        \node[circle,draw] (C) at (0.866, 0) {};
        \draw (A) -- (B) node [midway,left] {~~~{\scriptsize $i$}};
        \draw (A) -- (C) node [midway,below] {~~~{\scriptsize $j$}};
        \draw[double distance=2pt] (B) -- (C) node [midway,above] {~~~{\scriptsize $k,l$}};
    \end{tikzpicture}
    = (\nabla^2_{ij} U) (\nabla^3_{ikl} U) (\nabla^3_{jkl} U)
    = \sum_{k,l} (\nabla_{ij}^2 U) \Psi^{(k,l)}_i \Psi^{(k,l)}_j
    = \sum_{k,l} \Psi^{(k,l) \top} (\nabla^2 U) \Psi^{(k,l)}
    \leq 0
\end{equation*}
since $U$ is concave, where $\Psi^{(k,l)}_i = \nabla^3_{ikl} U$.

\paragraph{Proof that \texorpdfstring{$\val(S_3) \geq 0$}{val(S3)>=0}.}
By definition and by \autoref{lm:main_pf:repr:val_rect_rect},
\begin{equation*}
    \val(S_3) 
    = \int \Lambda_2\big(
        \val\left(
            ~\begin{tikzpicture}[baseline=-0.5ex]
                \node[rectangle,fill=gray,draw] (A) at (0,0) {};
                \node[rectangle,fill=gray,draw] (B) at (1,0) {};
                \node (C) at (1.7,0) {};
                \draw (A) -- (B);
                \draw (B) -- (C);
            \end{tikzpicture}
        \right)
    \big) \,d\mu
    = \int \Lambda_2(\nabla \psi) \,d\mu
    = \int \trace\left( \left( \nabla^2 \psi \right)^2 \right) d\mu
    \geq 0
\end{equation*}
where $\psi = -\Delta U - \frac{1}{2} \norm{\nabla U}^2$.

\subsection{Proof of \autoref{prop:main_pf:n=5}}

Recall that we denote $U = \log \mu$ and that $T_1, ..., T_6$ are as shown in \autoref{fig:main_pf:F5}.

\paragraph{Proof that \texorpdfstring{$\val(T_1) \geq 0$}{val(T1)>=0}.}
By definition,
$
    \val(T_1)
    = \int 
    (-1)^5 (4!) \trace\left(\left( \nabla^2 U \right)^5 \right)
    d\mu
    \geq 0
$
since $U$ is concave.

\paragraph{Proof that \texorpdfstring{$\val(T_2) \geq 0$}{val(T2)>=0}.}
By definition,
$
    \val(T_2)
    = \int \Psi^\top \Psi \,d\mu
    \geq 0
$
where
$
    \Psi = \val\left(
        ~\begin{tikzpicture}[baseline=-0.5ex]
			\node[rectangle,fill=gray,draw] (A) at (-0.866, 0.5) {};
			\node[rectangle,fill=gray,draw] (B) at (-0.866, -0.5) {};
			\node[rectangle,fill=gray,draw] (C) at (0, 0) {};
			\node (D) at (0.7, 0) {};
			\draw (A) -- (C);
			\draw (B) -- (C);
			\draw (C) -- (D);
        \end{tikzpicture}
    \right)
$.

\paragraph{Proof that \texorpdfstring{$\val(T_3) \geq 0$}{val(T3)>=0}.}
By definition,
$
    \val(T_3)
    = \int \Psi^\top \Psi \,d\mu
    \geq 0
$
where
$
    \Psi = \val\left(
        ~\begin{tikzpicture}[baseline=-0.5ex]
			\node[rectangle,fill=gray,draw] (A) at (0, 0) {};
			\node[rectangle,fill=gray,draw] (B) at (1, 0) {};
			\node[rectangle,fill=gray,draw] (C) at (2, 0) {};
			\node (D) at (2.7, 0) {};
			\draw (A) -- (B);
			\draw (B) -- (C);
			\draw (C) -- (D);
        \end{tikzpicture}
    \right)
$.

\subsubsection{Proof that \texorpdfstring{$\val(T_4) \geq 0$}{val(T4)>=0}} \label{subsubsec:main_pf:n=5:T4}
By definition,
\begin{align*}
    \val(T_4) 
    &= \int \Lambda_4\left(
        \nabla U,~
        \nabla U,~
        \nabla U,~
        \val\left(
            ~\begin{tikzpicture}[baseline=-0.5ex]
                \node[rectangle,fill=gray,draw] (A) at (0,0) {};
                \node[rectangle,fill=gray,draw] (B) at (1,0) {};
                \node (C) at (1.7,0) {};
                \draw (A) -- (B);
                \draw (B) -- (C);
            \end{tikzpicture}
        \right)
    \right) d\mu \\
    &= (-1)^4 (3!) \int 
    \trace\left(
        (\nabla^2_{ij} U)
        (\nabla^2_{jk} U)
        (\nabla^2_{kl} U)
        \nabla_l \val\left(
            ~\begin{tikzpicture}[baseline=-0.5ex]
                \node[rectangle,fill=gray,draw] (A) at (0,0) {};
                \node[rectangle,fill=gray,draw] (B) at (1,0) {};
                \node (C) at (1.7,0) {};
                \draw (A) -- (B);
                \draw (B) -- (C);
            \end{tikzpicture}
        \right)^i
    \right) d\mu \\
    &= 6 \int \left(
        \ominus~ \begin{tikzpicture}[baseline=-0.5ex]
            \node[circle,draw] (A) at (0:0.7) {};
            \node[circle,draw] (B) at (90:0.7) {};
            \node[circle,draw] (C) at (180:0.7) {};
            \node[circle,draw] (D) at (-90:0.7) {};
            \path[every loop/.style={looseness=10, in=30, out=-30}] (A) edge [loop] node {} (A);
            \draw (A) -- (B);
            \draw (B) -- (C);
            \draw (C) -- (D);
            \draw (D) -- (A);
        \end{tikzpicture}
        ~\ominus~
        \begin{tikzpicture}[baseline=-0.5ex]]
            \node[circle,draw] (A) at (0:0.7) {};
            \node[circle,draw] (B) at (90:0.7) {};
            \node[circle,draw] (C) at (180:0.7) {};
            \node[circle,draw] (D) at (-90:0.7) {};
            \node[circle,draw] (E) at (1.4,0) {};
            \draw (A) -- (B);
            \draw (B) -- (C);
            \draw (C) -- (D);
            \draw (D) -- (A);
            \draw (A) -- (E);
        \end{tikzpicture}
        ~\ominus~
        \begin{tikzpicture}[baseline=-0.5ex]]
            \node[circle,draw] (A) at (0:0.7) {};
            \node[circle,draw] (B) at (72:0.7) {};
            \node[circle,draw] (C) at (144:0.7) {};
            \node[circle,draw] (D) at (216:0.7) {};
            \node[circle,draw] (E) at (-72:0.7) {};
            \draw (A) -- (B);
            \draw (B) -- (C);
            \draw (C) -- (D);
            \draw (D) -- (E);
            \draw (E) -- (A);
        \end{tikzpicture}
    ~\right) d\mu.
\end{align*}
The third term in the integrand is equal to $-\trace((\nabla^2 U)^5) \geq 0$ since $U$ is concave.
For the first two terms, an application of \autoref{prop:main_pf:repr:IPP} shows that
\begin{equation*}
    \int \left(
        \ominus~ \begin{tikzpicture}[baseline=-0.5ex]
            \node[circle,draw] (A) at (0:0.7) {};
            \node[circle,draw] (B) at (90:0.7) {};
            \node[circle,draw] (C) at (180:0.7) {};
            \node[circle,draw] (D) at (-90:0.7) {};
            \path[every loop/.style={looseness=10, in=30, out=-30}] (A) edge [loop] node {} (A);
            \draw (A) -- (B);
            \draw (B) -- (C);
            \draw (C) -- (D);
            \draw (D) -- (A);
        \end{tikzpicture}
        ~\ominus~
        \begin{tikzpicture}[baseline=-0.5ex]]
            \node[circle,draw] (A) at (0:0.7) {};
            \node[circle,draw] (B) at (90:0.7) {};
            \node[circle,draw] (C) at (180:0.7) {};
            \node[circle,draw] (D) at (-90:0.7) {};
            \node[circle,draw] (E) at (1.4,0) {};
            \draw (A) -- (B);
            \draw (B) -- (C);
            \draw (C) -- (D);
            \draw (D) -- (A);
            \draw (A) -- (E);
        \end{tikzpicture}
    ~\right) d\mu
    = \int \left(~
        \begin{tikzpicture}[baseline=-0.5ex]
            \node[circle,draw] (A) at (0:0.7) {};
            \node[circle,draw] (B) at (90:0.7) {};
            \node[circle,draw] (C) at (180:0.7) {};
            \node[circle,draw] (D) at (-90:0.7) {};
            \draw (A) -- (B);
            \draw (B) -- (C);
            \draw (C) -- (D);
            \draw (D) -- (A);
            \draw (A) -- (C);
        \end{tikzpicture}
        ~+~
        2 \cdot \begin{tikzpicture}[baseline=-0.5ex]]
            \node[circle,draw] (A) at (0:0.7) {};
            \node[circle,draw] (B) at (90:0.7) {};
            \node[circle,draw] (C) at (180:0.7) {};
            \node[circle,draw] (D) at (-90:0.7) {};
            \draw[double distance=2pt] (A) -- (B);
            \draw (B) -- (C);
            \draw (C) -- (D);
            \draw (D) -- (A);
        \end{tikzpicture}
    ~\right) d\mu.
\end{equation*}
Now~
$\begin{tikzpicture}[baseline=-0.5ex]]
    \node[circle,draw] (A) at (0:0.7) {};
    \node[circle,draw] (B) at (90:0.7) {};
    \node[circle,draw] (C) at (180:0.7) {};
    \node[circle,draw] (D) at (-90:0.7) {};
    \draw[double distance=2pt] (A) -- (B) node [near start,above] {~~~{\scriptsize $j,k$}};
    \draw (B) -- (C);
    \draw (C) -- (D) node [near start,below] {{\scriptsize $i$}};
    \draw (D) -- (A);
\end{tikzpicture}
= \sum_{ijk} g_{ijk} \cdot g_{ijk} \geq 0$ where $g_{ijk} = (\nabla^2_{is} U) (\nabla^3_{jks} U)$.
For the remaining term, denoting the eigenvalue decomposition of $\nabla^2 U$ by $\nabla^2_{ij} U = \sum_l \sigma_l u_{(l)}^i u_{(l)}^j$ and letting $h = \nabla^3 U$,
\begin{align*}
    \begin{tikzpicture}[baseline=-0.5ex]
        \node[circle,draw] (A) at (0:0.7) {};
        \node[circle,draw] (B) at (90:0.7) {};
        \node[circle,draw] (C) at (180:0.7) {};
        \node[circle,draw] (D) at (-90:0.7) {};
        \draw (A) -- (B) node [near start,above] {{\scriptsize $i$}};
        \draw (B) -- (C) node [near end,above] {{\scriptsize $j$}};
        \draw (C) -- (D) node [near start,below] {{\scriptsize $j'$}};
        \draw (D) -- (A) node [near end,below] {{\scriptsize $i'$}};
        \draw (A) -- (C) node [midway,above] {{\scriptsize $k$}};
    \end{tikzpicture}
    &= (\nabla^2_{ij} U) (\nabla^2_{i' j'} U) (\nabla^3_{ii'k} U) (\nabla^3_{jj'k} U)
    = \sum_{l,l'} 
    \left( \sigma_l u_{(l)}^i u_{(l)}^j \right)
    \left( \sigma_{l'} u_{(l')}^{i'} u_{(l')}^{j'} \right)
    h_{ii'k} h_{jj'k} \\
    &= \sum_{l,l'} \, \sigma_l \sigma_{l'}
    \,\sum_k~ 
    \underbrace{\left( \sum_{i,i'} u_{(l)}^i u_{(l')}^{i'} h_{ii'k} \right)}~
    \underbrace{\left( \sum_{j,j'} u_{(l)}^j u_{(l')}^{j'} h_{jj'k} \right)}
    \geq 0 \\[-0.9em]
\end{align*}
since the two underbraced quantities are equal for each $l,l',k$,
and $\sigma_l \sigma_{l'} \geq 0$ since $U$ is concave.

\subsubsection{Proof that \texorpdfstring{$\val(T_5) \geq 0$}{val(T5)>=0}}
Denote the tensor
$g_{ij} = \nabla_j \val\left(
~\begin{tikzpicture}[baseline=-0.5ex]
	\node[rectangle,fill=gray,draw] (A) at (0,0) {};
	\node[rectangle,fill=gray,draw] (B) at (1,0) {};
	\node (C) at (1.7,0) {};
	\draw (A) -- (B);
	\draw (B) -- (C);
\end{tikzpicture}
\right)^i$,
which is symmetric by \autoref{lm:main_pf:repr:val_rect_rect}.
Then by definition, and by concavity of $U$,
\begin{align*}
	& \val(T_5) = \Lambda^{(3,\mu)}~
	(\nabla U)~
	\val\left(
	~\begin{tikzpicture}[baseline=-0.5ex]
		\node[rectangle,fill=gray,draw] (A) at (0,0) {};
		\node[rectangle,fill=gray,draw] (B) at (1,0) {};
		\node (C) at (1.7,0) {};
		\draw (A) -- (B);
		\draw (B) -- (C);
	\end{tikzpicture}
	\right)~
	\val\left(
	~\begin{tikzpicture}[baseline=-0.5ex]
		\node[rectangle,fill=gray,draw] (A) at (0,0) {};
		\node[rectangle,fill=gray,draw] (B) at (1,0) {};
		\node (C) at (1.7,0) {};
		\draw (A) -- (B);
		\draw (B) -- (C);
	\end{tikzpicture}
	\right) \\
	&~~ = \int d\mu~ (-1)^3 (2!) \trace\left(
	(\nabla^2 U)\, g\, g
	\right)
	= -2 \int d\mu~ (\nabla^2_{ij} U) g_{jk} g_{ki} 
	= -2 \sum_k \int d\mu~ g_{\bullet k}^\top (\nabla^2 U) g_{\bullet k}
	\geq 0.
\end{align*}

\subsubsection{Proof that \texorpdfstring{$\val(T_6) \geq 0$}{val(T6)>=0}}
\begin{claim}
    We have
    \begin{align*}
        \val\left(
            ~\begin{tikzpicture}[baseline=-0.5ex]
    			\node[rectangle,fill=gray,draw] (A) at (0, 0) {};
    			\node[rectangle,fill=gray,draw] (B) at (1, 0) {};
    			\node[rectangle,fill=gray,draw] (C) at (2, 0) {};
    			\node (D) at (2.7, 0) {};
    			\draw (A) -- (B);
    			\draw (B) -- (C);
    			\draw (C) -- (D);
            \end{tikzpicture}
        \right)^i
        &= ~
        \begin{tikzpicture}[baseline=-0.5ex]
            \node[circle,draw] (A) at (0,0) {};
            \node[circle,draw] (B) at (-0.866, 0.5) {};
            \node[circle,draw] (C) at (-0.866, -0.5) {};
            \node (D) at (0.7, 0) {};
            \draw (A) -- (B);
            \draw (A) -- (C);
            \draw (A) -- (D);
        \end{tikzpicture}
        ~+~
        \begin{tikzpicture}[baseline=-0.5ex]
            \node[circle,draw] (A) at (0,0) {};
            \path[every loop/.style={looseness=10, in=210, out=150, double distance=2pt}] (A) edge [loop] node {} (A);
            \node (B) at (0.7, 0) {};
            \draw (A) -- (B);
        \end{tikzpicture}
        ~+~ 2 \cdot
        \begin{tikzpicture}[baseline=-0.5ex]
            \node[circle,draw] (A) at (0,0) {};
            \path[every loop/.style={looseness=10, in=90, out=150}] (A) edge [loop] node {} (A);
            \node[circle,draw] (B) at (-0.866, -0.5) {};
            \node (C) at (0.7, 0) {};
            \draw (A) -- (B);
            \draw (A) -- (C);
        \end{tikzpicture}
        ~+~ 2 \cdot 
        \begin{tikzpicture}[baseline=-0.5ex]
            \node[circle,draw] (A) at (0,0) {};
            \node[circle,draw] (B) at (-1,0) {};
            \node (Z) at (0.7, 0) {};
            \draw[double distance=2pt] (B) -- (A);
            \draw (A) -- (Z);
        \end{tikzpicture} \\
        &~~~~ +~
        \begin{tikzpicture}[baseline=-0.5ex]
            \node[circle,draw] (A) at (0,0) {};
            \node[circle,draw] (B) at (-1,0) {};
            \node (Z) at (0.7, 0) {};
            \path[every loop/.style={looseness=10, in=210, out=150}] (B) edge [loop] node {} (B);
            \draw (B) -- (A);
            \draw (A) -- (Z);
        \end{tikzpicture}
        ~+~
        \begin{tikzpicture}[baseline=-0.5ex]
            \node[circle,draw] (A) at (0,0) {};
            \node[circle,draw] (B) at (-1,0) {};
            \node[circle,draw] (C) at (-2,0) {};
            \node (Z) at (0.7, 0) {};
            \draw (C) -- (B);
            \draw (B) -- (A);
            \draw (A) -- (Z);
        \end{tikzpicture} \\
        \text{and}~~
        \val\left(
            ~\begin{tikzpicture}[baseline=-0.5ex]
    			\node[rectangle,fill=gray,draw] (A) at (-0.866, 0.5) {};
    			\node[rectangle,fill=gray,draw] (B) at (-0.866, -0.5) {};
    			\node[rectangle,fill=gray,draw] (C) at (0, 0) {};
    			\node (D) at (0.7, 0) {};
    			\draw (A) -- (C);
    			\draw (B) -- (C);
    			\draw (C) -- (D);
            \end{tikzpicture}
        \right)^i
        &= 2 \cdot \left(
        \begin{tikzpicture}[baseline=-0.5ex]
            \node[circle,draw] (A) at (0,0) {};
            \node[circle,draw] (B) at (-1,0) {};
            \node (Z) at (0.7, 0) {};
            \draw[double distance=2pt] (B) -- (A);
            \draw (A) -- (Z);
        \end{tikzpicture}
        ~+~
        \begin{tikzpicture}[baseline=-0.5ex]
            \node[circle,draw] (A) at (0,0) {};
            \node[circle,draw] (B) at (-1,0) {};
            \node (Z) at (0.7, 0) {};
            \path[every loop/.style={looseness=10, in=210, out=150}] (B) edge [loop] node {} (B);
            \draw (B) -- (A);
            \draw (A) -- (Z);
        \end{tikzpicture}
        ~+~
        \begin{tikzpicture}[baseline=-0.5ex]
            \node[circle,draw] (A) at (0,0) {};
            \node[circle,draw] (B) at (-1,0) {};
            \node[circle,draw] (C) at (-2,0) {};
            \node (Z) at (0.7, 0) {};
            \draw (C) -- (B);
            \draw (B) -- (A);
            \draw (A) -- (Z);
        \end{tikzpicture}
        \right).
    \end{align*}
\end{claim}

\begin{proof}
    The results follow from the definition (\autoref{def:main_pf:repr:expr_vect}) by proceeding similarly as for the proof of \autoref{lm:main_pf:repr:val_rect_rect}.
    The computations are omitted due to space.
\end{proof}

\begin{claim}
    We have $\val(T_6) = 2 \int (A + B + C) \,d\mu = 2 \int (A' + B + C) \,d\mu$ where
    \begin{align*}
        A &= ~\ominus~ 2 \cdot \left(
            \begin{tikzpicture}[baseline=1.5ex,nodes={circle,draw,fill=white}]
                \node (A) at (0, 0.866) {};
                \node (B) at (-0.5, 0) {};
                \node (C) at (0.5, 0) {};
                \draw (A) -- (B);
                \draw (C) -- (A);
                \draw (B) -- ++(0,-0.1) -- ++(1,0) -- ++(0,0.1);
                \draw (B) -- ++(1,0);
                \draw (B) -- ++(0,0.1) -- ++(1,0) -- ++(0,-0.1);
                \node at (B) {};
                \node at (C) {};
            \end{tikzpicture}
            ~+~
            \begin{tikzpicture}[baseline=1.5ex,nodes={circle,draw}]
                \node (A) at (0, 0.866) {};
                \node (B) at (-0.5, 0) {};
                \node (C) at (0.5, 0) {};
                \draw[double distance=2pt] (A) -- (B);
                \draw (B) -- (C);
                \draw[double distance=2pt] (C) -- (A);
            \end{tikzpicture}
        \right)
        \\
        \text{and}~~~~
        A' &= 
        -2 (\nabla^2_{ij} U) \, g_{ik} \, g_{jk}
        \\
        &~~ +~ 4 \cdot
        \begin{tikzpicture}[baseline=1.5ex,nodes={circle,draw}]
            \node (A) at (0, 0.866) {};
            \node (B) at (-0.5, 0) {};
            \node (C) at (0.5, 0) {};
            \draw[double distance=2pt] (A) -- (B);
            \draw (B) -- (C);
            \draw[double distance=2pt] (C) -- (A);
        \end{tikzpicture}
        ~\ominus~ 
        2 \cdot \begin{tikzpicture}[baseline=-0.5ex]
            \node[circle,draw] (A) at (0:0.7) {};
            \node[circle,draw] (B) at (90:0.7) {};
            \node[circle,draw] (C) at (180:0.7) {};
            \node[circle,draw] (D) at (-90:0.7) {};
            \draw (A) -- (B);
            \draw (B) -- (C);
            \draw (C) -- (D);
            \draw (D) -- (A);
            \draw (A) -- (C);
        \end{tikzpicture} 
        ~~~~\text{where}~~
        g_{ik} = (\nabla_l U) (\nabla_{ikl}^3 U) + \nabla_{ikll}^4 U \\
        \text{and}~~~~
        B &= 4 \cdot 
        \begin{tikzpicture}[baseline=-0.5ex]]
            \node[circle,draw] (A) at (0:0.7) {};
            \node[circle,draw] (B) at (90:0.7) {};
            \node[circle,draw] (C) at (180:0.7) {};
            \node[circle,draw] (D) at (-90:0.7) {};
            \draw[double distance=2pt] (A) -- (B);
            \draw (B) -- (C);
            \draw (C) -- (D);
            \draw (D) -- (A);
        \end{tikzpicture}
        ~+~ 3 \cdot
        \begin{tikzpicture}[baseline=-0.5ex]
            \node[circle,draw] (A) at (0:0.7) {};
            \node[circle,draw] (B) at (90:0.7) {};
            \node[circle,draw] (C) at (180:0.7) {};
            \node[circle,draw] (D) at (-90:0.7) {};
            \draw (A) -- (B);
            \draw (B) -- (C);
            \draw (C) -- (D);
            \draw (D) -- (A);
            \draw (A) -- (C);
        \end{tikzpicture}
        ~\ominus~ 
        \begin{tikzpicture}[baseline=-0.5ex]]
            \node[circle,draw] (A) at (0:0.7) {};
            \node[circle,draw] (B) at (72:0.7) {};
            \node[circle,draw] (C) at (144:0.7) {};
            \node[circle,draw] (D) at (216:0.7) {};
            \node[circle,draw] (E) at (-72:0.7) {};
            \draw (A) -- (B);
            \draw (B) -- (C);
            \draw (C) -- (D);
            \draw (D) -- (E);
            \draw (E) -- (A);
        \end{tikzpicture} \\
        \text{and}~~~~
        C &= 2 \cdot
        \begin{tikzpicture}[baseline=-0.5ex]
            \node[circle,draw] (A) at (0:0.7) {};
            \node[circle,draw] (B) at (90:0.7) {};
            \node[circle,draw] (C) at (180:0.7) {};
            \node[circle,draw] (D) at (-90:0.7) {};
            \draw (A) -- (B);
            \draw (B) -- (C);
            \draw (C) -- (D);
            \draw (D) -- (A);
            \draw (A) -- (C);
        \end{tikzpicture}~.
    \end{align*}
\end{claim}

\begin{proof}
    The result follows from the previous claim and from repeated applications of \autoref{prop:main_pf:repr:IPP} (integrations by parts).
    The computations are omitted due to space.
\end{proof}

We showed previously in \autoref{subsubsec:main_pf:n=5:T4} that the terms $B$ and $C$ are non-negative. Moreover observe that, denoting $h = \nabla^4 U$, by concavity of $U$,
\begin{equation*}
    \begin{tikzpicture}[baseline=1.5ex,nodes={circle,draw,fill=white}]
        \node (A) at (0, 0.866) {};
        \node (B) at (-0.5, 0) {};
        \node (C) at (0.5, 0) {};
        \draw (A) -- (B);
        \draw (C) -- (A);
        \draw (B) -- ++(0,-0.1) -- ++(1,0) -- ++(0,0.1);
        \draw (B) -- ++(1,0);
        \draw (B) -- ++(0,0.1) -- ++(1,0) -- ++(0,-0.1);
        \node at (B) {};
        \node at (C) {};
    \end{tikzpicture}
    = (\nabla^2_{ij} U) h_{iabc} h_{jabc}
	= \sum_{a,b,c} h_{\bullet a b c}^\top (\nabla^2 U) h_{\bullet a b c}
    \leq 0.
\end{equation*}
Now distinguish two cases.
\begin{itemize}
    \item First suppose 
    $\int \begin{tikzpicture}[baseline=1.5ex,nodes={circle,draw}]
        \node (A) at (0, 0.866) {};
        \node (B) at (-0.5, 0) {};
        \node (C) at (0.5, 0) {};
        \draw[double distance=2pt] (A) -- (B);
        \draw (B) -- (C);
        \draw[double distance=2pt] (C) -- (A);
    \end{tikzpicture}
    \,d\mu
    \leq 0$;
    then we can conclude that $\int A \,d\mu \geq 0$ and so $\val(T_6) \geq 0$.
    \item Now suppose 
    $\int \begin{tikzpicture}[baseline=1.5ex,nodes={circle,draw}]
        \node (A) at (0, 0.866) {};
        \node (B) at (-0.5, 0) {};
        \node (C) at (0.5, 0) {};
        \draw[double distance=2pt] (A) -- (B);
        \draw (B) -- (C);
        \draw[double distance=2pt] (C) -- (A);
    \end{tikzpicture}
    \,d\mu
    > 0$.
    The first term in the definition of $A'$ is non-negative since $U$ is concave, 
    the second term is non-negative by the assumption,
    and the third term is precisely $-C$.
    So $\int (A' + C) \,d\mu \geq 0$, and we also conclude that $\val(T_6) \geq 0$.
\end{itemize}
This concludes the proof of \autoref{prop:main_pf:n=5}, and so of \autoref{thm:intro:GCMC_45}.

\vspace{1em}
\paragraph{Acknowledgments.}
I wish to thank Christopher Criscitiello for insightful discussions, and Tomas Vaškevičius for suggestions that improved the presentation of this work.

\printbibliography
\addcontentsline{toc}{section}{\refname}

\ifextended%
\newpage
\appendix
\phantomsection
\addcontentsline{toc}{section}{APPENDIX}

\section{Background on Otto calculus} \label{sec:apx_bg_otto}

The purpose of this appendix is purely to provide the reader unfamiliar with Otto calculus with some background. It can safely be skipped entirely. None of the notations of this appendix are reused elsewhere.

Wasserstein gradient flows (WGFs) are defined as follows \cite{ambrosio_gradient_2008}.
\begin{itemize}
    \item On the one hand, given any metric space $(\XXX, \dist)$ and any $f: \XXX \to \RR$ (with some topological and continuity assumptions), one can define a notion of gradient flow via,
    informally, $x(t) = \lim_{h \to 0} x_{\floor{t/h}}$ where $x_{k+1} = \argmin_\XXX f(\cdot) + \frac{1}{2h} \dist^2(\cdot, x_k)$. Note that there is no notion of actual gradient a priori here, and the denomination ``gradient flow'' is only motivated by analogy with the Euclidean case.
    (To be precise, there are several natural generalizations of gradient flow to metric spaces, and the book \cite{ambrosio_gradient_2008} carefully details the connections and equivalences between them.)
    \item On the other hand, one can show that $W_2$ (the optimal transport cost functional with transport cost $c(x,y) = \norm{x-y}_2^2$) satisfies the axioms of a metric function on $\PPP_2(\RR^d) = \left\{ \mu \in \PPP(\RR^d); \int \norm{x}^2 d\mu(x) < \infty \right\}$.
    So we define WGFs as the gradient flows (in the sense above) for the metric space $(\PPP_2(\RR^d), W_2)$.
    Let us emphasize that, in the end, the definition of $W_2$ as an optimal transport cost will not really matter in this appendix; its characterization by the Benamou-Brenier formula, recalled shortly, will be more practical.
\end{itemize}

A powerful guiding principle for the analysis of WGFs is given by Otto calculus, a non-rigorous analogy according to which the space $(\PPP_2(\RR^d), W_2)$ behaves very similarly to a Riemannian manifold.
In the remainder of this section, we describe the Otto calculus analogy in detail, assuming some familiarity with first-order Riemannian geometry (at the level of, e.g., \cite[Chapter~3]{boumal_introduction_2023}).
The reader is advised that all regularity assumptions are omitted, however, so some of the statements are mathematically false as stated; the focus is on the formal analogies.

\subsection{The Riemannian world} \label{subsec:apx_bg_otto:riemannian_world}

\begin{enumerate}[label=($\roman*$).]
    \item Let $(\MMM, g)$ a Riemannian manifold and $f: \MMM \to \RR$ (assumed sufficiently regular). 
    \item Denote by $T_x \MMM$, $T_x^* \MMM$ the tangent resp.\ cotangent spaces of $\MMM$ at $x$. That is, $T_x \MMM$ is defined as the linear space $\left\{ \gamma'(0); \gamma: (-1, 1) \to \MMM, \gamma(0) = x \right\}$, and $T_x^* \MMM$ is its (algebraic) dual space (i.e., the set of all linear forms on $T_x \MMM$).

    The differential of $f$ at $x$ is $D_x f \in T_x^*\MMM$ defined by
    \begin{equation*}
        \forall \gamma: (-1,1) \to \MMM ~~\text{s.t.}~~ \gamma(0) = x ~\text{and}~ \gamma'(0) = v,~~~~
        (D_x f)(v) = \lim_{\eps \to 0} \frac{1}{\eps} \left[ f(\gamma(\eps)) - f(x) \right].
    \end{equation*}
    \item For any $x$, the metric $g_x$ induces a bijection from $T_x \MMM$ to $T_x^* \MMM$ by
    \begin{equation*}
        \forall v, w \in T_x \MMM,~
        (g_x v)(w) = \innerprod{v}{w}_x.
    \end{equation*}
    Naturally, $(T_x \MMM, g_x)$ and $(T_x^*\MMM, g_x^{-1})$ are isometric inner-product spaces.
    It is common to use the same symbol $\innerprod{\cdot}{\cdot}_x$ for both inner products, i.e., $\innerprod{\cdot}{\cdot}_x$ can stand for $g_x(\cdot, \cdot)$ or for $g_x^{-1}(\cdot, \cdot)$ depending on the context.

    The Riemannian distance can be characterized as
    \begin{equation*}
        \dist_\MMM^2(x_0, x_1) = \inf_{\gamma: [0, 1] \to \MMM} \int_0^1 \norm{\gamma'(t)}_{\gamma(t)}^2 dt
        \subjto
        \gamma(0) = x_0 ~~\text{and}~~ \gamma(1) = x_1.
    \end{equation*}
    \item The Riemannian gradient of $f$ at $x$ is $\grad_x f = g_x^{-1} D_x f \in T_x\MMM$.
    
    Let $x(0) \in \MMM$. The Riemannian gradient flow is the unique solution of the ODE
    \begin{equation*}
        \frac{d}{dt} x(t) = -\grad_{x(t)} f
        ~~~~\text{with initial condition}~~~~ x(0).
    \end{equation*}
\end{enumerate}

\subsection{The Wasserstein world: first formalism} 
\label{subsec:apx_bg_otto:Wasserstein_villanifigalli}

Otto calculus comes in two common equivalent formalisms.
We start by the one used in the book \cite{figalli_invitation_2021}, as it is arguably the simplest one conceptually.
\begin{enumerate}[label=($\roman*'$).]
    \item Consider the metric space $(\PPP_2(\RR^d), W_2)$ and $F: \PPP_2(\RR^d) \to \RR$ (assumed sufficiently regular).
    For the sake of concision, we will also write $\PPP_2 \coloneqq \PPP_2(\RR^d)$.
    \item For any $\mu \in \PPP_2$, denote by $\support(\mu) \subset \RR^d$ the support of $\mu$, and define
    \begin{align*}
        T_\mu \PPP_2 &= \left\{
            -\nabla \cdot (\mu \nabla \phi),~
            \phi: \support(\mu) \to \RR
        \right\} \\
        T_\mu^* \PPP_2 &= \left\{
            \phi,~ \phi: \support(\mu) \to \RR
        \right\} / \RR.
    \end{align*}
    A curve of measures $(\mu_t)_{t \in \RR}$ in $\PPP_2$ is called \emph{absolutely continuous in the $W_2$ sense} if it is differentiable in time in the distributional sense and if, for all $t \in \RR$, its \emph{velocity} $\partial_t \mu_t$ is an element of $T_{\mu_t} \PPP_2$ (and if it satisfies certain regularity conditions which we omit here).
    In particular, all WGFs are absolutely continuous curves in the $W_2$ sense.
    
    The analog of the ``differential'' of $F$ at $\mu$ is the first variation, or rather its restriction to $\support(\mu)$, $\restr{F'[\mu]}{\support(\mu)} \in T_\mu^* \PPP_2$.
    Note that the quotient in the definition of $T_\mu^* \PPP_2$ reflects the fact that $F'[\mu]: \RR^d \to \RR$ is defined up to an arbitrary additive constant.
    
    A natural algebraic duality ``pairing'' of $T_\mu \PPP_2$ and $T_\mu^* \PPP_2$ is given simply by
    \begin{equation*}
        T_\mu \PPP_2 \times T_\mu^* \PPP_2
        \ni
        \left( -\nabla \cdot (\mu \nabla \phi), \psi \right)
        ~\mapsto~
        \int_{\RR^d} \left[ -\nabla \cdot (\mu \nabla \phi) \right] \, \psi
        = \int_{\RR^d} \nabla \phi(x)^\top \nabla \psi(x) \,d\mu(x)
        \in \RR.
    \end{equation*}
    \item For any $\mu \in \PPP_2$, define
    \begin{equation*}
        \olg_\mu^{-1}: T_\mu^* \PPP_2 \to T_\mu \PPP_2,~
        \phi \mapsto -\nabla \cdot (\mu \nabla \phi).
    \end{equation*}
    The operator $\olg_\mu^{-1}$ is invertible as a consequence of classical results on elliptic PDEs (see footnote~25 of \cite{figalli_invitation_2021}), so we can also define $\olg_\mu = (\olg_\mu^{-1})^{-1}$, which is a bijection from $T_\mu \PPP_2$ to $T_\mu^* \PPP_2$.
    Define a formal inner product on $T_\mu^* \PPP_2$ by
    \begin{equation*}
        \innerprod{\phi}{\psi}_\mu 
        \coloneqq \int_{\RR^d} \phi \, \left( \olg_\mu^{-1} \psi \right)
        = \int_{\RR^d} \nabla \phi(x)^\top \nabla \psi(x) \,d\mu(x),
    \end{equation*}
    and a formal inner product on $T_\mu \PPP_2$ by $\innerprod{-\nabla \cdot (\mu \nabla \phi)}{-\nabla \cdot (\mu \nabla \psi)}_\mu = \innerprod{\phi}{\psi}_\mu$---equivalently, it is defined precisely such that $(T_\mu \PPP_2, \innerprod{\cdot}{\cdot}_\mu)$ and $(T_\mu^* \PPP_2, \innerprod{\cdot}{\cdot}_\mu)$ are isometric via the bijection~$\olg_\mu$.

    With these notations, the Benamou-Brenier theorem asserts that the $2$-Wasserstein distance is characterized by
    \begin{equation} \label{eq:apx_bg_otto:wasserstein_villanifigalli:benamou_brenier}
        W_2^2(\ol\mu_0, \ol\mu_1) = \inf_{(\mu_t, \phi_t)_t} \int_0^1 \norm{-\nabla \cdot (\mu_t \nabla \phi_t)}_{\mu_t}^2 dt
        \subjto
        \begin{cases}
            \forall t, \phi_t \in T_{\mu_t}^* \PPP_2, \\
            \partial_t \mu_t = -\nabla \cdot (\mu_t \nabla \phi_t), \\
            \mu_0 = \ol\mu_0 ~~\text{and}~~ \mu_1 = \ol\mu_1.
        \end{cases}
    \end{equation}
    \item Thanks to the Benamou-Brenier theorem,
    one can show that the abstractly-defined gradient flows for the metric space $(\PPP_2, W_2)$ actually have the following form.
    Define the Wasserstein gradient of a functional $F: \PPP_2(\RR^d) \to \RR$ (assumed sufficiently regular) at $\mu$ as
    $-\nabla \cdot (\mu \nabla F'[\mu]) \in T_\mu \PPP_2$.
    Then the WGF $(\mu_t)_{t \geq 0}$ of $F$ is the unique distributional solution of the PDE
    \begin{equation} \label{eq:intro:WGF_PDE}
        \partial_t \mu_t = +\nabla \cdot \left( \mu_t \nabla F'[\mu_t] \right)
        ~~~~\text{with initial condition}~~~~ \mu_0.
    \end{equation}
    In other words, the WGF $(\mu_t)_t$ is the unique absolutely continuous curve in the $W_2$ sense whose velocity at each $t$ is minus the Wasserstein gradient of $F$ at $\mu_t$.
\end{enumerate}

\vspace{1em}
Note that for $F(\mu) = H(\mu) = \int_{\RR^d} d\mu \log \frac{d\mu}{dx}$, we have $H'[\mu](x) = \log \frac{d\mu}{dx} + 1$, so its WGF reads
\begin{equation*}
	\partial_t \mu_t
	= \nabla \cdot (\mu_t \nabla \log \mu_t)
	= \nabla \cdot (\nabla \mu_t)
	= \Delta \mu_t.
\end{equation*}
Thus the heuristics presented here indeed recover the fact that the WGF of $H$ is precisely the heat flow \cite{jordan_variational_1998}.

\subsection{The Wasserstein world: second formalism} 
\label{subsec:apx_bg_otto:wasserstein_ambrosiogigli}

\begin{enumerate}[label=($\roman*''$).]
    \item Consider the metric space $(\PPP_2, W_2)$ and $F: \PPP_2 \to \RR$ (assumed sufficiently regular),
    where $\PPP_2 = \PPP_2(\RR^d)$.
    \item For any $\mu \in \PPP_2$, define
    \begin{equation*}
        T_\mu \PPP_2 = \left\lbrace
            \nabla \phi,~ \phi: \support(\mu) \to \RR
        \right\rbrace.
    \end{equation*}
    A curve of measures $(\mu_t)_{t \in \RR}$ in $\PPP_2$ is called \emph{absolutely continuous in the $W_2$ sense} if, for each $t$, there exists a function $\phi_t$ such that $\partial_t \mu_t = -\nabla \cdot (\mu_t \nabla \phi_t)$ (and if it satisfies certain regularity conditions which we omit here).
    In this case, $\phi_t$ is uniquely defined, and $\dot{\mu}_t \coloneqq \nabla \phi_t \in T_{\mu_t} \PPP_2$ is called the \emph{Wasserstein velocity} of the curve at time $t$.

    We do not introduce any notion of cotangent space, and there is no natural analog of the ``differential'' of $F$ at $\mu$.%
    \footnote{
        One could easily introduce a notion of cotangent space, even though it is typically not done because there is no real reason to do so.
        Namely, one could define the set
        $
            T_\mu^* \PPP_2 = \left\{ (\nabla \phi) \mu, ~ \phi: \support(\mu) \to \RR \right\}
        $
        consisting of vector-valued measures over $\RR^d$,
        with the natural algebraic duality ``pairing''
        \begin{equation*}
            T_\mu \PPP_2 \times T_\mu^* \PPP_2
            \ni
            \left( \nabla \phi, (\nabla \psi) \mu \right) 
            ~\mapsto~
            \int_{\RR^d} \nabla \phi \cdot (\nabla \psi) \,\mu
            = \int_{\RR^d} \nabla \phi(x)^\top \nabla \psi(x) \,d\mu(x)
            \in \RR,
        \end{equation*}
        and define a bijection $\olg_\mu: T_\mu \PPP_2 \to T_\mu^* \PPP_2$ simply by $\olg_\mu \nabla \phi = (\nabla \phi) \mu$.
    }
    \item For any $\mu \in \PPP_2$, define a formal inner product on $T_\mu \PPP_2$ by
    \begin{equation*}
        \innerprod{\nabla \phi}{\nabla \psi}_\mu \coloneqq \int_{\RR^d} \nabla \phi(x)^\top \nabla \psi(x) \,d\mu(x).
    \end{equation*}
    With these definitions, the Benamou-Brenier theorem asserts that the $2$-Wasserstein distance is characterized by
    \begin{equation} \label{eq:apx_bg_otto:wasserstein_ambrosiogigli:benamou_brenier}
        W_2^2(\ol\mu_0, \ol\mu_1) = \inf_{(\mu_t)_{t \in [0,1]}} \int_0^1 \norm{\dot{\mu}_t}_{\mu_t}^2 dt
        \subjto
        \mu_0 = \ol\mu_0 ~\text{and}~ \mu_1 = \ol\mu_1,
    \end{equation}
    where the infimum ranges over absolutely continuous curves in the $W_2$ sense, and where $\dot{\mu}_t \in T_{\mu_t} \PPP_2$ denotes the Wasserstein velocity.
    \item Following this formalism, the Wasserstein gradient of $F$ at $\mu$ is defined as
    $\nabla F'[\mu] \in T_\mu \PPP_2$,
    which is consistent with common usage in the machine learning literature.
    
    Let $\mu_0 \in \PPP_2$. The Wasserstein gradient flow of $F$ is the unique curve such that
    \begin{equation*}
        \forall t,~ \dot{\mu_t} = -\nabla F'[\mu_t]
        ~~~~\text{with initial condition}~~~~ \mu_0.
    \end{equation*}
    In words, the WGF $(\mu_t)_t$ is the unique absolutely continuous curve in the $W_2$ sense whose Wasserstein velocity at each $t$ is minus the Wasserstein gradient of $F$ at $\mu_t$.
\end{enumerate}

In a sense, by defining velocity as $\dot{\mu}_t = \nabla \phi_t$ instead of $\partial_t \mu_t$, this second formalism completely does away with the ``additive'' geometry of $\PPP_2(\RR^d)$ inherited from its embedding in the space of signed measures.

\subsection{The Wasserstein geodesic equation} \label{subsec:apx_bg_otto:wass_geodesic}

The dual optimality condition of a certain convex reformulation of the Benamou-Brenier formula
\eqref{eq:apx_bg_otto:wasserstein_villanifigalli:benamou_brenier}
shows that the infimum is attained at curves 
$(\mu_t, \phi_t)_t$ such that \cite[Section~6.1]{santambrogio_optimal_2015}
\begin{equation} \label{eq:intro:wass_geo}
    \begin{cases}
        \partial_t \mu_t = -\nabla \cdot (\mu_t \nabla \phi_t) \\
        \partial_t \phi_t = -\frac{1}{2} \norm{\nabla \phi_t}^2.
    \end{cases}
\end{equation}
This system is called the \emph{(constant-speed) Wasserstein geodesic equation}.

\section{General facts about the Bell polynomials} \label{sec:apx_bell}

We start by setting notations for the coefficients of the Bell polynomials, which will be used throughout this section and only in this section.
\begin{definition} \label{def:apx_bell:coeffs}
    For any $n, k \geq 1$ and any sequence of integers $\bh = (h_n)_{n \geq 1} \in \NN^{\NN^*}$, let
    \begin{equation*}
        \alpha^{n,k}_{\bh}
        = \frac{n!}{h_1!~ h_2!~ h_3!~ ...}~ \frac{1}{(1!)^{h_1}\, (2!)^{h_2}\, (3!)^{h_3}\, ...}
    \end{equation*}
    if $h_1 + h_2 + h_3 + ... = k$ and $h_1 + 2 h_2 + 3 h_3 + ... = n$, and $0$ otherwise.
    The (partial exponential) Bell polynomial of order $(n,k)$, denoted by $B_{n,k}$, is defined as
    \begin{equation*}
        B_{n,k}(X_1, ..., X_{n-k+1}) = \sum_{\bh \in \NN^{\NN^*}} \alpha^{n,k}_{\bh}~ X_1^{h_1} \,...\, X_{n-k+1}^{h_{n-k+1}}.
    \end{equation*}
    Even though $B_{n,k}$ only has $n-k+1$ indeterminates, we will also write
    $B_{n,k}(X_1, X_2, X_3, ...)$ for $B_{n,k}(X_1, ..., X_{n-k+1})$.
    By convention,
    $B_{0,0} = 1$,
    $B_{n,0} = B_{0,k} = 0$,
    and $B_{n,k}=0$ for $k>n$ or $k<0$ or $n<0$.
\end{definition}

The expressions of $B_{n,k}$ for the first few values of $n, k \geq 1$ are given in \autoref{subsec:fdb_wasserstein:findim}.
We also introduce the following shorthand for equality of tensors up to reordering of the indices.
\begin{definition}
	We denote by $\frakS_n$ the set of permutations of $\{1, ..., n\}$.
	For two $n$-tensors $g^{i_1 ... i_n}$ and $h^{i_1 ... i_n}$, we will write $g \eqpermut h$ if there exists $\sigma \in \frakS_n$ such that $g^{i_{\sigma(1)} ... i_{\sigma(n)}} = h^{i_1 ... i_n}$.
\end{definition}

\subsection{The Faa di Bruno's formula for derivatives along a curve} \label{subsec:apx_bell:fdb}

Consider $f: \RR^d \to \RR$ and $(x_t)_{t \in \RR}$ any curve in $\RR^d$.
By Faa di Bruno's formula on $(f \circ x)(t)$,
\begin{equation*}
    \frac{d^n}{dt^n} f(x_t)
    = \sum_{k=1}^n (\nabla^k f) : B_{n,k}\left( \frac{dx}{dt}, \frac{d^2 x}{dt^3}, \frac{d^3 x}{dt^3}, ... \right)
\end{equation*}
where $\nabla^k f$ refers to the tensor $\nabla_{i_1} ... \nabla_{i_k} f(x_t)$ and $B_{n,k}$ to the partial Bell polynomial, with the understanding that multiplication refers to tensor product of vectors in $\RR^n$ (and the order does not matter since the tensors $\nabla^k f$ are symmetric).

Let us show the formula for the first few values of $n$ explicitly.
For the sake of concision, only in this paragraph, we will write $d$ or $D$ for $\frac{d}{dt}$.
One can easily check that
\begin{align*}
    D(f(x_t)) &= (\nabla f) : dx \\
    D^2(f(x_t)) &= (\nabla f) : d^2 x + (\nabla^2 f) : (dx)^{\otimes 2} \\
    D^3(f(x_t)) &= (\nabla f) : d^3 x + 3 (\nabla^2 f) : \left[ dx \otimes (d^2 x) \right] + (\nabla^3 f) : (dx)^{\otimes 3} \\
    D^4(f(x_t)) &=
    (\nabla f) : d^4 x
    + (\nabla^2 f) : \left[ 3 (d^2 x)^{\otimes 2} + 4 (dx) \otimes (d^3 x) \right] 
    + 6 (\nabla^3 f) : \left[ (dx)^{\otimes 2} \otimes (d^2 x) \right] \\
    & \qquad
    + (\nabla^4 f) : (dx)^{\otimes 4}.
\end{align*}
Taking a step back, if we think about what is the seed that generated the above structure, we can frame it as follows.
Define a sequence of symmetric multi-linear operators over tensors by, for each $n \geq 1$,
$\Lambda_n(x_1 \otimes ... \otimes x_n) = (\nabla^n f) : (x_1 \otimes ... \otimes x_n)$.
What we used is really that for any $x_1, ..., x_n \in \left\{ d^k x, k \geq 1 \right\}$,
\begin{equation*}
    D \Lambda_n(x_1 \otimes ... \otimes x_n) = \Lambda_{n+1} ((dx) \otimes x_1 \otimes ... \otimes x_n) + \Lambda_n( d(x_1 \otimes ... \otimes x_n) ),
\end{equation*}
and that $d$ satisfies the Leibniz product rule.
We record this fact in a general abstract proposition for future use:
\begin{proposition} \label{prop:apx_bell:abstract_vectorspaces}
    We call derivation over an algebra $A$ any linear operator $d: A \to A$ satisfying the Leibniz product rule
    $d(ab) = a d(b) + d(a) b$.
    We denote the tensor algebra of a vector space $\XXX$ by
    $T(\XXX) = \bigoplus_{k=0}^\infty \XXX^{\otimes k}$.
    (For the purpose of this proposition, it is not necessary to specify what is $\XXX^{\otimes 0}$, in other words here the direct sum can be taken starting from $k=1$.)

    Let $\XXX$ and $\YYY$ be vector spaces and consider a sequence of symmetric multi-linear operators ${ \Lambda_n: \XXX^n \to \YYY }$ ($n \geq 1$).
    Extend $\Lambda_n$ to $\XXX^{\otimes n}$ linearly and write $\Lambda_n : g_n$ for $\Lambda_n(g_n)$.
    Suppose there exist a linear operator $D: \YYY \to \YYY$, a derivation $d$ over $T(\XXX)$ and an element $\dot{x} \in \XXX$ such that
    \begin{align*}
        \forall n \geq 1,~~
        \forall x \in \XXX,&~
        &
        D\left( \Lambda_n : x^{\otimes n} \right) &= \Lambda_{n+1}: \dot{x} \otimes x^{\otimes n} + \Lambda_n : d\left( x^{\otimes n} \right) 
        && \\
        \text{i.e.,}~~~~ 
        \forall g_n \in \XXX^{\otimes n},&~
        &
        D\left( \Lambda_n : g_n \right) &= \Lambda_{n+1} : \dot{x} \otimes g_n + \Lambda_n : d(g_n)
        &&
    \end{align*}
    by symmetry and polarization.
    Then for any $n \geq 1$,
    \begin{equation*}
        D^{n-1} \Lambda_1(\dot{x})
        = \sum_{k=1}^n \Lambda_k : B_{n,k}\left( \dot{x}, d \dot{x}, d^2 \dot{x}, ... \right).
    \end{equation*}
\end{proposition}
For the Faa di Bruno's formula for time-derivatives along a curve in $\RR^d$ over a time interval $I$,
$\YYY$ corresponds to $\RR^I$,
$\XXX$ to $(\RR^d)^I$,
$D$ and $d$ to $\frac{d}{dt}$,
$\Lambda_n : (x_{1t})_t \otimes ... \otimes (x_{nt})_t$ to $\left( \nabla^n f(x_t) : x_{1t} \otimes ... \otimes x_{nt} \right)_t$,
and $\dot{x}$ to $\frac{dx_t}{dt}$.

\begin{proof}
    Denoting by $\Lambda: T(\XXX) \to \YYY$ the linear operator such that $\restr{\Lambda}{\XXX^{\otimes n}} = \Lambda_n$ for all $n$, the proposition can be restated as: suppose
    $
        \forall g \in T(\XXX),~ 
        D \Lambda(g) = \Lambda(\dot{x} \otimes g) + \Lambda(d(g))
    $,
    then
    $
        \forall n \geq 1$, $
        D^{n-1} \Lambda(\dot{x}) = \Lambda\left( \sum_{k=1}^n B_{n,k} \left( \dot{x}, d\dot{x}, d^2\dot{x}, ... \right) \right)
    $.
    
    Let us prove this by induction.
    The case $n=1$ is clear. For $n \geq 1$, for brevity, we write $B_{n,k}$ for $B_{n,k}(\dot{x}, d\dot{x}, d^2\dot{x}, ...)$, and let us suppose 
    $
        D^{n-1} \Lambda(\dot{x}) = \Lambda\left( \sum_{k=1}^n B_{n,k} \right)
    $. Then
    \begin{equation*}
        D D^{n-1} \Lambda(\dot{x})
        = D \Lambda\left( \sum_{k=1}^n B_{n,k} \right)
        = \Lambda\left( \sum_{k=1}^n \dot{x} \otimes B_{n,k} + d(B_{n,k}) \right).
    \end{equation*}
    Now by the lemma below, we have the telescoping sum
    \begin{equation*}
        \sum_{k=1}^n d(B_{n,k}) + \dot{x} \otimes B_{n,k}
        \eqpermut \sum_{k=1}^n B_{n+1,k} - \dot{x} \otimes B_{n,k-1} + \dot{x} \otimes B_{n,k}
        \eqpermut \sum_{k=1}^n B_{n+1,k}
        + \underbrace{ \dot{x} \otimes \left( B_{n,n} - B_{n,0} \right) }_{= \dot{x}^{\otimes (n+1)} = B_{n+1,n+1}}
    \end{equation*}
    since $B_{n,n}(X_1) = X_1^n$ and $B_{n,0}=0$.
    Hence $
        D D^{n-1} \Lambda(\dot{x})
        = \Lambda\left( \sum_{k=1}^{n+1} B_{n+1,k} \left( \dot{x}, d\dot{x}, d^2\dot{x}, ... \right) \right)
    $,
    which concludes the proof by induction.
\end{proof}

\begin{lemma} \label{lm:apx_bell:deriv_B}
    For any derivation $d$ over $T(\XXX)$ and any $\dot{x} \in \XXX$,
    for any $n, k \geq 0$,
    \begin{equation*}
        d \left( B_{n,k}(\dot{x}, d\dot{x}, d^2\dot{x}, ...) \right)
        \eqpermut B_{n+1,k}(\dot{x}, d\dot{x}, d^2\dot{x}, ...)
        - \dot{x} \otimes B_{n,k-1}(\dot{x}, d\dot{x}, d^2\dot{x}, ...).
    \end{equation*}
\end{lemma}

\begin{proof}
    We use the properties of the partial Bell polynomial $B_{n,k}(X_1, X_2, X_3, ...)$ that
    \begin{equation*}
        \frac{\partial B_{n,k}}{\partial X_i}(X_1, ...) = \binom{n}{i} B_{n-i,k-1}(X_1, ...) 
        ~~~~\text{and}~~~~
        B_{n+1,k}(X_1, ...) = \sum_{i=0}^\infty \binom{n}{i} X_{i+1} B_{n-i,k-1}(X_1, ...).
    \end{equation*}
    Since derivations of polynomials follow the chain rule,
    this implies as announced
    \begin{equation*}
        d \left( B_{n,k}(\dot{x}, d\dot{x}, d^2\dot{x}, ...) \right)
        \eqpermut \sum_{i=1}^\infty \frac{\partial B_{n,k}}{\partial X_i} (\dot{x},...)
        \otimes d^i \dot{x}
        \eqpermut B_{n+1,k}(\dot{x},...) 
        - \dot{x} \otimes B_{n,k-1}(\dot{x}, ...).
        \rqedhere
    \end{equation*}
\end{proof}

\begin{samepage}
\begin{remark} \label{rk:apx_bell:abstract_modules}
    The proposition stated above was relatively conservative; the version below also captures the case $n=0$, at the cost of some gymnastics in the definitions.
    \begin{proposition*}
        We denote the tensor algebra of a $R$-module $\XXX$ (for a commutative ring $R$)
        as $T(\XXX) = \bigoplus_{k=0}^\infty \XXX^{\otimes k}$ with $\XXX^{\otimes 0} = R$.
        (Recall that the axioms for [modules over] commutative rings are the same as for [vector spaces over] fields except for the existence of multiplicative inverses.)

        Let $\XXX$ and $\YYY$ be $R$-modules and consider a sequence of symmetric multi-linear operators ${ \Lambda_n: \XXX^n \to \YYY }$ ($n \geq 0$).
        Extend $\Lambda_n$ to $\XXX^{\otimes n}$ linearly and let $\Lambda: T(\XXX) \to \YYY$ be the linear operator such that $\restr{\Lambda}{\XXX^{\otimes n}} = \Lambda_n$ for all $n$.
        Suppose there exist a linear operator $D: \YYY \to \YYY$, a derivation $d$ over $T(\XXX)$ and an element $\dot{x} \in \XXX$ such that 
        $
            \forall g \in T(\XXX),~ D \Lambda(g) = \Lambda(\dot{x} \otimes g) + \Lambda(d(g))
        $.
        Then 
        $
            \forall n \geq 0,~ D^n \Lambda(1_R) = \Lambda\left( \sum_{k=0}^n B_{n,k} \left( \dot{x}, d\dot{x}, d^2\dot{x}, ... \right) \right)
        $.
    \end{proposition*}
    For the Faa di Bruno's formula for time-derivatives along a curve in $\RR^d$ over a time interval~$I$,
    both $R$ and $\YYY$ correspond to $\RR^I$,
    $\XXX$ to $(\RR^d)^I \simeq R^d$,
    $D$ and $d$ to $\frac{d}{dt}$,
    $\Lambda_0(1_R)$ to $(f(x_t))_t$,
    $\Lambda_n : (x_{1t})_t \otimes ... \otimes (x_{nt})_t$ to $\left( D^n f(x_t) : x_{1t} \otimes ... \otimes x_{nt} \right)_t$,
    and $\dot{x}$ to $\frac{dx_t}{dt}$.
    The proof is identical to the one of \autoref{prop:apx_bell:abstract_vectorspaces}, as \autoref{lm:apx_bell:deriv_B} also holds for $\XXX$ being a $R$-module and $T(\XXX)$ including $\XXX^{\otimes 0}=R$.
\end{remark}
\end{samepage}

\subsection{Composition of Bell polynomials} \label{subsec:apx_bell:compos}

In this subsection we derive a composition rule for Bell polynomials which will be useful for the proof of \autoref{thm:fdb_wasserstein:fdb_wasserstein}.

\subsubsection{The scalar-scalar-scalar case}
For ease of exposition we start by a simple setting where all of the quantities considered are scalars.
Consider
\begin{equation*}
    t \in \RR \xrightarrow{~x~} x_t \in \RR \xrightarrow{~f~} y_t \in \RR \xrightarrow{~g~} z_t \in \RR.
\end{equation*}
That is, $x, f, g$ are all $\RR \to \RR$, and we write
$z_t = g(y_t) = g(f(x_t))$.
Additionally denote $x_{nt} = \frac{d^n}{dt^n} x_t$ and likewise for $y_{nt}, z_{nt}$.
Then by Faa di Bruno's formula,
\begin{equation*}
    z_{nt} = \frac{d^n}{dt^n} g(y_t) =  \sum_{k=0}^n g^{(k)}(y_t) B_{n,k}(y_{1t}, y_{2t}, ...)
    ~~\text{and}~~
    y_{nt} = \frac{d^n}{dt^n} f(x_t) = \sum_{k=0}^n f^{(k)}(x_t) B_{n,k}(x_{1t}, x_{2t}, ...).
\end{equation*}
On the other hand, also by Faa di Bruno's formula,
\begin{equation*}
    z_{nt} = \frac{d^n}{dt^n} (g \circ f)(x_t) = \sum_{k=0}^n (g \circ f)^{(k)}(x_t)
    B_{n,k}(x_{1t}, x_{2t}, ...),
\end{equation*}
and by Faa di Bruno's formula with variable $x$ instead of $t$,
\begin{align} \label{eq:apx_bell:FdB_scalar_gfx}
    (g \circ f)^{(k)}(x)
    = \frac{d^k}{dx^k} g(f(x))
    &= \sum_{l=0}^k g^{(l)}(f(x))
    B_{k,l}\left( \frac{df}{dx}, \frac{d^2f}{dx^2}, ... \right) \\
    \text{so}~~~~
    (g \circ f)^{(k)}(x_t)
    &= \sum_{l=0}^k g^{(l)}(y_t)
    B_{k,l}\left( f'(x_t), f^{(2)}(x_t), ... \right) \\
    \text{so}~~~~
    z_{nt}
    &= \sum_{k=0}^n \sum_{l=0}^k g^{(l)}(y_t)
    B_{k,l}\left( f'(x_t), f^{(2)}(x_t), ... \right)
    B_{n,k}(x_{1t}, x_{2t}, ...).
\end{align}
Since the functions $x, f, g$ were arbitrary, this reasoning establishes a composition rule for the Bell polynomials, as formalized below.

\begin{samepage}
\begin{proposition} \label{prop:apx_bell:compos_scalar}
    For any scalar sequences $\bx,\bmf,\bg \in \RR^\NN$, 
    let $\by,\bz \in \RR^\NN$ be defined by
    $\forall n \geq 0, z_n = \sum_{k=0}^n g_k B_{n,k}(\by)$
    and
    $\forall k \geq 0, y_k = \sum_{l=0}^k f_l B_{k,l}(\bx)$.
    Then $\forall n \geq 0, z_n = \sum_{k=0}^n \sum_{l=0}^k g_l B_{k,l}(\bmf) B_{n,k}(\bx)$.
\end{proposition}
\end{samepage}

\begin{proof}
    The case $n=0$ can be treated separately easily since $B_{0,k} = 1$ if $k=0$ and $0$ otherwise.
    For the case $n \geq 1$, choose any $N \geq n$ and apply the reasoning above to
    $x(t) = \sum_{k=1}^N x_k \frac{t^k}{k!}$,
    $f(x) = \sum_{k=1}^N f_k \frac{x^k}{k!}$,
    $g(y) = \sum_{k=1}^N g_k \frac{y^k}{k!}$,
    and evaluate at $t=0$. This yields precisely the announced formula since $B_{n,k}(X_1, X_2, X_3, ...)$ only depends on $X_1, ..., X_{n-k+1}$ and since $B_{n,k}=0$ for $k>n$.
\end{proof}


\subsubsection{Generalization to tensor algebras} \label{subsubsec:apx_bell:compos:abstract}
The following proposition is a generalization of \autoref{prop:apx_bell:compos_scalar} to tensor algebras in the spirit of \autoref{prop:apx_bell:abstract_vectorspaces}.

\begin{proposition} \label{prop:apx_bell:compos_abstract}
    Let $\XXX, \YYY, \ZZZ$ be vector spaces 
    and consider a sequence of symmetric multi-linear operators $\Lambda_n: \XXX^n \to \YYY$ ($n \geq 1$), extended linearly to $\Lambda_n: \XXX^{\otimes n} \to \YYY$.
    Also consider a sequence of symmetric multi-linear operators $Q_n: \YYY^n \to \ZZZ$ ($n \geq 1$), similarly extended to $Q_n: \YYY^{\otimes n} \to \ZZZ$.
    Suppose there exist a derivation $d$ over $T(\XXX)$, a linear map $D: \YYY \to \YYY$ and an element $\dot{x} \in \XXX$ such that
    \begin{equation*}
        \forall n \geq 1,~ \forall g_n \in \XXX^{\otimes n},~
        D\left( \Lambda_n : g_n \right) = \Lambda_{n+1} : \dot{x} \otimes g_n + \Lambda_n : d(g_n).
    \end{equation*}
    Moreover, extend $D$ to a derivation over 
    $T(\YYY)$ by $D(y_1 \otimes ... \otimes y_n) = (D y_1) \otimes ... \otimes y_n + ... + y_1 \otimes ... \otimes (D y_n)$
    and suppose there exists a linear map $\Delta: \ZZZ \to \ZZZ$ such that, denoting $\dot{y} = \Lambda_1(\dot{x})$,
    \begin{equation*}
        \forall n \geq 1,~ \forall h_n \in \YYY^{\otimes n},~
        \Delta\left( Q_n : h_n \right) = Q_{n+1} : \dot{y} \otimes h_n + Q_n : D(h_n).
    \end{equation*}
    Then
    \begin{equation} \label{eq:apx_bell:compos:abstract:prop_ccl}
        \forall n \geq 1,~
        \Delta^{n-1} Q_1(\dot{y})
        = \sum_{k=1}^n \sum_{l=1}^k Q_l : B_{k,l}(\Lambda_1, \Lambda_2, ...) : B_{n,k}(\dot{x}, d\dot{x}, d^2\dot{x}, ...)
    \end{equation}
    where $B_{k,l}(\Lambda_1, \Lambda_2, ...)$
    and the operators denoted by ``$:$'' are defined in \autoref{def:apx_bell:compos:abstract:notations} below.
\end{proposition}

Let us give the rigorous meaning of the above identity.
This will take some gymnastics, but in some sense it is only an exercise in formality, as there is no ambiguity on how it should be defined.
The following definition will only be used within the present subsection.
\begin{definition} \label{def:apx_bell:compos:abstract:notations}
    For vector spaces $\XXX, \YYY, \ZZZ$,
    denote by $\LLL(\XXX, \YYY)$ the set of linear maps from $\XXX$ to $\YYY$,
    and by $\LLL(\XXX^{\otimes n}, \YYY)$ the set of $n$-linear maps, or isomorphically, of linear maps over $n$-tensors.
    \begin{itemize}
        \item For all $k, l \in \NN^*$, let
        $\LLL_k = \LLL(\XXX^{\otimes k}, \YYY)$
        and 
        $\LLL_{k,l} = \bigoplus_{\substack{k_1,...,k_l>0, \\ k_1+...+k_l=k}} \LLL_{k_1} \otimes ... \otimes \LLL_{k_l}$.
        Note that, for any $\Lambda \in \LLL_{k,l}$ and $\Lambda' \in \LLL_{k',l'}$, we have $\Lambda \otimes \Lambda' \in \LLL_{k+k', l+l'}$.
        Furthermore, for any sequence of operators $\Lambda_k \in \LLL_k$ ($k \geq 1$), let 
        $B_{k,l}(\Lambda_1, \Lambda_2, ...) 
        \coloneqq \sum_{\bh \in \NN^{\NN^*}} \alpha^{k,l}_{\bh} \left( \Lambda_1^{\otimes h_1} \otimes \Lambda_2^{\otimes h_2} \otimes ... \right)
        \in \LLL_{k,l}$.
        \item Define a bilinear operator $\LLL_{k,l} \times \XXX^{\otimes k} \to \YYY^{\otimes l}$, denoted by ``$:$'', as follows.
        For any $\Lambda^{(1)} \otimes ... \otimes \Lambda^{(l)} \allowbreak \in \allowbreak \LLL_{k,l}$, say $\Lambda^{(i)} \in \LLL_{k_i}$ with $k_1+...+k_l=k$, and any $x_1, ..., x_k \in \XXX$,
        \begin{equation*}
            \left( \Lambda^{(1)} \otimes ... \otimes \Lambda^{(l)} \right) : \left( x_1 \otimes ... \otimes x_k \right)
            ~\coloneqq~ 
            \left( \Lambda^{(1)} : x^{(1)} \right)
            \otimes ... \otimes
            \left( \Lambda^{(l)} : x^{(l)} \right)
        \end{equation*}
        where $x^{(i)} = x_{k_1+...+k_{i-1}+1} \otimes ... \otimes x_{k_1+...+k_{i-1}+k_i}$.
        This suffices to define ``$:$'' as an operator over $\LLL_{k,l} \times \XXX^{\otimes k}$ by bilinearity.
        \item Define a bilinear operator $\LLL(\YYY^{\otimes l}, \ZZZ) \times \LLL_{k,l} \to \LLL(\XXX^{\otimes k}, \ZZZ)$, also denoted by ``$:$'', as follows.
        For any $Q_l \in \LLL(\YYY^{\otimes l}, \ZZZ)$ and $\Lambda^{(1)} \otimes ... \otimes \Lambda^{(l)} \in \LLL_{k,l}$, say $\Lambda^{(i)} \in \LLL_{k_i}$ with $k_1+...+k_l=k$,
        \begin{equation} \label{eq:apx_bell:def_tQl}
            \left( Q_l : \Lambda^{(1)} \otimes ... \otimes \Lambda^{(l)} \right) \Big( x_1 \otimes ... \otimes x_k \Big)
            ~\coloneqq~
            Q_l\left(
                \Lambda^{(1)} : x^{(1)},
                ...,
                \Lambda^{(l)} : x^{(l)}
            \right)
        \end{equation}
        for any $x_1,...,x_k \in \XXX$,
        where $x^{(i)} = x_{k_1+...+k_{i-1}+1} \otimes ... \otimes x_{k_1+...+k_{i-1}+k_i}$.
        This suffices to define an element $Q_l : \Lambda^{(1)} \otimes ... \otimes \Lambda^{(l)}$ of $\LLL(\XXX^{\otimes k}, \ZZZ)$ by linearity. In turn, this suffices to define ``$:$'' as an operator over $\LLL(\YYY^{\otimes l}, \ZZZ) \times \LLL_{k,l}$ by bilinearity.
    \end{itemize}
    Note that using the notation $Q_l(y_1, ..., y_l) \eqqcolon Q_l : (y_1 \otimes ... \otimes y_l)$, the last equation rewrites
    \begin{equation*}
        \left( Q_l : \Lambda^{(1)} \otimes ... \otimes \Lambda^{(l)} \right) : \Big( x_1 \otimes ... \otimes x_k \Big)
        ~=~
        Q_l : \left(
            \left( \Lambda^{(1)} : x^{(1)} \right)
            \otimes ... \otimes
            \left( \Lambda^{(l)} : x^{(l)} \right)
        \right).
    \end{equation*}
    Consequently, by bilinearity, the two notations ``$:$'' defined above are associative in the sense that for all
    $Q_l \in \LLL(\YYY^{\otimes l}, \ZZZ)$,
    $\bm\Lambda \in \LLL_{k,l}$,
    $\bx \in \XXX^{\otimes k}$,
    it holds
    $(Q_l : \bm\Lambda) : \bx = Q_l : (\bm\Lambda : \bx)$.
    This justifies the absence of parentheses in \eqref{eq:apx_bell:compos:abstract:prop_ccl}.
\end{definition}

\begin{remark}
    Thanks to the second item of the definition above, $\LLL_{k,l}$ can be viewed as a subspace of $\LLL(\XXX^{\otimes k}, \YYY^{\otimes l})$, by identifying $\bm{\Lambda} \in \LLL_{k,l}$ to $\bx \mapsto \bm{\Lambda} : \bx$---although one should still check that this identification is injective, which we did not do.
    In this interpretation, the operator of the second item $\LLL_{k,l} \times \XXX^{\otimes k} \to \YYY^{\otimes l}$ is simply the evaluation operator $\bm\Lambda : \bx = \bm\Lambda(\bm{x})$,
    and the operator of the third item $\LLL(\YYY^{\otimes l}, \ZZZ) \times \LLL_{k,l} \to \LLL(\XXX^{\otimes k}, \ZZZ)$ is simply the composition operator $Q_l : \bm\Lambda = Q_l \circ \bm\Lambda$.
\end{remark}

\begin{remark} \label{rk:apx_bell:compos:abstract:xBLambda}
    The following elementary computation will be useful several times in the proofs to come:
    for any sequence of operators $\Lambda_n \in \LLL(\XXX^{\otimes n}, \YYY)$ and any $x \in \XXX$,
    \begin{align*}
        \forall n, k,~ 
        B_{n,k}(\Lambda_1, \Lambda_2, ...) : x^{\otimes n}
        &= \sum_{\bh \in \NN^{\NN^*}} \alpha^{n,k}_{\bh}\,
        \left( \Lambda_1^{\otimes h_1} \otimes \Lambda_2^{\otimes h_2} \otimes ... \right) : x^{\otimes n} \\
        &= \sum_{\bh \in \NN^{\NN^*}} \alpha^{n,k}_{\bh}\,
        \left( \Lambda_1 : x \right)^{\otimes h_1} \otimes
        \left( \Lambda_2 : x^{\otimes 2} \right)^{\otimes h_2} \otimes ... \\
        &= B_{n,k}\left( \Lambda_1(x), \Lambda_2(x), ... \right).
    \end{align*}
\end{remark}

The following lemma shows that, under the conditions of the proposition, we have a composition rule at the level of derivations.
\begin{lemma} \label{lm:apx_bell:compos:abstract:DeltaEk}
    In the setting of \autoref{prop:apx_bell:compos_abstract}, the sequence of multi-linear operators
    \begin{equation*}
        E_n = \sum_{k=1}^n Q_k : B_{n,k}(\Lambda_1, \Lambda_2, ...) \in \LLL(\XXX^{\otimes n}, \ZZZ)
    \end{equation*}
    satisfies
    \begin{equation*}
        \forall n,~ \forall x \in \XXX,~
        \Delta (E_n : x^{\otimes n}) = E_{n+1} : \dot{x} \otimes x^{\otimes n} + E_n : d(x^{\otimes n}).
    \end{equation*}
\end{lemma}

Our proof of \autoref{lm:apx_bell:compos:abstract:DeltaEk} goes by explicit computations and is rather tedious, so we chose to delay it to \autoref{subsubsec:apx_bell:compos:proof_abstract_lm}.
We now turn to the proof of the proposition.

\begin{proof}[Proof of \autoref{prop:apx_bell:compos_abstract}]
    Let the $n$-linear operator $E_n = \sum_{k=1}^n Q_k : B_{n,k}(\Lambda_1, \Lambda_2, ...) \in \LLL(\XXX^{\otimes n}, \ZZZ)$ for all $n \geq 1$.
    $E_n$ is symmetric since the $\Lambda_k, Q_k$ are, and
    \autoref{lm:apx_bell:compos:abstract:DeltaEk} shows that
    for all $n \geq 1$ and $x \in \XXX$,
    $
        \Delta (E_n : x^{\otimes n}) = E_{n+1} : \dot{x} \otimes x^{\otimes n} + E_n : d(x^{\otimes n}).
    $
    So by the abstract Faa di Bruno's formula of \autoref{prop:apx_bell:abstract_vectorspaces}, we have 
    \begin{align*}
        \forall n \geq 1,~
        \Delta^{n-1} E_1(\dot{x}) &= \sum_{k=1}^n E_k : B_{n,k}(\dot{x}, d\dot{x}, d^2\dot{x}, ...) \\
        = \Delta^{n-1} Q_1 \Lambda_1 \dot{x} &= \sum_{k=1}^n \left( \sum_{l=1}^k Q_l : B_{k,l}(\Lambda_1, \Lambda_2, ...) \right) : B_{n,k}(\dot{x}, d\dot{x}, d^2\dot{x}, ...),
    \end{align*}
    which is the announced identity.
\end{proof}

\begin{remark}[Composition rule at the level of Bell polynomials]
    Despite the title of this section, \autoref{prop:apx_bell:compos_abstract} arguably falls short of establishing a composition rule for Bell polynomials. Indeed, it assumes that $d, D, \dot{x}$ and $(\Lambda_n)_n$, resp.\ $D, \Delta, \dot{y}$ and $(Q_n)_n$, are related by a derivation rule which is a priori stronger than an equation with Bell polynomials (by \autoref{prop:apx_bell:abstract_vectorspaces}). That is, one may ask whether the following refinement holds:
    {\it
        With the notations of \autoref{prop:apx_bell:compos_abstract}, suppose that
        \begin{equation*}
            \forall n \geq 1,
            D^{n-1} \Lambda_1(\dot{x}) = \sum_{k=1}^n \Lambda_k : B_{n,k}(\dot{x}, d\dot{x}, d^2\dot{x}, ...) 
            ~\text{and}~
            \Delta^{n-1} Q_1(\dot{y}) = \sum_{k=1}^n Q_k : B_{n,k}(\dot{y}, D\dot{y}, D^2\dot{y}, ...).
        \end{equation*}
        Then can we conclude to \eqref{eq:apx_bell:compos:abstract:prop_ccl}?
    }
\end{remark}

\subsubsection{Proof of \autoref{lm:apx_bell:compos:abstract:DeltaEk}} \label {subsubsec:apx_bell:compos:proof_abstract_lm}

\begin{proof}[Proof of \autoref{lm:apx_bell:compos:abstract:DeltaEk}]
    We proceed by explicit computations. Fix $n \geq 1$ and $x \in \XXX$.
    The left-hand side of the desired equality is given by
    \begin{align*}
        \Delta(E_n : x^{\otimes n})
        &= \sum_{k=1}^n \Delta \big( Q_k : B_{n,k}(\Lambda_1, \Lambda_2, ...) : x^{\otimes n} \big) \\
        &= \sum_{k=1}^n \Delta \big( Q_k : B_{n,k}(\Lambda_1(x), \Lambda_2(x), ...) \big) \\
        &= \sum_{k=1}^n Q_{k+1} : \dot{y} \otimes B_{n,k}(\Lambda_1(x), \Lambda_2(x), ...) 
        + Q_k : D \left[B_{n,k}(\Lambda_1(x), \Lambda_2(x), ...)\right]
    \end{align*}
    where the first equality uses the associativity of the ``$:$'', the second uses \autoref{rk:apx_bell:compos:abstract:xBLambda}, and the third uses the definition of $\Delta$.
    The right-hand side of the desired equality is given by
    \begin{align*}
        & E_{n+1} : \dot{x} \otimes x^{\otimes n} + E_n : d(x^{\otimes n})
        = E_{n+1} : \dot{x} \otimes x^{\otimes n} + n E_n : dx \otimes x^{\otimes (n-1)} \\
        &= \frac{1}{n+1} \sum_{k=1}^{n+1} Q_k :
        \sum_{\bh \in \NN^{\NN^*}} \alpha^{n+1,k}_{\bh}
        \sum_{i=1}^\infty i \, h_i \left( \Lambda_i : \dot{x} \otimes x^{\otimes (i-1)} \right)
        \otimes \Lambda_i(x)^{\otimes (h_i-1)}
        \otimes \bigotimes_{j \neq i} \Lambda_j(x)^{\otimes h_j} \\
        &~~ + \sum_{k=1}^{n} Q_k :
        \sum_{\bh \in \NN^{\NN^*}} \alpha^{n,k}_{\bh}
        \sum_{i=1}^\infty i \, h_i \left( \Lambda_i : dx \otimes x^{\otimes (i-1)} \right)
        \otimes \Lambda_i(x)^{\otimes (h_i-1)}
        \otimes \bigotimes_{j \neq i} \Lambda_j(x)^{\otimes h_j} \\
        &= \sum_{k=1}^{n+1}
        \underbrace{
            \frac{1}{n+1} Q_k :
            \sum_{\bh \in \NN^{\NN^*}} \alpha^{n+1,k}_{\bh}
            \sum_{i=1}^\infty i \, h_i \left( \Lambda_i : \dot{x} \otimes x^{\otimes (i-1)} \right)
            \otimes \Lambda_i(x)^{\otimes (h_i-1)}
            \otimes \bigotimes_{j \neq i} \Lambda_j(x)^{\otimes h_j} 
        }_{\eqqcolon A_k} \\
        &~~ + \sum_{k=1}^n Q_k : D \left[B_{n,k}(\Lambda_1(x), \Lambda_2(x), ...)\right] \\
        &~~ - \sum_{k=1}^{n} Q_k :
        \sum_{\bh \in \NN^{\NN^*}} \alpha^{n,k}_{\bh}
        \sum_{i=1}^\infty h_i \left( \Lambda_{i+1} : \dot{x} \otimes x^{\otimes i} \right)
        \otimes \Lambda_i(x)^{\otimes (h_i-1)}
        \otimes \bigotimes_{j \neq i} \Lambda_j(x)^{\otimes h_j}.
    \end{align*}
    where the second equality uses \autoref{claim:apx_bell:compos:abstract:calc1} below twice, and the third uses \autoref{claim:apx_bell:compos:abstract:calc2} below.

    Let us compute the expression on the first line in the equation above.
    For $A_1$, we have $B_{n+1,1}(X_1, X_2, ...) = X_{n+1}$, i.e., $\alpha^{n+1,1}_{\bh} = 1$ for $\bh = (\ind_{[i = n+1]})_{i \in \NN^*}$ and $0$ for all other $\bh$, so
    \begin{equation*}
        A_1 = \frac{1}{n+1} Q_1 :
        (n+1) \left( \Lambda_{n+1} : \dot{x} \otimes x^{\otimes n} \right)
        = Q_1 : \Lambda_{n+1} : \dot{x} \otimes x^{\otimes n}.
    \end{equation*}
    For $A_{n+1}$, we have $B_{n+1,n+1}(X_1, X_2, ...) = X_1^{n+1}$, i.e., $\alpha^{n+1,n+1}_{\bh} = 1$ for $\bh = ((n+1) \ind_{[i=1]})_{i \in \NN^*}$ and $0$ for all other $\bh$, so
    \begin{align*}
        A_{n+1} 
        = \frac{1}{n+1} Q_{n+1} : (n+1) \Lambda_1(\dot{x}) \otimes \Lambda_1(x)^{\otimes n}
        &= Q_{n+1} : \dot{y} \otimes \Lambda_1(x)^{\otimes n} \\
        &= Q_{n+1} : \dot{y} \otimes B_{n,n}(\Lambda_1(x), \Lambda_2(x), ...).
    \end{align*}
    For the $A_k$ with $2 \leq k \leq n$, let
    \begin{equation*}
        A_k = Q_k : \sum_{i=1}^\infty
        \underbrace{
            \sum_{\bh} \frac{h_i \, i}{n+1} \,\alpha^{n+1,k}_{\bh} 
            \left( \Lambda_i : \dot{x} \otimes x^{\otimes (i-1)} \right)
            \otimes \Lambda_i(x)^{\otimes (h_i-1)}
            \otimes \bigotimes_{j \neq i} \Lambda_j(x)^{\otimes h_j}.
        }_{\eqqcolon A_k^{(i)}}
    \end{equation*}
    For each $j \in \NN^*$, denote by $\bmone_j$ the sequence $(\ind_{[m=j]})_{m \in \NN^*}$.
    Then for $A_k^{(1)}$, by a change of indices $\bh' = \bh - \bmone_1$,
    \begin{align*}
        A_k^{(1)} &= \sum_{\bh ; h_1>0}~ \frac{h_1}{n+1} \,\alpha^{n+1,k}_{\bh}~
        (\Lambda_1 \dot{x}) \otimes (\Lambda_1 x)^{\otimes (h_1-1)} \otimes \bigotimes_{j \neq 1} \Lambda_j(x)^{\otimes h_j} \\
        &= (\Lambda_1 \dot{x})
        \otimes \sum_{\bh'} 
        ~
        \underbrace{ 
            \frac{h'_1+1}{n+1} \,\alpha^{n+1,k}_{\bh'+\bmone_1} 
        }_{= \alpha^{n,k-1}_{\bh'}}
        ~
        \underbrace{ 
            (\Lambda_1 x)^{\otimes h'_1} \otimes \bigotimes_{j \neq 1} \Lambda_j(x)^{\otimes h'_j}
        }_{= \bigotimes_{j=1}^\infty \Lambda_j(x)^{\otimes h'_j}}
        ~
        = \dot{y} \otimes B_{n,k-1}(\Lambda_1(x), \Lambda_2(x), ...),
    \end{align*}
    where for the underbraced coefficient we used \autoref{claim:apx_bell:compos:abstract:calc3} below.
    For the $A_k^{(i)}$ with $i \geq 2$, by a change of indices $\bh' = \bh - \bmone_i + \bmone_{i-1}$,
    \begin{align*}
        & A_k^{(i)}
        = \sum_{\bh} \frac{h_i \, i}{n+1} \,\alpha^{n+1,k}_{\bh} 
        \left( \Lambda_i : \dot{x} \otimes x^{\otimes (i-1)} \right)
        \otimes \Lambda_i(x)^{\otimes (h_i-1)}
        \otimes \left( \Lambda_{i-1}(x)^{\otimes h_{i-1}} \otimes \bigotimes_{j \not\in \{i, i-1\}} \Lambda_j(x)^{\otimes h_j} \right) \\
        &\eqpermut \sum_{\bh'} \frac{(h'_i+1) \, i}{n+1} \,\alpha^{n+1,k}_{\bh'+\bmone_i-\bmone_{i-1}}
        \left( \Lambda_i : \dot{x} \otimes x^{\otimes (i-1)} \right)
        \otimes \Lambda_i(x)^{\otimes h'_i}
        \otimes \Lambda_{i-1}(x)^{\otimes (h'_{i-1}-1)} \otimes \! \bigotimes_{j \not\in \{i, i-1\}} \! \Lambda_j(x)^{\otimes h'_j} \\
        &\eqpermut \sum_{\bh'} \frac{(h'_i+1) \, i}{n+1} \,\alpha^{n+1,k}_{\bh'+\bmone_i-\bmone_{i-1}}
        \left( \Lambda_i : \dot{x} \otimes x^{\otimes (i-1)} \right)
        \otimes \Lambda_{i-1}(x)^{\otimes (h'_{i-1}-1)} \otimes \bigotimes_{j \neq i-1} \Lambda_j(x)^{\otimes h'_j}
    \end{align*}
    and so
    \begin{align*}
        &\sum_{i=2}^\infty A_k^{(i)} = \sum_{i=1}^\infty A_k^{(i+1)} \\
        &\eqpermut \sum_{\bh'} \sum_{i=1}^\infty 
        ~\underbrace{
            \frac{(h'_{i+1}+1) \, (i+1)}{n+1} \,\alpha^{n+1,k}_{\bh'+\bmone_{i+1}-\bmone_i}
        }_{= h'_i \,\alpha^{n,k}_{\bh'}}~
        \left( \Lambda_{i+1} : \dot{x} \otimes x^{\otimes i} \right)
        \otimes \Lambda_i(x)^{\otimes (h'_i-1)} \otimes \bigotimes_{j \neq i} \Lambda_j(x)^{\otimes h'_j}
    \end{align*}
    where for the underbraced coefficient we again used \autoref{claim:apx_bell:compos:abstract:calc3}.
    
    Putting everything together, we get that the right-hand side of the desired equality is equal to
    \begin{align*}
        & A_{n+1} + \sum_{k=2}^n Q_k : A_k^{(1)} + \sum_{k=1}^n Q_k : D \left[B_{n,k}(\Lambda_1(x), \Lambda_2(x), ...)\right] + 0 \\
        &= \sum_{k=2}^{n+1} Q_k : \dot{y} \otimes B_{n,k-1}(\Lambda_1(x), \Lambda_2(x), ...)
        + \sum_{k=1}^n Q_k : D \left[B_{n,k}(\Lambda_1(x), \Lambda_2(x), ...)\right],
    \end{align*}
    which is indeed equal to the desired left-hand side.
\end{proof}

\begin{claim} \label{claim:apx_bell:compos:abstract:calc1}
    For any $n \geq 1$ and $x, X \in \XXX$,
    \begin{equation*}
        E_{n+1} : X \otimes x^{\otimes n}
        = \frac{1}{n+1} \sum_{k=1}^{n+1} Q_k :
        \!\sum_{\bh \in \NN^{\NN^*}}\! \alpha^{n+1,k}_{\bh}
        \sum_{i=1}^\infty i \, h_i \left( \Lambda_i : X \otimes x^{\otimes (i-1)} \right)
        \otimes \Lambda_i(x)^{\otimes (h_i-1)}
        \otimes \bigotimes_{j \neq i} \Lambda_j(x)^{\otimes h_j}.
    \end{equation*}
\end{claim}

\begin{proof}
    Denoting $E_{n+1}(x) = E_{n+1} : x^{\otimes (n+1)}$ for all $x \in \XXX$, we have by the polarization identity
    and by \autoref{rk:apx_bell:compos:abstract:xBLambda} that
    \begin{align*}
        & E_{n+1} : X \otimes x^{\otimes n}
        = \frac{1}{(n+1)!} \restr{\frac{\partial}{\partial \lambda_0} ... \frac{\partial}{\partial \lambda_n}}{\bm{\lambda}=0} 
        E_{n+1}\left( \lambda_0 X + (\lambda_1+...+\lambda_n) x \right) \\
        &= \frac{1}{(n+1)!} \restr{\frac{\partial}{\partial \lambda_0}}{\lambda_0=0} \restr{\frac{\partial^n}{\partial \lambda^n}}{\lambda=0} 
        E_{n+1}\left( \lambda_0 X + \lambda x \right) \\
        &= \frac{1}{(n+1)!} \restr{\frac{\partial}{\partial \lambda_0}}{\lambda_0=0} \restr{\frac{\partial^n}{\partial \lambda^n}}{\lambda=0}
        \sum_{k=1}^{n+1} Q_k : B_{n+1,k}\Big( \Lambda_1(\lambda_0 X + \lambda x), \Lambda_2(\lambda_0 X + \lambda x), ... \Big).
    \end{align*}
    Now, using that $Q_k$ is symmetric,
    \begin{align*}
        &~~~~ \frac{\partial}{\partial \lambda_0}
        Q_k : B_{n+1,k}\Big( \Lambda_1(\lambda_0 X + \lambda x), \Lambda_2(\lambda_0 X + \lambda x), ... \Big) \\
        &= \sum_{\bh \in \NN^{\NN^*}} \alpha^{n+1,k}_{\bh} 
        \frac{\partial}{\partial \lambda_0} 
        Q_k : \left( \bigotimes_{i=1}^\infty \Lambda_i(\lambda_0 X + \lambda x)^{\otimes h_i} \right) \\
        &= \sum_{\bh \in \NN^{\NN^*}} \alpha^{n+1,k}_{\bh} Q_k : 
        \sum_{i=1}^\infty h_i \left( \frac{\partial}{\partial \lambda_0} \Lambda_i(\lambda_0 X + \lambda x) \right)
        \otimes \Lambda_i(\lambda_0 X + \lambda x)^{\otimes (h_i-1)}
        \otimes \bigotimes_{j \neq i} \Lambda_j(\lambda_0 X + \lambda x)^{\otimes h_j}
    \end{align*}
    and, by multi-linearity of $\Lambda_i$, $\frac{\partial}{\partial \lambda_0} \Lambda_i(\lambda_0 X + \lambda x) 
    = i \Lambda_i : X \otimes (\lambda_0 X + \lambda x)^{\otimes (i-1)}$.
    So at $\lambda_0=0$ we get
    \begin{align*}
        &~~~~~ \restr{\frac{\partial}{\partial \lambda_0}}{\lambda_0=0}
        Q_k : B_{n+1,k}\Big( \Lambda_1(\lambda_0 X + \lambda x), \Lambda_2(\lambda_0 X + \lambda x), ... \Big) \\
        &= \sum_{\bh \in \NN^{\NN^*}} \alpha^{n+1,k}_{\bh} Q_k : 
        \sum_{i=1}^\infty h_i \, i \, \left( \Lambda_i : X \otimes (\lambda x)^{\otimes (i-1)} \right)
        \otimes \Lambda_i(\lambda x)^{\otimes (h_i-1)}
        \otimes \bigotimes_{j \neq i} \Lambda_j(\lambda x)^{\otimes h_j}
    \end{align*}
    which is homogeneous in $\lambda$ of order
    $(i-1) + i (h_i-1) + \sum_{j \neq i} j h_j = (\sum_j j h_j) - 1 = n$ (for any $i$), since the only $\bh$'s for which $\alpha^{n+1,k}_{\bh}>0$ satisfy $\sum_j j h_j = n+1$.
    So, applying $\restr{\frac{\partial^n}{\partial \lambda^n}}{\lambda=0}$,
    \begin{align*}
        & E_{n+1} : X \otimes x^{\otimes n} \\
        &= \frac{1}{n+1} \sum_{k=1}^{n+1} Q_k :
        \sum_{\bh \in \NN^{\NN^*}} \alpha^{n+1,k}_{\bh}
        \sum_{i=1}^\infty h_i \, i \, \left( \Lambda_i : X \otimes x^{\otimes (i-1)} \right)
        \otimes \Lambda_i(x)^{\otimes (h_i-1)}
        \otimes \bigotimes_{j \neq i} \Lambda_j(x)^{\otimes h_j}
    \end{align*}
    as announced.
\end{proof}

\begin{claim} \label{claim:apx_bell:compos:abstract:calc2}
    For any $n \geq 1$ and $x \in \XXX$,
    \begin{align*}
        & D \left[B_{n,k}(\Lambda_1(x), \Lambda_2(x), ...)\right] \\
        &\eqpermut \sum_{\bh \in \NN^{\NN^*}} \alpha^{n,k}_{\bh}
        \sum_{i=1}^\infty h_i \left( \Lambda_{i+1} : \dot{x} \otimes x^{\otimes i} + i\, \Lambda_i : dx \otimes x^{\otimes (i-1)} \right)
        \otimes \Lambda_i(x)^{\otimes (h_i-1)}
        \otimes \bigotimes_{j \neq i} \Lambda_j(x)^{\otimes h_j}.
    \end{align*}
\end{claim}

\begin{proof}
    Follows from the definition
    $B_{n,k}(\Lambda_1(x), \Lambda_2(x), ...)
    = \sum_{\bh \in \NN^{\NN^*}} \alpha^{n,k}_{\bh}
    \bigotimes_{i=1}^\infty \Lambda_i(x)^{\otimes h_i}$
    and the definition of $D$.
\end{proof}

\begin{claim} \label{claim:apx_bell:compos:abstract:calc3}
    For any $i$, denote by $\bmone_i \in \NN^{\NN^*}$ the sequence such that $\bmone_{in} = 1$ if $n=i$ and $0$ otherwise.
    For any $n, k \geq 1$ and sequence $\bh \in \NN^{\NN^*}$, we have
    \begin{equation*}
        \frac{h_1+1}{n+1} \,\alpha^{n+1,k}_{\bh+\bmone_1} = \alpha^{n,k-1}_{\bh}
        ~~~~\text{and}~~~~
        \frac{(h_{i+1}+1) \, (i+1)!}{h_i ~ i! ~ (n+1)} \, \alpha^{n+1,k}_{\bh+\bmone_{i+1}-\bmone_i} 
        = \alpha^{n,k}_{\bh}.
    \end{equation*}
\end{claim}

\begin{proof}
    The claim follows straightforwardly from the definition of the $\alpha^{n,k}_{\bh}$ (\autoref{def:apx_bell:coeffs}). For ease of reference:
    \begin{equation*}
        \alpha^{n,k}_{\bh}
        = \frac{n!}{h_1!~ h_2!~ h_3!~ ...}~ \frac{1}{(1!)^{h_1}\, (2!)^{h_2}\, (3!)^{h_3}\, ...}
    \end{equation*}
    if $\sum_{i=1}^\infty i\, h_i = n$ and $\sum_{i=1}^\infty h_i=k$, and $0$ otherwise.
\end{proof}

\begin{remark}
    Instead of direct computations, one could try to prove \autoref{prop:apx_bell:compos_abstract} by introducing a derivation over $T\left( \bigoplus_{k=1}^\infty \LLL_k \right)$, showing that \autoref{prop:apx_bell:abstract_vectorspaces} applies to it because of an identity analogous to \eqref{eq:apx_bell:FdB_scalar_gfx}, and proceeding as in the scalar case.
    However we do not in fact have a simple identity analogous to \eqref{eq:apx_bell:FdB_scalar_gfx} in this setting, because the $x$ variables are not scalar.
\end{remark}

\section{Proof of \autoref{thm:fdb_wasserstein:fdb_wasserstein}} \label{sec:apx_pf_fdbwass}

We start by a preliminary proposition.

\begin{proposition} \label{prop:apx_pf_fdbwass:partialtmut}
    For any transport couple $\partial_t \mu_t = -\nabla \cdot (\mu_t \Phi_t)$,
    for any smooth scalar field $\varphi: \RR^d \to \RR$ (independent of time),
    \begin{gather*}
        \forall n \geq 1,~
        (\cD^n \varphi)_t = \sum_{k=1}^n \nabla^k \varphi : B_{n,k}(\Phi_t, (\cD \Phi)_t, (\cD^2 \Phi)_t, ...) \\
        \text{and}~~~~
        \partial_t^n \mu_t = \sum_{k=1}^n (-\nabla)^k : \left[ \mu_t B_{n,k}(\Phi_t, (\cD \Phi)_t, (\cD^2 \Phi)_t, ...) \right]
    \end{gather*}
    where $B_{n,k}(\Phi_t, (\cD \Phi)_t, (\cD^2 \Phi)_t, ...)$ is the $k$-tensor field defined by interpreting the Bell polynomial $B_{n,k}$ as a polynomial over the ring of tensor fields equipped with the tensor product.
\end{proposition}

\begin{proof}
    The first part is exactly an application of \autoref{prop:apx_bell:abstract_vectorspaces} for $\Lambda_k(g_k) = \nabla^k \varphi : g_k$ and appropriate choices of $d, D, \XXX, \YYY$.
    Nonetheless we provide a self-contained proof by induction.
    For $n=1$, we have by definition $\cD \varphi = \nabla \varphi \cdot \Phi_t = \nabla \varphi \cdot B_{1,1}(\Phi_t)$.
    For general $n \geq 1$, writing $B_{n,k}$ for $B_{n,k}(\Phi, \cD \Phi, \cD^2 \Phi, ...)$, suppose $(\cD^n \varphi)_t = \sum_{k=1}^n \nabla^k \varphi : B_{n,k}$.
    By the lemma below, since $\nabla^k \varphi$ is symmetric,
    \begin{align*}
        (\cD^{n+1} \varphi)_t
        = \cD \left[ \sum_{k=1}^n \nabla^k \varphi : B_{n,k} \right]
        &= \sum_{k=1}^n \nabla^k \varphi : \cD \left[ B_{n,k} \right] + (\cD \nabla^k \varphi) : B_{n,k} \\
        &= \sum_{k=1}^n \nabla^k \varphi : \left( B_{n+1,k} - \Phi \otimes B_{n,k-1} \right)
        + \nabla^{k+1} \varphi : (\Phi \otimes B_{n,k}) \\
        &= \sum_{k=1}^n \nabla^k \varphi : B_{n+1,k} - 0 + \nabla^{n+1} \varphi : B_{n+1,n+1}
        = \sum_{k=1}^{n+1} \nabla^k \varphi : B_{n+1,k}
    \end{align*}
    by a telescoping sum, using that $B_{n+1,0} = 0$ and that $\Phi \otimes B_{n,n} = \Phi^{\otimes (n+1)} = B_{n+1,n+1}$.
    
    For the second part of the proposition, one can check by induction that $\int \varphi ~d(\partial_t^n \mu_t) = \frac{d^n}{dt^n} \int \varphi \,d\mu_t$, and so
    \begin{align*}
        \int \varphi ~d(\partial_t^n \mu_t)
        &= \frac{d^n}{dt^n} \int \varphi \,d\mu_t
        = \int (\cD^n \varphi)_t d\mu_t
        = \int \sum_{k=1}^n \nabla^k \varphi : B_{n,k}(\Phi_t, ...) \,d\mu_t \\
        &= \sum_{k=1}^n \int \varphi \,d\left( (-\nabla)^k : \left[ B_{n,k}(\Phi_t,...) \mu_t \right] \right)
        = \int \varphi \,d\left( \sum_{k=1}^n (-\nabla)^k : \left[ B_{n,k}(\Phi_t,...) \mu_t \right] \right).
    \end{align*}
    This concludes the proof since $\varphi$ was arbitrary.
\end{proof}

\begin{lemma}
    For any velocity field $(\Phi_t)_t$ and scalar field $\varphi: \RR^d \to \RR$,
    for any $n, k \geq 1$,
    \begin{equation*}
        \nabla^k \varphi : \cD \left[ B_{n,k}(\Phi, \cD \Phi, \cD^2 \Phi, ...) \right]
        = \nabla^k \varphi : \Big(
            B_{n+1,k}(\Phi, \cD \Phi, \cD^2 \Phi, ...) 
            - \Phi \otimes B_{n,k-1}(\Phi, \cD \Phi, \cD^2 \Phi, ...)
        \Big).
    \end{equation*}
\end{lemma}

\begin{proof}
    The result follows by the exact same manipulations as for \autoref{lm:apx_bell:deriv_B} (in fact it is exactly an application of that lemma for appropriate choices of $d$ and $\XXX$).
\end{proof}

We now turn to the proof of \autoref{thm:fdb_wasserstein:fdb_wasserstein}.
Throughout this section, fix a transport couple $\partial_t \mu_t = -\nabla \cdot (\mu_t \Phi_t)$ and a functional $\FFF: \PPP_2(\RR^d) \to \RR$,
and we adopt the shorthands introduced in the statements of \autoref{prop:fdb_wasserstein:wassdiffs:expr_firstvars} and \autoref{thm:fdb_wasserstein:fdb_wasserstein}.

\pagebreak

We saw in \autoref{lm:fdb_wasserstein:wassdiffs:meas_space_FdB} that
\begin{equation*}
    \frac{d^n}{dt^n} \FFF(\mu_t)
    = \sum_{k=1}^n \FFF^{(k)}[\mu_t] :_{\int} B_{n,k}\left( \partial_t \mu_t, \partial_t^2 \mu_t, \partial_t^3 \mu_t, ... \right).
\end{equation*}
On the other hand, we saw in \autoref{prop:apx_pf_fdbwass:partialtmut} that
\begin{align*}
    \partial_t^n \mu_t &= \sum_{k=1}^n (-\nabla)^k : \left[ \mu_t B_{n,k}(\Phi_t, (\cD \Phi)_t, (\cD^2 \Phi)_t, ...) \right] \\
    &= \sum_{k=1}^n L_k^{(\mu_t)} : B_{n,k}(\Phi_t, (\cD \Phi)_t, (\cD^2 \Phi)_t, ...)
\end{align*}
where $L_n^{(\mu)}$ is the symmetric $n$-linear operator $\left( L^2_\mu \cap \mathfrak{X}(\RR^d) \right)^{\otimes n} \to \MMM(\RR^d)$ given by $L_n^{(\mu)}(\Phi) = (-\nabla)^n : (\mu \Phi^{\otimes n})$.
Moreover, we saw in \autoref{prop:fdb_wasserstein:wassdiffs:expr_firstvars} that
\begin{align*}
    D^n_\mu \FFF(\Phi)
    &= \sum_{k=1}^n \FFF^{(k)}[\mu] :_{\int} B_{n,k}\left( -\nabla : [\mu \Phi], (-\nabla)^2 : [\mu \Phi^{\otimes 2}], (-\nabla)^3 : [\mu \Phi^{\otimes 3}], ... \right) \\
    &= \sum_{k=1}^n \FFF^{(k)}[\mu] :_{\int} B_{n,k}\left( L^{(\mu)}_1(\Phi), L^{(\mu)}_2(\Phi), L^{(\mu)}_3(\Phi), ... \right).
\end{align*}
This suggests applying the composition rule from \autoref{prop:apx_bell:compos_abstract} to
\begin{itemize}
    \item $\XXX \coloneqq \left[ L^2_{\mu_\bullet} \cap \mathfrak{X}(\RR^d) \right]^I$ the set of time-dependent smooth vector fields $(\Psi_t)_{t \in I}$ such that $\Psi_t \in L^2_{\mu_t}$ for all $t \in I$,
    and $\YYY \coloneqq \left[ \MMM(\RR^d) \right]^I$
    and $\ZZZ \coloneqq \RR^I$,
    \item $\Lambda_n \coloneqq \big( L_n^{(\mu_t)} \big)_{t \in I}$ 
    and $Q_n \coloneqq \big(\FFF^{(n)}[\mu_t] \big)_{t \in I}$ acting pointwise in time,
    so that $E_n = \sum_{k=1}^n Q_k : B_{n,k}(\Lambda_1, \Lambda_2, ...) = \left( D^n_{\mu_t} \FFF \right)_{t \in I}$ by \autoref{prop:fdb_wasserstein:wassdiffs:expr_firstvars} and \autoref{rk:apx_bell:compos:abstract:xBLambda},
    \item $d \coloneqq \cD$
    and $D \coloneqq \partial_t$
    and $\Delta \coloneqq \frac{d}{dt}$,
    \item $\dot{x} \coloneqq (\Phi_t)_t$, 
    so that $\dot{y} = \Lambda_1(\dot{x}) = \big( L_1^{(\mu_t)} \Phi_t \big)_t = \left( -\nabla \cdot (\mu_t \Phi_t) \right)_t = (\partial_t \mu_t)_t$.
\end{itemize}
Which these choices, the conclusion \eqref{eq:apx_bell:compos:abstract:prop_ccl} of \autoref{prop:apx_bell:compos_abstract} reads
\begin{gather*}
    \forall n \geq 1,~
    \Delta^{n-1} Q_1(\dot{y})
    = \sum_{k=1}^n ~\underbrace{
        \sum_{l=1}^k Q_l : B_{k,l}(\Lambda_1, \Lambda_2, ...) 
    }_{= E_k}
    \,: B_{n,k}(\dot{x}, d\dot{x}, d^2\dot{x}, ...), \\
    \text{i.e.,}~~~~
    \frac{d^{n-1}}{dt^{n-1}} \int \FFF^{(1)}[\mu_t] ~d\big( -\nabla \cdot (\mu_t \Phi_t) \big)
    = \frac{d^n}{dt^n} \FFF(\mu_t)
    = \sum_{k=1}^n D^k_{\mu_t} \FFF : B_{n,k}(\Phi, (\cD \Phi)_t, (\cD^2 \Phi)_t, ...),
\end{gather*}
which is precisely the desired identity to prove \autoref{thm:fdb_wasserstein:fdb_wasserstein}.

Thus it only remains to check that the conditions of \autoref{prop:apx_bell:compos_abstract} hold with the choices of $\XXX, \YYY, \ZZZ, (\Lambda_n)_n, (Q_n)_n, d, D, \Delta, \dot{x}$ above.
Namely, the condition linking $d, D, \dot{x}$ and $(\Lambda_n)_n$ reads
\begin{align*}
    \forall n \geq 1,~ \forall x \in \XXX,~
    D\left( \Lambda_n : x^{\otimes n} \right) &= \Lambda_{n+1} : \dot{x} \otimes x^{\otimes n} + \Lambda_n : d(x^{\otimes n}) ~~~~ \iff   \nonumber \\
    \forall (\Psi_t)_t \in \left[ L^2_{\mu_\bullet} \cap \mathfrak{X}(\RR^d) \right]^I\!,
    \partial_t\big[ (-\nabla)^n \!:\! (\mu_t \Psi_t^{\otimes n}) \big]
    &= (-\nabla)^{n+1} \!:\! \left( \mu_t \Phi_t \otimes \Psi_t^{\otimes n} \right)
    + (-\nabla)^n \!:\! \left( \mu_t \cD(\Psi_t^{\otimes n}) \right),
\end{align*}
which is easily verified using the properties of the convective derivatives listed in \autoref{subsec:fdb_wasserstein:convderiv}: we indeed have that for any smooth $\varphi: \RR^d \to \RR$,
\begin{align*}
    \frac{d}{dt} \int \nabla^n \varphi : \Psi_t^{\otimes n} d\mu_t 
    &= \int \left[ (\cD \nabla^n \varphi) : \Psi_t^{\otimes n} + \nabla^n \varphi : \cD(\Psi_t^{\otimes n}) \right] d\mu_t \\
    &= \int \left[ \nabla^{n+1} \varphi : \Phi_t \otimes \Psi_t^{\otimes n} + \nabla^n \varphi : \cD(\Psi_t^{\otimes n}) \right] d\mu_t \\
    &= \int \varphi ~d\left[ 
        (-\nabla)^{n+1} : \left( \mu_t \Phi_t \otimes \Psi_t^{\otimes n} \right)
        + (-\nabla)^n : \left( \mu_t \cD(\Psi_t^{\otimes n}) \right)
    \right].
\end{align*}
As for the condition linking $D, \Delta, \dot{y}$ and $(Q_n)_n$, it reads
\begin{align*}
    \forall n \geq 1,~ \forall y \in \YYY,~
    \Delta\left( Q_n : y^{\otimes n} \right) &= Q_{n+1} : \dot{y} \otimes y^{\otimes n} + Q_n : D(y^{\otimes n}) ~~~~ \iff \\
    \forall (s_t)_t \in \left[ \MMM(\RR^d) \right]^I,~
    \frac{d}{dt} \left( \FFF^{(n)}[\mu_t] :_{\int} s_t^{\otimes n} \right)
    &= \FFF^{(n+1)}[\mu_t] :_{\int} \left( \partial_t \mu_t \otimes s_t^{\otimes n} \right)
    + \FFF^{(n)}[\mu_t] :_{\int} \partial_t (s_t^{\otimes n}),
\end{align*}
and it is indeed true, simply by definition of the $n$-th variation of a functional over $\PPP(\RR^d)$.
This concludes the proof of \autoref{thm:fdb_wasserstein:fdb_wasserstein}.

\begin{remark}
    As a by-product of the proof, we have by \autoref{lm:apx_bell:compos:abstract:DeltaEk} that for any $(\Psi_t)_t \in \left[ L^2_{\mu_\bullet} \cap \mathfrak{X}(\RR^d) \right]^I$,
    \begin{equation*}
        \forall n \geq 1,~
        \frac{d}{dt} \left( D^n_{\mu_t} \FFF : \Psi_t^{\otimes n} \right)
        = D^{n+1}_{\mu_t} \FFF : \Phi_t \otimes \Psi_t^{\otimes n}
        + D^n_{\mu_t} \FFF : \cD \left[ \Psi_t^{\otimes n} \right].
    \end{equation*}
\end{remark}

\section{Further discussion: the convective derivative as a formal connection on~\texorpdfstring{$\PPP(\XXX)$}{P(X)}} \label{sec:apx_connection}

In this section we further the discussion of \autoref{subsec:fdb_wasserstein:discussion}.
Namely, we show that for a smooth manifold $\XXX$ equipped with a connection, the convective derivative behaves formally like a connection on $\PPP(\XXX)$, and we compute its torsion and Riemann curvature tensors.
Since our goal is merely to show a formal analogy, and since this discussion is not crucial to the main results of the paper, we will omit the justification of some of the claims.
Due to typesetting reasons, in this section only, we will use $\upphi$ and $\uppsi$ to denote vector fields, in place of $\Phi$ and $\Psi$ in the rest of the document.

\paragraph{Connections on smooth manifolds in finite dimension.}
Let us start by reviewing the structure of (non-Riemannian) differential manifolds in the classical finite-dimensional case.
\begin{enumerate}[label=($\roman*$).]
	\item Let $\XXX$ be a smooth manifold. Its tangent bundle is denoted by $T\XXX$. The set of smooth vector fields on $\XXX$ is denoted by $\mathfrak{X}(\XXX)$.
	\item 
	\cite[Definition~5.1]{boumal_introduction_2023} A connection is an operator
	$
		\nabla: T\XXX \times \mathfrak{X}(\XXX) \to T\XXX,~
		(u, V) \mapsto \nabla_u V
	$
	such that $\nabla_u V \in T_x \XXX$ whenever $u \in T_x \XXX$ and, for all $U, V, W \in \mathfrak{X}(\XXX)$, $u, w \in T_x \XXX$, $a, b \in \RR$ and scalar field $f: \XXX \to \RR$,
	\begin{enumerate}[label=\arabic*.]
		\item $(\nabla_U V)(x) \coloneqq \nabla_{U(x)} V$ defines a smooth vector field;
		\item $\nabla_{au+bw} V = a \nabla_u V + b \nabla_w V$;
		\item $\nabla_u (aV + bW) = a \nabla_u V + b \nabla_u W$;
		\item $\nabla_u (fV) = D_x f(u) ~ V(x) + f(x) \nabla_u V$
		where $D_x f$ denotes the differential of $f$ at $x$.
	\end{enumerate}
    With the notation $\nabla_u f \coloneqq D_x f(u)$, item~3.\ reads $\nabla_u (fV) = (\nabla_u f) V + f (\nabla_u V)$, i.e., a Leibniz product rule.
    This is consistent with the notion of total covariant derivative (see $(v)$).
    
	For $\Phi, \Psi \in \mathfrak{X}(\XXX)$, the smooth vector field $\nabla_\Phi \Psi$ may also be denoted by $\Phi \cdot \nabla \Psi$, or in index notation by $[\nabla_\Phi \Psi]^i = \Phi^j (\nabla_j \Psi^i)$.
	\item 
\begin{samepage}
	\cite[Theorem~5.29]{boumal_introduction_2023}
	Fix a smooth curve $c: I \to \XXX$ where $I$ is an interval of $\RR$, and denote $\mathfrak{X}(c) = \left\{ Z: I \to T\XXX ~\text{smooth and s.t.}~ \forall t \in I,~ Z(t) \in T_{c(t)} \XXX \right\}$.
	There exists a unique operator $\frac{D}{dt}: \mathfrak{X}(c) \to \mathfrak{X}(c)$, called the induced covariant derivative, such that for all $Y, Z \in \mathfrak{X}(c)$, $U \in \mathfrak{X}(\XXX)$, $a, b \in \RR$ and $g: I \to \RR$,
	\begin{enumerate}[label=\arabic*.]
		\item $\frac{D}{dt}(aY + bZ) = a \frac{D}{dt} Y + b \frac{D}{dt} Z$;
		\item $\frac{D}{dt}(gZ)(t) = g'(t) Z(t) + g(t) \frac{D}{dt} Z(t)$;
		\item $\frac{D}{dt}(U \circ c)(t) = \nabla_{c'(t)} U$
		where $c'(t) \in T_{c(t)} \XXX$ is the velocity of $c$.
	\end{enumerate}
	The intrinsic acceleration of $c$ is defined as $c''(t) \coloneqq \frac{D}{dt} c'(t)$.
	Geodesic curves are defined as the curves $c$ such that $c''(t) = 0$ for all $t$.
\end{samepage}
	\item \cite{lee_introduction_2018}
	The \emph{Lie bracket of two vector fields $U$ and $V$} is defined as the vector field $[U,V]$ such that, for any scalar field $f$,
	\begin{equation*}
		\nabla_{[U,V]} f = \nabla_U (\nabla_V f) - \nabla_V (\nabla_U f), 
		~~\text{i.e.,}~~
		(\nabla_i f) [U, V]^i = U^j \nabla_j \left( V^i \nabla_i f \right) - V^j \nabla_j \left( U^i \nabla_i f \right).
	\end{equation*}
	
	The torsion tensor is the $(1, 2)$-tensor $T$ such that,
    for any $U, V \in \mathfrak{X}(\XXX)$,
	\begin{equation*}
		T(U, V) = \nabla_U V - \nabla_V U - [U,V],
		~~\text{i.e.,}~~
		T^i_{~jk} U^j V^k = \left( U^j \nabla_j V^i - V^j \nabla_j U^i \right) - [U,V]^i.
	\end{equation*}
	We say $(\XXX, \nabla)$ is torsion-free if $T=0$. This condition is equivalent to the second-order differentials of scalar fields being symmetric.
	
	The Riemann curvature tensor is the $(1,3)$-tensor field $R$ such that, for any $U, V, W \in \mathfrak{X}(\XXX)$,
	\begin{equation*}
		V^j W^k R^{~~i~}_{jk~l} U^l
		= R(V, W) U = \nabla_V (\nabla_W U) - \nabla_W (\nabla_V U) - \nabla_{[V,W]} U.
	\end{equation*}
	\item The statement of \autoref{prop:fdb_wasserstein:findim:curved:def_totalcovariantderiv} holds without change for any smooth manifold $\XXX$ and any connection $\nabla$ (not necessarily a Riemannian manifold and its Levi-Civita connection).
	This defines the total covariant derivatives of tensor fields.
	
	The statement of \autoref{prop:fdb_wasserstein:findim:curved:asym} also holds without change, assuming that $(\XXX, \nabla)$ is torsion-free.
\end{enumerate}

\paragraph{The convective derivative as a connection on $\PPP(\XXX)$.}
Let $\XXX$ be a smooth manifold equipped with a connection $\nabla$.
We now describe the formal (non-Riemannian) differential structure of $\PPP(\XXX)$ equipped with the convective derivative.
\begin{enumerate}[label=($\roman*'$).]
	\item Consider the set $\PPP(\XXX)$, written $\PPP$ in this paragraph for the sake of concision.
	Define
	$T_\mu \PPP = \mathfrak{X}(\XXX) / \RRR$
	with the equivalence relation $\upphi \RRR \uppsi \iff \forall n, \restr{\nabla^n (\upphi - \uppsi)}{\support(\mu)} = 0$,
	let $T\PPP = \left\{ (\mu, \upphi), \upphi \in T_\mu \PPP \right\}$,
	and let us call \emph{vector field on $\PPP$} any mapping $\bm\Phi: \PPP \to T\PPP$ such that $\bm\Phi[\mu] \in T_\mu\PPP$ for all $\mu$.
	Denote by $\mathfrak{X}(\PPP)$ the set of vector fields $\bm\Psi$ on $\PPP$ that are ``smooth'' in a sense which we will not go into.
	\item Define an operator $\bm\nabla: T\PPP \times \mathfrak{X}(\PPP) \to T\PPP,~ (\upphi, \bm\Psi) \mapsto \bm\nabla_\upphi \bm\Psi$ as follows.
	For any $\upphi \in T_\mu\PPP$, consider any transport couple $\partial_t \mu_t = -\nabla \cdot (\mu_t \upphi_t)$ such that $\mu_0 = \mu$ and $\upphi_0 = \upphi$. Then $\bm\nabla_\upphi \bm\Psi \coloneqq \left(\cD (\bm\Psi[\mu_s])_s \right)_{t=0} \in T_\mu \PPP$.
    To see that this definition is independent of the choice of the transport couple under certain regularity assumptions, note that if $\mu \mapsto \bm\Psi[\mu](x)^i \in \RR$ admits a smooth first variation, denoted by $\frac{\delta \bm\Psi[\mu](x)^i}{\delta \mu}$, for any $i$ and $x \in \XXX$, then
    \begin{equation*}
        \left( \cD (\bm\Psi[\mu_s])_s \right)_t(x)^i
        = \int_{\XXX} d\mu_t(x') \left( \upphi_t \cdot \nabla \frac{\delta \bm\Psi[\mu_t](x)^i}{\delta \mu_t} \right)(x')
        + \left( \upphi_t \cdot \nabla \bm\Psi[\mu_t]^i \right)(x).
    \end{equation*}

	Note that when $\XXX = \RR^d$, the (Otto calculus analog of) Levi-Civita connection corresponds to $\bm\nabla$ composed with the orthogonal projection, in $L^2_\mu$, onto the set of gradient fields \cite[Definition~5.1]{gigli_second_2012}.

    Let us call \emph{scalar field on $\PPP$} any functional $\FFF: \PPP \to \RR$ and denote $\bm\nabla_\upphi \FFF \coloneqq D^1_\mu \FFF(\upphi), {\forall \upphi \in T_\mu \PPP}$.
	\item Fix a transport couple $(\mu_t, \upphi_t)_{t \in I}$ and denote $\mathfrak{X}((\mu_t, \upphi_t)_t) = \left\{ (\uppsi_t)_{t \in I} ~\text{s.t.}~ \forall t, \uppsi_t \in T_{\mu_t}\PPP \right\}$.
	One can check that the convective derivative $\cD$, as defined in \autoref{def:fdb_wasserstein:convderiv:convderiv}, maps the set $\mathfrak{X}((\mu_t, \upphi_t)_t)$ to itself, and \autoref{subsec:fdb_wasserstein:convderiv} show that $\cD$ satisfies properties analogous to the induced covariant derivative ($\frac{D}{dt}$) in the finite-dimensional case.
	\item We compute the torsion and the Riemann curvature tensor of $(\PPP, \bm\nabla)$ in the next subsections.
    In particular, if $(\XXX, \nabla)$ is torsion-free, resp.\ has zero curvature, then $(\PPP, \bm\nabla)$ does too.
	\item When $\XXX$ is torsion-free and has zero curvature, the higher-order Wasserstein differentials of a functional $\FFF$---as defined as in \autoref{def:fdb_wasserstein:wassdiffs:wassdiffs}---are precisely the total covariant derivatives of $\FFF$ viewed as a scalar field over $\PPP$.
    When $\XXX$ has non-zero torsion or non-zero curvature, however, the latter objects are asymmetric, and the former objects are equal to the symmetric parts of the latter.
\end{enumerate}

Let us end our discussion by two remarks.
\begin{itemize}
    \item 
    There is another natural candidate for the notion of scalar fields on $\PPP$, namely, measure-dependent scalar fields $\bm\lambda: \PPP \to \left\{(\mu, \lambda), \lambda: \XXX \to \RR ~\text{smooth} \right\}$.
    Note that this second notion contains the first one as a special case by setting $\forall x \in \XXX, \bm\lambda[\mu](x) = \FFF(\mu)$.
    It is natural to define $\bm\nabla_\upphi \bm\lambda$ as $\left(\cD (\bm\lambda[\mu_s])_s \right)_{t=0}$ for any $\upphi \in T_\mu \PPP$, for any transport couple $(\mu_t, \upphi_t)_t$ such that $\mu_0 = \mu$ and $\upphi_0 = \upphi$.
    We did not investigate this possibility.
    \item
    Define $\tT_\mu \PPP = T_\mu \PPP / \sim_\mu$,
    where ``$\sim_\mu$'' is the equivalence relation
    $
        \upphi \sim_\mu \uppsi \iff \nabla \cdot (\mu \upphi) = \nabla \cdot (\mu \uppsi)
    $,
    and let $\tT\PPP = \left\{ (\mu, \tilde{\upphi}), \tilde{\upphi} \in \tT_\mu \PPP \right\}$.
    Then one can show, under certain regularity assumptions, that $T_\mu \PPP$ can be replaced by $\tT_\mu \PPP$ everywhere, i.e., the $\upphi$ only come into play in the definitions above via $\nabla \cdot (\mu \upphi)$.
\end{itemize}

\subsection{Computation of the torsion tensor of \texorpdfstring{$(\PPP(\XXX), \bm{\nabla})$}{(P(X),bmnabla)}}

We continue to write $\PPP$ for $\PPP(\XXX)$ for brevity. Consider a scalar field $\FFF$ on $\PPP$, i.e., a smooth functional $\FFF: \PPP \to \RR$, and let $\bm\Phi, \bm\Psi \in \mathfrak{X}(\PPP)$. Then, at any $\mu \in \PPP$,
\begin{align*}
    \big[ \bm\nabla_{\bm\Phi} (\bm\nabla_{\bm\Psi} \FFF) \big](\mu)
    &= \bm\nabla_{\upphi} \left( D^1 \FFF(\bm\Psi) \right)
    \hspace{20em} \text{where}~~ \upphi = \bm\Phi[\mu] \\
    &= \bm\nabla_{\upphi} \left(
        \int_\XXX \bm\Psi[\bullet] \cdot \nabla \FFF'[\bullet] d(\bullet)
    \right) \\
    &= \left( \cD \left( \int_\XXX \bm\Psi[\mu_s] \cdot \nabla \FFF'[\mu_s] d\mu_s \right)_s \right)_{t = 0} 
    \text{where}~ \partial_t \mu_t = -\nabla \cdot (\mu_t \upphi_t), \mu_0 = \mu, \upphi_0 = \upphi \\
    &= \int_\XXX \left( \cD \left(\bm\Psi[\mu_s]\right)_s \right)_{t=0} \cdot \nabla \FFF'[\mu] d\mu
    + \int_\XXX \bm\Psi[\mu] \cdot \left(\cD \left(\nabla\FFF'[\mu_s]\right)_s \right)_{t=0} d\mu \\
    &= \int_\XXX (\bm\nabla_\upphi \bm\Psi) \cdot \nabla \FFF'[\mu] d\mu
    + \int_\XXX \uppsi^i \left( \bm\nabla_\upphi \nabla_i \FFF' \right) d\mu
    \hspace{7em} \text{where}~~ \uppsi = \bm\Psi[\mu]
\end{align*}
by definition, where in the fourth line we used (extensions to $\XXX$ instead of $\RR^d$ of) \autoref{lm:fdb_wasserstein:convderiv:ddt_lambdat_mut} and \autoref{lm:fdb_wasserstein:convderiv:product_rule}.
Now
for any fixed $\xi \in \mathfrak{X}(\XXX)$,
\begin{align*}
    \xi^i \left( \cD \left(\nabla_i\FFF'[\mu_s]\right)_s \right)_t
    &= \xi^i \nabla_i \partial_t \FFF'[\mu_t]
    + \xi^i \left( \upphi_t \cdot \nabla \nabla_i \FFF'[\mu_t] \right) \\
    &= \nabla_\xi \left( \int_\XXX d\mu_t(x') \upphi_t(x') \cdot \nabla \FFF''[\mu_t](\bullet, x') \right)
    + \xi^i \left( \nabla_{\upphi_t} \nabla_i \FFF'[\mu_t] \right) \\
    &= \int_\XXX d\mu_t(x') \nabla_{1,\xi} \nabla_{2,\upphi_t(x')} \FFF''[\mu_t]
    + \nabla_{\upphi_t} (\nabla_\xi \FFF'[\mu_t])
    - (\nabla_{\upphi_t} \xi) \cdot \nabla \FFF'[\mu_t]
\end{align*}
where $\nabla_{n,\xi}$ denotes derivative w.r.t.\ to the $n$-th variable in the direction $\xi$,
and where for the second term we used that $\xi \otimes \nabla_{\upphi} \nabla \FFF'[\mu_t] = \nabla_{\upphi} \left( \xi \otimes \nabla \FFF'[\mu_t] \right) - (\nabla_{\upphi} \xi) \otimes (\nabla \FFF'[\mu_t])$ by \autoref{prop:fdb_wasserstein:findim:curved:def_totalcovariantderiv}.
So
\begin{align*}
    \int_\XXX \uppsi^i \left( \bm\nabla_\upphi \nabla_i \FFF' \right) d\mu
    &= \int_\XXX d\mu(x) \int_\XXX d\mu(x')
    \nabla_{1,\uppsi(x)} \nabla_{2,\upphi(x')} \FFF''[\mu] \\
    &~~ + \int_\XXX \Big(
        \nabla_\upphi (\nabla_\uppsi \FFF'[\mu])
        - (\nabla_\upphi \uppsi) \cdot \nabla \FFF'[\mu]
    \Big) d\mu,
\end{align*}
and so by the symmetric computations with $\bm\Phi, \bm\Psi$ swapped,
\begin{multline*}
    \big[ \bm\nabla_{\bm\Phi} (\bm\nabla_{\bm\Psi} \FFF)
    - \bm\nabla_{\bm\Psi} (\bm\nabla_{\bm\Phi} \FFF) \big](\mu)
    = \int_\XXX \left( \bm\nabla_\upphi \bm\Psi - \bm\nabla_\uppsi \bm\Phi \right) \cdot \nabla \FFF'[\mu] d\mu 
    ~+ 0 \\
    + \int_\XXX \Big(
        \nabla_\upphi (\nabla_\uppsi \FFF'[\mu]) - \nabla_\uppsi (\nabla_\upphi \FFF'[\mu]) - \left( \nabla_\upphi \uppsi - \nabla_\uppsi \upphi \right) \cdot \nabla \FFF'[\mu]
    \Big) d\mu
\end{multline*}
since $\FFF''[\mu]: \XXX \times \XXX \to \RR$ is symmetric.
For the first term, note that it is equal by definition to $D^1_\mu \FFF(\bm\nabla_\upphi \bm\Psi - \bm\nabla_\uppsi \bm\Phi) = \big[ \bm\nabla_{\left( \bm\nabla_\Phi \bm\Psi - \bm\nabla_\Psi \bm\Phi \right)} \FFF \big](\mu)$.
For the last term, note that by definition,
\begin{equation*}
    \nabla_\upphi (\nabla_\uppsi f) - \nabla_\uppsi (\nabla_\upphi f) - \left( \nabla_\upphi \uppsi - \nabla_\uppsi \upphi \right) \cdot \nabla f 
    = [\upphi, \uppsi] \cdot \nabla f - \left( \nabla_\upphi \uppsi - \nabla_\uppsi \upphi \right) \cdot \nabla f 
    = -T(\upphi, \uppsi) \cdot \nabla f
\end{equation*}
for any $f: \XXX \to \RR$.
Thus, we have
\begin{equation*}
    \big[ \bm\nabla_{\bm\Phi} (\bm\nabla_{\bm\Psi} \FFF)
    - \bm\nabla_{\bm\Psi} (\bm\nabla_{\bm\Phi} \FFF)
    - \bm\nabla_{\left( \bm\nabla_\Phi \bm\Psi - \bm\nabla_\Psi \bm\Phi \right)} \FFF \big](\mu)
    = \int_\XXX -T(\upphi, \uppsi) \cdot \nabla \FFF'[\mu] d\mu
    = -\bm\nabla_{T(\upphi,\uppsi)} \FFF.
\end{equation*}
In other words, since $\FFF$ was arbitrary, the torsion tensor $\bm T$ of $(\PPP, \bm\nabla)$ is given by, for any $\bm\Phi, \bm\Psi \in \mathfrak{X}(\PPP)$, 
\begin{equation*}
    \forall \mu \in \PPP,~
    \bm T(\bm\Phi, \bm\Psi)[\mu] = T(\bm\Phi[\mu], \bm\Psi[\mu]) \in T_\mu \PPP.
\end{equation*}

\subsection{Computation of the Riemann curvature tensor of \texorpdfstring{$(\PPP(\XXX), \bm{\nabla})$}{(P(X),bmnabla)}}

In this subsection we assume for simplicity that $(\XXX, \nabla)$ is torsion-free.
We continue to write $\PPP$ for $\PPP(\XXX)$ for brevity. 
Furthermore, for any $\nu \in \PPP$ and $y \in \XXX$, we will write $\frac{\delta \mathsf{expr}(\nu)}{\delta \nu(y)}$ for the first variation of $\mathsf{expr}(\nu)$ w.r.t.\ $\nu$ at $y$,
and for any $\bm\Xi \in \mathfrak{X}(\PPP)$ and $x_0, y, z \in \XXX$, we will write $\bm\Xi'[\nu]^i(x_0, y) = \frac{\delta \bm\Xi[\nu]^i(x_0)}{\delta \nu(y)} \in T_{x_0} \XXX$
and $\bm\Xi''[\nu]^i(x_0,y,z) = \bm\Xi''[\nu]^i(x_0,z,y) = \frac{\delta^2 \bm\Xi[\nu]^i(x_0)}{\delta \nu(y) \delta \nu(z)} \in T_{x_0} \XXX$.

Let us compute the Riemann curvature tensor $\bm R$ of $(\PPP, \bm\nabla)$, defined by, for any $\bm\Phi, \bm\Psi, \bm\Xi \in \mathfrak{X}(\PPP)$,
\begin{equation*}
    \bm R(\bm\Phi, \bm\Psi) \bm\Xi
    = \bm\nabla_{\bm\Phi} (\bm\nabla_{\bm\Psi} \bm\Xi)
    - \bm\nabla_{\bm\Psi} (\bm\nabla_{\bm\Phi} \bm\Xi)
    - \bm\nabla_{[\bm\Phi, \bm\Psi]} \bm\Xi
\end{equation*}
where $[\bm\Phi, \bm\Psi] = \bm\nabla_{\bm\Phi} \bm\Psi - \bm\nabla_{\bm\Psi} \bm\Phi$ since $(\PPP, \bm\nabla)$ is torsion-free since $(\XXX, \nabla)$ is.
Fix $\bm\Phi, \bm\Psi, \bm\Xi \in \mathfrak{X}(\PPP)$ and $\mu \in \PPP$ and denote $\upphi = \bm\Phi[\mu], \uppsi = \bm\Psi[\mu], \xi = \bm\Xi[\mu]$ and $\xi'(\bullet, y) = \bm\Xi'[\mu](\bullet, y) \in T_\mu \PPP$ for any $y \in \XXX$. Then
\begin{equation} \label{eq:apx_connection:comput_curv:bmnablaPsiXi}
    \forall \nu \in \PPP,~~~~
    (\bm\nabla_{\bm\Psi} \bm\Xi)[\nu]^i
    = \int_\XXX \bm\Psi[\nu](y) \cdot \nabla_y \bm\Xi'[\nu]^i(\bullet, y) ~d\nu(y)
    + \bm\Psi[\nu] \cdot \nabla \bm\Xi[\nu]^i.
\end{equation}
So for any transport couple $\partial_t \mu_t = -\nabla \cdot (\mu_t \upphi_t)$ with $\mu_0 = \mu$ and $\upphi_0 = \upphi$, by (extensions to $\XXX$ instead of $\RR^d$ of) \autoref{lm:fdb_wasserstein:convderiv:ddt_lambdat_mut} and \autoref{lm:fdb_wasserstein:convderiv:product_rule},
\begin{align*}
    & \left( \bm\nabla_{\bm\Phi} (\bm\nabla_{\bm\Psi} \bm\Xi) \right)[\mu]^i
    = \bm\nabla_\upphi (\bm\nabla_{\bm\Psi} \bm\Xi)^i \\
    &= \int_\XXX \left( \cD (\bm\Psi[\mu_s])_s \right)_{t=0}(y) \cdot \nabla_y \xi'(\bullet, y)^i ~d\mu(y)
    + \int_\XXX \uppsi(y) \cdot \left( \cD (\nabla_y\bm\Xi'[\mu_s]^i(\bullet,y))_s \right)_{t=0} ~d\mu(y) \\
    &~~ + \left( \cD (\bm\Psi[\mu_s])_s \right)_{t=0} \cdot \nabla \xi^i
    + \uppsi \cdot \left( \cD (\nabla \bm\Xi[\mu_s]^i)_s \right)_{t=0}.
\end{align*}
Now for the second term, by direct computations,
\begin{align*}
    \left( \cD (\nabla_{y,j} \bm\Xi'[\mu_s]^i(\bullet, y))_s \right)_t
    &= \nabla_{y,j} \frac{d}{dt} \bm\Xi'[\mu_t]^i(\bullet, y)
    + \upphi_t^k \nabla_{\bullet,k} \left( \nabla_{y,j} \bm\Xi'[\mu_t]^i(\bullet, y) \right) \\
    &= \nabla_{y,j} \int_\XXX \upphi_t(z)^k \nabla_{z,k} \bm\Xi''[\mu_t]^i(\bullet, y, z) d\mu_t(z)
    + \nabla_{\bullet, \upphi_t} \nabla_{y,j} \bm\Xi'[\mu_t]^i(\bullet, y) \\
    \text{so}~\,
    \left( \cD (\nabla_{y,j} \bm\Xi'[\mu_s]^i(\bullet, y))_s \right)_{t=0}
    &= \nabla_{y,j} \int_\XXX \upphi(z)^k \nabla_{z,k} \xi''(\bullet,y,z) d\mu(z) + \nabla_{\bullet,\upphi} \nabla_{y,j} \xi'(\bullet,y)^i
\end{align*}
where $\xi''(\bullet, y, z) = \bm\Xi''[\mu](\bullet, y, z) \in T_\mu \PPP$.
For the fourth term, by the formal computation of \autoref{rk:fdb_wasserstein:convderiv:convderiv_nabla_curved},
\begin{align*}
    \left( \cD (\nabla_j \bm\Xi[\mu_s]^i)_s \right)_t
    &= \nabla_j (\cD (\bm\Xi[\mu_s]^i)_s)_t
    - (\nabla_j \upphi_t^k) (\nabla_k \bm\Xi[\mu_t]^i)
    - \upphi_t^k R^{~~i~}_{jk~l} \bm\Xi[\mu_t]^l \\
    \text{so}~~~~
    \left( \cD (\nabla_j \bm\Xi[\mu_s]^i)_s \right)_{t=0}
    &= \nabla_j \left( \bm\nabla_\upphi \bm\Xi \right)^i
    - (\nabla_j \upphi^k) (\nabla_k \xi^i)
    - \upphi^k R^{~~i~}_{jk~l} \xi^l.
\end{align*}
So
\begin{align*}
    \left( \bm\nabla_{\bm\Phi} (\bm\nabla_{\bm\Psi} \bm\Xi) \right)[\mu]^i 
    &= \int_\XXX \left( \bm\nabla_\upphi \bm\Psi \right)(y) \cdot \nabla_y \xi'(\bullet, y)^i ~d\mu(y) 
    + \left( \bm\nabla_\upphi \bm\Psi \right) \cdot \nabla \xi^i \\
    &~~ + \int_\XXX \uppsi^j(y) \left( 
        \nabla_{y,j} \int_\XXX \upphi(z)^k \nabla_{z,k} \xi''(\bullet,y,z) d\mu(z) + \nabla_{\bullet,\upphi} \nabla_{y,j} \xi'(\bullet,y)^i
    \right) d\mu(y) \\
    &~~ + \uppsi^j \left(
        \nabla_j \left( \bm\nabla_\upphi \bm\Xi \right)^i
        - (\nabla_j \upphi^k) (\nabla_k \xi^i)
        - \upphi^k R^{~~i~}_{jk~l} \xi^l
    \right) \\
    &= \int_\XXX \left( \bm\nabla_\upphi \bm\Psi \right)(y) \cdot \nabla_y \xi'(\bullet, y)^i ~d\mu(y)
    + \left( \bm\nabla_\upphi \bm\Psi \right) \cdot \nabla \xi^i \\
    &~~ + \int_\XXX d\mu(y) \int_\XXX d\mu(z) \nabla_{y,\uppsi} \nabla_{z,\upphi} \xi''(\bullet, y, z)
    + \int_\XXX \nabla_{\bullet,\upphi} \nabla_{y,\uppsi} \xi'(\bullet,y)^i d\mu(y) \\
    &~~ + \nabla_{\bullet,\uppsi} \left(
        \int_\XXX \upphi(y) \cdot \nabla_y \xi'(\bullet, y)^i d\mu(y) 
        + \upphi \cdot \nabla \xi^i
    \right)
    - (\nabla_\uppsi \upphi) \cdot \nabla \xi^i
    - \left[ R(\uppsi, \upphi) \xi \right]^i \\
    &= \left( \bm\nabla_{(\bm\nabla_{\bm\Phi} \bm\Psi)} \bm\Xi \right)[\mu]^i \\
    &~~ + \int_\XXX d\mu(y) \int_\XXX d\mu(z) \nabla_{y,\uppsi} \nabla_{z,\upphi} \xi''(\bullet, y, z) \\
    &~~ + \int_\XXX \nabla_{\bullet,\upphi} \nabla_{y,\uppsi} \xi'(\bullet,y)^i d\mu(y) 
    + \int_\XXX \nabla_{\bullet,\uppsi} \nabla_{y,\upphi} \xi'(\bullet,y)^i d\mu(y) \\
    &~~ + \nabla_\uppsi (\nabla_\upphi \xi)^i - (\nabla_\uppsi \upphi) \cdot \nabla \xi^i
    + \left[ R(\upphi, \uppsi) \xi \right]^i.
\end{align*}
Thus, by the symmetric computations with $\bm\Phi, \bm\Psi$ swapped,
\begin{align*}
    \!\Big( 
        \bm\nabla_{\bm\Phi} (\bm\nabla_{\bm\Psi} \bm\Xi) 
        - \bm\nabla_{\bm\Psi} (\bm\nabla_{\bm\Phi} \bm\Xi) 
    \Big)[\mu]^i 
    &= 
    \left( \bm\nabla_{\bm\Phi} (\bm\nabla_{\bm\Psi} \bm\Xi) \right)[\mu]^i 
    - \left( \bm\nabla_{\bm\Psi} (\bm\nabla_{\bm\Phi} \bm\Xi) \right)[\mu]^i \\
    &= \left( \bm\nabla_{(\bm\nabla_{\bm\Phi} \bm\Psi - \bm\nabla_{\bm\Psi} \bm\Phi)} \bm\Xi \right)[\mu]^i
    + 0 + 0 \\
    & ~\,  + \nabla_\uppsi (\nabla_\upphi \xi)^i - \nabla_\upphi (\nabla_\uppsi \xi)^i
    - (\nabla_\uppsi \upphi - \nabla_\upphi \uppsi) \cdot \nabla \xi^i
    + 2 \left[ R(\upphi, \uppsi) \xi \right]^i \\
    &= \left( \bm\nabla_{(\bm\nabla_{\bm\Phi} \bm\Psi - \bm\nabla_{\bm\Psi} \bm\Phi)} \bm\Xi \right)[\mu]^i
    + \left[ R(\upphi, \uppsi) \xi \right]^i.
\end{align*}
In other words,
the curvature tensor $\bm R$ of $(\PPP, \bm\nabla)$ is given by, for any $\bm\Phi, \bm\Psi, \bm\Xi \in \mathfrak{X}(\PPP)$,
\begin{equation*}
    \forall \mu \in \PPP,~
    \Big( \bm R(\bm\Phi, \bm\Psi) \bm\Xi \Big)[\mu]^i
    = \Big[ R\big( \bm\Phi[\mu], \bm\Psi[\mu] \big)\, \bm\Xi[\mu] \Big]^i
\end{equation*}
assuming $\XXX$ is torsion-free.

\fi

\end{document}